\documentclass[1p,11pt]{elsarticle}
\journal{Computer Science Review}
\biboptions{longnamesfirst,compress,sort}

\usepackage[english]{babel}
\usepackage[utf8]{inputenc}

\usepackage{amsmath,amsthm,amssymb,a4,graphicx}
\usepackage{hyperref,color,amsthm}

\newtheorem{property}{Property}[section]
\newtheorem{thm}[property]{Theorem}

\newtheorem{proposition}[property]{Proposition}
\newtheorem{theorem}[property]{Theorem}
\newtheorem{lemma}[property]{Lemma}
\newtheorem{corollary}[property]{Corollary}

\theoremstyle{definition}

\newtheorem{definition}[property]{Definition}

\newcommand{\mids}{{\bf mid}}
\renewcommand{\O}{O}

\newcommand{\clw}{{\bf clw}}
\newcommand{\rw}{{\bf rw}}

\DeclareMathOperator{\operatorClassNP}{\sf NP}
\newcommand{\classNP}{\ensuremath{\operatorClassNP}}

\newcommand{\bw}{{\bf bw}}
\newcommand{\fvs}{{\bf fvs}}
\newcommand{\vc}{{\bf vc}}

\newcommand{\obs}{{\bf obs}}
\newcommand{\tw}{{\bf tw}}
\newcommand{\cw}{{\bf cw}}
\newcommand{\FPT}{{\sf FPT}}

\begin{document}

\begin{frontmatter}

\title{Confronting Intractability via Parameters\tnoteref{fn1}
\\ \bigskip
{\em\small Happy 60th Birthday Mike!}
}

\tnotetext[fn1]{We would like to  dedicate this paper to the 60th birthday of {\sf Michael R. Fellows}. He is the \emph{alma mater} of the area of parameterized complexity.  The core ideas of this field  germinated from 
work of Langston and Fellows in the late 80's (such as 
\cite{FellowsL88nonc,FellowsL89,FellowsL94})
and Abrahamson, Fellows,  Ellis, and Mata \cite{AbrahamsonFEM89}, and then from a serendipitous meeting of Downey with Fellows in December 1990. But there is {\em no doubt} much of its current vigor comes from Mike's vision and irrepressible energy.}


 \author[label1]{Rodney G. Downey\corref{cor1}}
 \ead{Rod.Downey@msor.vuw.ac.nz}
  \author[label2]{Dimitrios M. Thilikos\corref{cor2}}\ead{sedthilk@math.uoa.gr}
  \address[label1]{School of Mathematics, Statistics  and Operations Research, Victoria University, New
Zealand.  Supported by the Marsden Fund of New Zealand.}
      \address[label2]{Department of Mathematics, National \& Kapodistrian University of Athens, 
Panepistimioupolis, GR-15784,  Athens, Greece.}
\cortext[cor1]{Supported by the Marsden Fund of New Zealand.}
\cortext[cor2]{Supported by the project  ``Kapodistrias'' (A$\Pi$ 02839/28.07.2008) of the National and Kapodistrian University of Athens (project code: 70/4/8757).}

\begin{abstract}One approach to confronting computational 
hardness is to  try to understand the contribution of 
various parameters to the running time of algorithms and the complexity of 
computational tasks. Almost no computational
tasks in real life are specified by their size alone.
It is not hard to imagine that some parameters contribute 
more intractability than others and it seems reasonable  
to develop a theory of computational complexity which seeks to exploit
this fact. 
Such a theory should be able to address the needs of 
practicioners in algorithmics.
The last twenty years have seen the development of 
such a theory. This theory has a large number of successes 
in terms of a rich collection of algorithmic techniques 
both practical and theoretical, and a fine-grained intractability theory.
Whilst the  theory has been widely used in a number of areas of applications
including computational biology, linguistics, VLSI design, learning theory and many
others, knowledge of the area is highly varied. We hope that this article will
show both the basic theory and point at the wide array
of techniques available. 
Naturally the treatment is condensed, and the 
reader who wants more should go to the texts of  Downey and Fellows
\cite{DowneyF99para},  Flum and Grohe \cite{FlumG06para},
Niedermeier \cite{Niedermeier06invi}, and the upcoming 
undergraduate text Downey and Fellows \cite{DowneyF12fund}.
\end{abstract}

\begin{keyword}
Parameterized Complexity, Parameterized Algorithms
\end{keyword}
\end{frontmatter}

\newpage 
\setcounter{tocdepth}{2}
\tableofcontents

\newpage 
\section{Introduction}
\subsection{Preamble}

There is little question that the computer has caused a profound change 
in society. At our fingertips are devices that are performing 
billions of computations a second, and the same is true of embedded 
devices in all manner of things. In the same way that much of mathematics 
was developed to understand the physics of the world around us, 
developing mathematics to understand computation is an imperative. 
It is little wonder that the famous ${\sf P}^{\mbox{?}\atop \mbox{$_{\mbox =}$}}{\sf NP}$ question is 
one of the Clay Prize questions and is regarded by some 
as the single most important problem in mathematics.
This question is one in {\sl complexity theory} which seeks to understand 
the  {\sl resources} (such as time or space) needed for computation.

The goal of this paper is to describe the methods and development of an
area  of complexity theory called \emph{parameterized complexity}. 
This is a {\sl fine-grained} complexity analysis which is often more attuned 
to analyzing computational questions which arise {\sl in practice} than traditional 
worst case analysis.
The idea here is that if you are working in some computational science 
and wish to understand what is \emph{\sl feasible} in your area, then 
perhaps this is the methodology you should know.

As articulated by the first author in~\cite{Downey03para}, people working in the area 
tend to be a bit schizophrenic in that, even in computer science, 
paper reviews range from saying that
``parameterized complexity is now well known so why are you
including this introductory stuff?'', to
the other extreme where reviewers say that they have never heard
about it.

Whilst the use of parameterized complexity in applications has been growing at a very fast rate, there still remain significant
groups who seem unaware of the techniques and ideas of the area. Everyone, apparently, is aware of NP-completeness
and seems vaguely aware of randomized or approximation techniques. However, this seems not to be the case yet
for parameterized complexity and/or algorithms. Moreover, it is certainly not the case that parameterized complexity has
become part of the standard curriculum of theoretical computer science as something all researchers in complexity should know.

%
%

\subsection{The issue}

In the first book on the subject by Downey and Fellows \cite{DowneyF99para}, much is made of the 
issue ``what is the role of complexity theory?''. 
After the work of Alan Turing and others in the mid-20th century, we understand 
what it means for something to be computable.
But for \emph{\sl actual} computations, 
it is not enough to know that the problem is computable 
in theory. If we need to compute something it is not good
if the running time will be so large that no computer ever conceived 
will be able to run the computation in the life time of the universe. Hence,
the \emph{\sl goal}  is to understand 
what is \emph{\sl feasible}. 
The point is that for \emph{\sl any}
combinatorial problem, we will deal with only a small finite fraction of
the universe, so how do we quantify this notion of 
feasibility?
Classically, this is achieved by identifying \emph{\sl feasibility}
with \emph{\sl polynomial-time computability.}

When we first learn of this, many of us immediately think
of polynomial like $n^{10^{5000}}$ and say ``surely 
that is not feasible!''.
 We tend to forget that,
because of such intuitive objections, the original suggestions that
asymptotic analysis was a
reasonable way to measure
complexity and
polynomial time was a reasonable measure of feasibility, were
initially
quite controversial. The reasons that
the now central ideas of asymptotic polynomial time ({\sf P}-time)
and {\sf NP}-completeness have survived the test of time
are basically two.
Firstly, the methodologies associated with polynomial-time
and {\sf P}-time reductions have proved to be readily amenable to mathematical
analysis in most situations.  In a word, {\sf P} has good closure
properties.  It offers a readily negotiable
mathematical currency.  Secondly, and even more importantly,
the ideas seem to work \emph{\sl in practice}:
for ``natural'' problems
the universe seems \emph{\sl kind} in the sense that
(at least historically)
if a ``natural'' computational problem is in {\sf P} then
usually we can find an algorithm having
a polynomial running time with small degree and small constants.
This last point, while
seeming a religious, rather than a mathematical, statement,
seems the key driving force behind the universal use of {\sf P}
as a classification tool for ``natural'' problems.

Even granting this, the use of things like {\sf NP}-completeness is a 
very coarse tool. We are given some problem and show that it is 
{\sf NP}-complete. What to do? As pointed out in Garey and Johnson \cite{GareyJ79comp},
this is only an initial foray. All such a result says is that 
our hope for an exact algorithm for the general problem 
which is feasible is likely in vain.

The problem with the standard approach of showing {\sf NP}-completeness is that
it offers no  methodology of seeking practical and feasible 
algorithms for 
more restricted versions of the problem which might be both relevant 
and feasible. As we will see, parameterized complexity seeks to 
have a conversation with the problem which enables us 
to do just that.

\subsection{The idea} The idea behind parameterized complexity is that 
we should  look more deeply into the actual \emph{\sl structure} of the problem in order 
to seek some kind of hidden feasibility. Classical complexity 
views a problem as an instance and a question. The running time is 
specified by the input's size.\\

\noindent{\sl Question :} When will the input of a problem coming from ``real life''
have no more structure than its size?

\noindent{\sl Answer :} \emph{Never!}\\

For real-world computations we \emph{\sl always} know more about the problem. 
The problem is planar, the problem has small width, the problem only concerns small 
values of the parameters. Cities are regular, objects people construct tend to be 
built in some comprehensible way from smaller ones constructed earlier, people's 
minds do not tend to work with more than a few alternations of quantifiers\footnote{As anyone who
has taught analysis will know!}, etc. Thus, \emph{\sl why not have a complexity theory
which exploits these structural parameters?}
\emph{\sl Why not have a complexity theory more \emph{\sl fine-tuned} to actual
applications?}

\subsection{The contribution of this article}

Before we launch into examples and definitions, we make some remarks about what this 
article offers. The area of parameterized complexity (as an explicit discipline) 
has been around for around 20 years. 
The group or researchers who have adopted the methodology most strongly are 
those (particularly in Europe) who are working in \emph{\sl applications}.
What has emerged is a very rich collection of distinctive positive techniques 
which should be known to researchers in computational applications in areas as diverse 
as linguistics, biology, cognition, physics, etc. The point here is that once 
you focus on a an extended ``computational dialog'' with the problem
new techniques emerge around this. These techniques 
vary from simple (and practical) local reduction rules (for example 
Karsten Weihe's solution to the ``Europe train station problem'' \cite{Weihe98cove}) to some which 
are highly theoretical 
and use some very deep mathematics (for example, the recent proof that 
topological embedding is  fixed-parameter tractable by  Grohe, Kawarabayashi, Marx, and Wollan \cite{GroheKMW10find},
which uses most of the structure theory of the Graph Minors project of 
Robertson and Seymour~\cite{RobertsonS85}). 

The second  part of this article looks at these positive techniques systematically,
and we hope that the reader will find useful tools there. This area 
is still rapidly developing, and there are many internal contests amongst the 
researchers to develop techniques to beat 
existing bounds. For further details, the reader can see the web site \href{http://fpt.wikidot.com/}{{\tt http://fpt.wikidot.com/}}\ .

The first part of the article is devoted to limitations. It is worth mentioning that 
the area is concerned with tractability \emph{\sl within polynomial time.} Hence
classical methods do not readily apply. The question is 
``how does the problem resides in polynomial time''.
To illustrate this, we begin with the basic motivating examples.\\

\noindent {\sc Vertex Cover}\footnote{Current practice in the area 
is often to write this as {\sc Para-Vertex Cover}, or sometime {\sc p-$k$-Vertex Cover},
but this seems unnecessarily notational. We will usually supress the 
aspect of the problem being regarded as a parameter, or even that the 
problem is considered as a parameterized one when the context is clear. We believe that this will aide the general readability of the paper.}\\
{\sl Instance:}~~A graph $G=(V,E)$.\\
{\sl Parameter:} A positive integer $k$.\\
{\sl Question:}~~Does $G$ have a vertex cover of size $\leq k$?
(A {\em vertex cover} of a graph $G$ is a set $S\subseteq V(G)$ such
that for all edges $\{x,y\}$ of $G$, either $x\in S$ or
$y\in S$.)\\

\noindent {\sc Dominating Set}
\\
{\sl Instance:} A graph $G$.\\
{\sl Parameter:} A positive integer $k$.\\
{\sl Question:}
Does $G$ have a dominating set of size $k$? (A {\it dominating set} is a
set $S\subseteq V(G)$ where, for each
$u\in V(G)-S$ there is a $v\in S$ such that $\{u,v\}\in E(G)$.)\\

Of course both of these problems (without the parameter) are famously {\sf NP}-complete
by the work of Karp~\cite{Karp74}. With the parameter \emph{\sl fixed} 
then both of the problems are in polynomial time simply by trying all
of the ${n}\choose {k} $ subsets of size $k$ where $n$ is the number of vertices of 
the input graph $G$.
What we now know is that  there is an algorithm running in time 
$1.2738^k+ O(n)$ (see \cite{ChenKX10impr}) (i.e., linear time for a fixed $k$ and with an additive component that is mildly exponential 
in $k$), whereas the only known algorithm for 
{\sc Dominating Set} is to try all possibilities. This takes time more or less
$n^{k}$. Moreover, we remark that the methods for solving 
{\sc Vertex Cover} are simple reduction rules which 
are ``industrial strength'' in that they run extremely well in practice 
(e.g. \cite{CDRST03,LangstonPSSV08inno}), even for $k$ beyond which would seem reasonable,
like $k=2000$. 
The reader might wonder: does  this matter?
Table \ref{table3}  from the original book \cite{DowneyF99para} 
 illustrates the difference
between a running time of $\Omega(n^{k+1})$
(that is, where exaustive search is necessary) and a  running time of 
$2^k n$.  The latter has been achieved for several
natural parameterized problems.
(In fact, as we have seen above,  the constant $2^k$
can sometimes be significantly improved and its contribution can sometimes even be additive.)

\begin{table}[h]
\centering
{\tabcolsep=1em\begin{tabular}{| c | c |c|c|}
\hline
&$n=50$&$n=100$&$n=150$
\\
\hline
$k=2$&625&2,500&5,625\\
\hline
$k=3$&15,625&125,000&421,875\\
\hline
$k=5
$&390,625&6,250,000&31,640,625\\
\hline
$k=10$&$1.9\times 10^{12}$&$9.8\times 10^{14}$&$3.7\times 10^{16}$\\
\hline
$k=20$ & $1.8\times 10^{26}$&$9.5\times 10^{31}$&$2.1\times 10^{35}$\\
\hline
\end{tabular}
}
\vspace{4pt}
\caption{The Ratio $\frac{n^{k+1}}{2^kn}$ for Various
Values of $n$ and $k$.}
\label{table3}
\end{table}

Even before we give the formal definitions, the intent of the definitions will
be clear. If we have algorithms running in time 
$n^{c}$, for some $c$ that is independent of $k$, we regard this as being (fixed-parameter) tractable, and 
if $c$  increases with $k$, then this problem  is regarded as being intractable.
In the latter case, we cannot usually prove intractability as it would 
separate {\sf P} from {\sf NP}, as it would in the case of {\sc Dominating Set}, for example.
But we can have a \emph{\sl completeness programme} in the same spirit as {\sf NP}-completeness. 

We mention that the methodology has deep connections with classical complexity.
For example, one of the more useful assumptions for establishing lower bounds 
in classical complexity is what is called the \emph{\sl exponential time hypothesis} (ETH)
which is that not only is $n$-variable 3SAT not in polynomial time, but
in fact it does not have an algorithm running in subexponential time. (Impagliazzo,
Paturi  and Zane \cite{ImpagliazzoPZ01whic}). With this hypothesis, many 
lower bounds can be made rather sharp. Recently Chen and Grohe demonstrated an 
isomorphism between subexponential time complexity and parameterized complexity 
\cite{ChenG07anis}. The connection between subexponential time complexity and 
parameterized complexity was noticed long ago by Abrahamson, Downey and Fellows \cite{AbrahamsonDF95fixed}.
Another notable connection is with \emph{\sl polynomial time approximation
schemes}. Here the idea is to give a solution which is approximate
to within $\epsilon$ of the correct solution. 
Often the PCP theorem allows us to show that no 
such approximation scheme exists unless {\sf P}$=${\sf NP}. But sometimes they do,
but can have awful running times.
For example, here is a table from Downey \cite{Downey03para}:

\begin{itemize}
\item  Arora \cite{Arora96poly} gave a $\O(n^{\frac{3000}{\epsilon}})$ PTAS
for
{\sc Euclidean Tsp}
\item 
Chekuri and Khanna \cite{ChekuriK00apta} gave a $\O(n^{12 (\log({1/\epsilon})/\epsilon^{8})})$
PTAS
for {\sc Multiple Knapsack}
\item  Shamir and Tsur \cite{ShamirT98them} gave a $\O(n^{2^{2^{\frac{1}{\epsilon
}}}-1)})$ PTAS for {\sc Maximum Subforest}
\item  Chen and Miranda \cite{ChenM01apol} gave a $\O(n^{(3mm!)^{\frac{m}{\epsilon}+1}})$
PTAS for {\sc General Multiprocessor Job Scheduling}
\item  Erlebach {\em et al.} \cite{ErlebachJS05poly} gave a
$\O(n^{\frac{4}{\pi}(\frac{1}{\epsilon^2}+1)^2(\frac{1}{\epsilon^2}+2)^2})$
PTAS for {\sc Maximum Independent Set} for geometric graphs.
\end{itemize}

Table \ref{table1} below calculates
some running times for these PTAS's with a 20\% error.

\begin{table}[h]

\begin{center}

\begin{tabular}{|r|r|}
\hline
Reference &Running  Time for a 20\% Error
\\
\hline
Arora \cite{Arora96poly}&$\O(n^{1
5000})$\\
\hline
Chekuri and Khanna \cite{ChekuriK00apta}&
$\O(n^{9,375,000})$\\
\hline
Shamir and Tsur \cite{ShamirT98them}&
$\O(n^{958,267,391})$\\
\hline
Chen and Miranda \cite{ChenM01apol} &
$>\O(n^{10^{60}})$ \\
& (4 Processors)\\
\hline
Erlebach {\em et al.} \cite{ErlebachJS05poly} &
$\O(n^{523,804})$ \\
\hline
\end{tabular}
\end{center}
\caption{The
Running Times for Some Recent PTAS's with 20\% Error.}
\label{table1}
\end{table}

Downey \cite{Downey03para} argues as follows:
\begin{quote}{\sl
``By anyone's measure, a running time of $n^{500,000}$ is bad and
$n^{9,000,000}$  is even worse.
The optimist would argue that these examples are important in that they
prove that PTAS's exist, and are but a first foray. The optimist
would also argue that with  more effort
and better combinatorics, we will be able to come up with
some $n\log n$ PTAS for the problems. For example,
Arora \cite{Arora97near} also came up with another PTAS for {\sc Euclidean Tsp},
but this time it was nearly linear and practical.

But this situation is akin to {\sf P} vs {\sf NP}.
Why not argue that some exponential algorithm is just the first
one and with more effort and better combinatorics
we will find a feasible algorithm for {\sc Satisfiability}?
What if a lot of effort is
spent in  trying to find a practical PTAS's without success?
As with {\sf P} vs {\sf NP}, what is desired
is either an {\em efficient\footnote{An {\em Efficient Polynomial-Time Approximation Scheme (EPTAS)}
is an $(1+\epsilon)$-approximation algorithm that runs in $f(1/\epsilon)\cdot n^{O(1)}$ steps.
If, additionally,  $f$ is a polynomial function then we say that we have a   {\em Fully Polynomial-Time Approximation Scheme (FPTAS)}.
}} PTAS (EPTAS),
  or a proof that no such PTAS exists\footnote{The same  issue can also be raised if we consider FPTAS's instead of EPTAS's.}. 
A primary use of {\sf NP}-completeness is to give compelling evidence
that many problems are unlikely to have better than exponential
algorithms generated by complete search.''}
\end{quote}

The methods of parameterized complexity allow us to address 
the issue. Clearly,  the bad running times are caused by the 
presence of $\frac{1}{\epsilon}$ in the exponent. 
What we could do is  parameterize the problem by taking the parameter 
to be $k=\frac{1}{\epsilon}$ and then perform a reduction to a kind of
core problem (see Subsection~\ref{ptasc}). If we can do this, then \emph{\sl not only  
the particular algorithm is infeasible, but moreover, 
there cannot be a feasible algorithm unless something  unlikely 
(like a miniature {\sf P}$=${\sf NP}) occurs.}

\subsection{Other coping strategies} 
Finally, before we move to formal definitions, 
we mention other strategies which attempt to understand 
what is feasible. One that springs to mind is the theory of average-case 
complexity. This has not really been widely applied as 
the distributions seem rather difficult to apply.
 Similar comments apply to 
the theory of smoothed analysis. 

In fact one of the reasons that parameterized complexity has been used so often in 
practice is that it is widely applicable. We also mention that 
often it can be used to explain unexpected 
tractability of algorithms. Sometimes they seem to work because of 
underlying hidden parameters  in the  input. 
For example,
 the number of ``lets'' in some structured programming
languages in practice is usually bounded 
by some constant, and sometimes engineering 
considerations make sure that, for example, the number of 
wafers in VLSI design is small. 
 It is even conceivable 
there might be a parametric explanation of the 
tractability of 
the Simplex Algorithm. 

If the reader finds this all useful, then we refer him/her to two recent issues 
of the \emph{\sl The Computer Journal} \cite{Downey08thec} devoted to 
aspects and applications of parameterized complexity, to the survey
of Downey and McCartin \cite{RodneyM04some} on parameterized algorithms, 
the (somewhat dated) survey Downey \cite{Downey03para} for issues in complexity,
and other articles such as \cite{DowneyFS99thev,DowneyFS97survey,Fell0ows02survey,Fellows01para,Fellows03survey,DowneyFS99para} as well as the books by 
Downey and Fellows \cite{DowneyF99para}, Niedermeier \cite{Niedermeier06invi} and Flum and Grohe \cite{FlumG04para}.

\subsection{Organization and notation}

The paper is organized as follows:

In Section \ref{preliminaries} we will give the basic definitions and 
some examples to show the kinds of parameterizations 
we can look at. We consider as fortunate the fact that a problem can have different parameterizations
with different complexities.

In Section \ref{intract}, we will introduce some of the basic 
hardness classes, and in particular the \emph{\sl main standard}
of hardness, the class $W[1]$. The gold standard is 
established by an analog of the Cook-Levin Theorem 
discussed in Subsection \ref{analog}.
Parameterized reductions are more refined than the corresponding 
classical ones, and, for instance, it would appear that 
natural parameterized versions of 
{\sc 3-CNF Sat}
and {\sc CNF Sat} 
do \emph{\sl not} have the same parameterized complexity.
In Subsection \ref{W_hiera} we see how this gives rise 
to a hierarchy based on logical depth. We also mention the Flum-Grohe
${\sf A}[t]$ hierarchy which is another parameterized hierarchy of 
problems based on another measure of logical depth.
In Subsections \ref{ptasc} and 
\ref{optimality} we look at PTAS's, approximation, and 
lower bounds based on strong parameterized hypotheses. 
In Subection~\ref{ocapp} 
we look at other applications to classical questions which 
are sensitive to combinatorics in polynomial time, and 
in Subsections~\ref{opcla} and \ref{parapp} look at 
other parameterized classes such as counting classes and
the notion of parameterized approximation. 
(The latter seeks, for example, an {\sf FPT}-algorithm which on
input $k$ either delivers a ``no size $k$ dominating set''
or produces one of size $2k$.)
Subsection~\ref{limker}  deals with an important recent development.
One of the most important practical techniques  \emph{\sl kernelization.}
Here one takes a problem specified by $(x,k)\in \Sigma^{*}\times \Bbb{N}$ and 
produces, typically in polynomial time, a 
small version of the problem: $(x',k')$ such that 
$(x,k)$ is a yes iff $(x',k')$ is a yes, and 
moreover $|x'|\le f(k)$
and usually $k'\le k$. This technique is 
widely used in practice as it usually relies on a number of 
easily implementable reduction rules as we will
discuss in Subsection \ref{kernelization}.
We will look at recent techniques 
which say when this technique can be used 
to give $x'$ as above with $|x'|$ polynomially bounded.
The final part of the complexity section deals with 
\emph{\sl some} of the other classes we have left out.

In Section \ref{algorithms}, we turn to techniques for 
the design of parameterized algorithms.
We will focus mainly on graphs; for to lack of space we do not 
discuss too many applications, except \emph{\sl en passant}.

The principal contributor of high running times in 
classical algorithms comes from branching, and large search trees.
For this reason we begin by looking at methods 
which restrict the branching in Subection \ref{derand},
including bounded search trees  and greedy localization.
Then in Subection \ref{fasoc} we look at the use of 
automata and logic in the design of parameterized algorithms.

In turn this leads to meta-theoretical methods 
such as applications of Courcelle's theorem on 
graphs of bounded treewidth, and other methods 
such as local treewidth and First Order Logic (FOL).
This is discussed in Subection \ref{cource} and \ref{ftpfol}.

In Subections  
\ref{graphminors}  and \ref{irrelevant}
we look at the highly impractical, but powerful
methods emerging from the Graph Minors project.
Later, we examine how these methods can be sped up.

Having worked in the stratosphere of algorithm design we return to focus on
singly-exponential algorithm design techniques 
such as iterative compression, bidimensionality theory,
and then in Subsection \ref{kernelization} move to 
kernelization, and finish with variations.

Of course, often we will use all of this in combination.
For example, we might kernelize, then begin a branch 
tree of some bounded size, and the rekernelize the smaller graphs.
It is often the case that this method is provably faster and certainly 
this is how it it is often done in practice.
However, we do not have space  in this already long paper to 
discuss such refinements in more detail.

\section{Preliminaries}
\label{preliminaries}

\subsection{Basic definitions}
\label{definitions}

The first thing to do is to define a proper notion of tractability for parameterized problems. This induces the definition of the 
parametrized complexity class {\sf FPT},
namely the class of  fixed-parameter tractable problems.

\begin{definition}
A {\em parameterized problem} (also called {\em parameterized language}) 
is a subset $\Pi$ of $\Sigma^{*}\times \Bbb{N}$ where $\Sigma$ is some alphabet.  
In  the input $(I,k)\in \Sigma^{*}\times \Bbb{N}$ of a parameterized problem, we call $I$ as the {\em main part of the input}
and $k$ as the {\em parameter of the input}. We also agree that  $n=|(I,k)|$.
We say that $\Pi$ is {\em fixed parameter tractable} if there exists a
function $f: \Bbb{N}\rightarrow \Bbb{N}$ and 
an algorithm deciding whether $(I,k)\in \Pi$
in $$O(f(k) \cdot n^{c})$$
steps, where $c$ is a constant not depending on the parameter $k$ of the problem. We call such an algorithm {\em {\sf FPT}-algorithm} or, more concretely, to visualize the choice of $f$ and $c$, we say that $\Pi\in O(f(k) \cdot n^{c})$-\FPT.
We define 
the parameterized class {\sf FPT} as the one  containing all parameterized  problems that can be 
solved by an {\sf FPT}-algorithm.
 \end{definition}

Observe that an apparently 
 more demanding definition of  an {\sf FPT}-algoritm would ask for algorithms runnning in  
$O(f(k) + n^{c})$
steps, since then the exponential 
part would be additive rather than multiplicative. However, this would not define a different parameterized complexity class.
To see this, suppose that some parameterized problem $\Pi\subseteq \Sigma\times \Bbb{N}$ can be solved 
by an algorithm ${\cal A}$ that  can decide whether some $(I,k)$ belongs in $\Pi$ in $f(k)\cdot n^{c}$ steps.
In case $f(k)\leq n$, the same algorithm requires $n^{c+1}$ steps, while if $n<f(k)$, the algorithm 
runs in less than $(f(k))^{c+1}$ steps. In both cases, ${\cal A}$ solves $\Pi$  in at most $g(k)+n^{c'}$ steps
where $g(k)=(f(k))^{c+1}$ and $c'=c+1$.


Time bounds for parameterized algorithms have two parts. The $f(k)$ is called {\em parameter dependence}
and, is typically a super-polynomial function. The $n^{c}$ is a polynomial function and we will call it {\em polynomial part}. While in classic algorithm design there is only polynomial part to improve, in {\sf FPT}-algorithms it appears to be more important to improve the parameter dependence. Clearly, for practical purposes, an $O(f(k) + n^{c})$ step {\sf FPT}-algorithm is  more welcome than one running in $O(f(k) \cdot  n^{c})$ steps.

%

\subsection{Nomenclature of parameterized problems}
 Notice that 
a problem of classic complexity whose input has several integers has 
several parameterizations depending on which one  is  chosen to be the parameter.
We complement the name of a parameterized problem so to indicate the 
parameterization that we choose. In many cases, the parameterization 
refers to a property of the input. As a driving example we consider the following problem: \\

\noindent{\sc Dominating Set}\\
{\sl Instance:} a graph $G$ and an integer $k$\\
{\sl Question}:  does $G$ have a dominating set of size at most $k$? \\

\noindent {\sc Dominating Set} has several parameterizations. The most popular one is the following one:\\

\noindent{\sc $k$-Dominating Set}\\
{\sl Instance:} a graph $G$ and an integer $k$.\\
{\sl Parameter:} $k$.\\
{\sl Question:}  does $G$ have a dominating set of size at most $k$? \\

\noindent Moreover, one can define parameterizations that do not depend on integers appearing 
explicitly in the input of the problem.
For this, one may set up a ``promise'' variant of the problem based on a suitable restriction of its inputs. That way, we may 
define the following  parameterization of {\sc Dominating Set}:\\

\noindent{\sc $d$-Dominating Set}\\
{\sl Instance:} a graph $G$ with maximum degree $d$ and an integer $k$.\\
{\sl Parameter:} $d$.\\
{\sl Question:}  does $G$ have  a dominating set of size at most $k$? \\

\noindent In the  above problem the promise-restriction is ``with maximum degree $d$''.
In general, we often omit this  restriction as it is  becomes clear by the chosen parameterization.
Finally, we stress  that we can define the parameterization by combining a promise 
restriction with  parameters that appear in the input. As an example, we can define the following 
parameterization of  {\sc Dominating Set}:\\

\noindent{\sc $d$-$k$-Dominating Set}\\
{\sl Instance:} a graph $G$ with maximum degree $d$ and an integer $k$.\\
{\sl Parameter:} $d+k$.\\
{\sl Question:}  does $G$ have a dominating set of size at most $k$? \\

Finally, the promise-restriction can be just a property of the main part of the input.  A typical example is the following parameterized problem.\\

\noindent{\sc $k$-Planar Dominating Set}\\
{\sl Instance:} a planar graph $G$ and an integer $k$.\\
{\sl Parameter:} $k$.\\
{\sl Question:}  does $G$ have  a dominating set of size at most $k$? \\

Certainly, different parameterizations may belong to different parameterized complexity classes. 
For instance, {\sc $d$-$k$-Dominating Set} belongs to $O((d+1)^{k}\cdot n)$-$\FPT$,
using the bounded search tree method presented in Subsection~\ref{subs_boundedst}.
Also, as we will see in Subsection~\ref{subexpalg}, {\sc $k$-Planar Dominating Set} belongs to $2^{O(\sqrt{k})}\cdot n^{O(1)}$-$\FPT$.
On the other side, {\sc $d$-Dominating Set} $\not\in \FPT$, unless ${\sf P}={\sf NP}$.
This follows by the well known fact that {\sc Dominating Set} is {\sf NP}-complete for 
graphs of maximum degree $3$ and therefore, not even an $n^{f(d)}$-algorithm is expected to exist for 
this problem.
The parameterized complexity of {\sc $k$-Dominating Set}
needs the definition of the {\sf W}-hierarchy (defined in Subsection~\ref{W_hiera}). 
While the problem can be solved in $O(n^{k+1})$ steps,  it
is known to be complete for the second level of the ${\sf W}$-hierarchy.
This indicates that an {\sf FPT}-algorithm is unlikely  to exist for this parameterization.

\section{Parameterized complexity}
\label{intract}

\subsection{Basics}
In this section we will look at some basic methods of establishing apparent 
parameterized intractability. We begin with the class ${\sf W}[1]$ and the ${\sf W}$-hierarchy,
and later look at variations, including the ${\sf A}$ and ${\sf M}$ hierarchies, 
connections with approximation, bounds on kernelization and the like.

The role of the theory of {\sf NP}-completeness is to give some kind of 
outer boundary for tractability. That is, if we identify {\sf P} with ``feasible'',
then showing that a problem is {\sf NP}-complete would suggest that the 
problem is computationally intractable. Moreover,  we would
believe that a deterministic algorithm for the problem would require worst-case exponential time.

However, showing that some problem is in {\sf P} does not say that the problem 
is feasible.
Good examples are  the standard parameterizations of  {\sc Dominating Set} 
or {\sc Independent Set} for which we know of no algorithm significantly better
than trying all possibilities. For a fixed $k$, trying all possibilities
 takes 
time
$\Omega(n^{k+1})$, which is infeasible for large $n$ and reasonable $k$,
\emph{\sl in spite of the fact that the problem is in ${\sf P}$}.  Of course, we would { like} to prove that
there is {\sl no} \FPT\ algorithm for such a problem, but, as with classical complexity, the best we can 
do is to formulate some sort of completeness/hardness program.
Showing that {\sc $k$-Dominating Set} is not in \FPT\ would also show, as a corollary, that ${\sf P}\neq  {\sf NP}.$

A hardness program needs three things. First, it needs a notion of easiness, which we have:
{\FPT}. Second, it needs a notion of reduction, and third, it needs some core problem
which we believe to be intractable.

Following naturally from the concept of fixed-parameter
tractability is an appropriate notion of reducibility that expresses the fact that two parameterized problems have comparable parameterized complexity. 
That is, if problem (language) $A$ reduces to problem (language) $B$, and problem $B$ is fixed-parameter tractable, then so too is problem $A$.  

\begin{definition}[Downey and Fellows \cite{DowneyF95fixe-I,DowneyF95fixe-II}-Parameterized reduction] 
A {\em parameterized reduction\footnote{Strictly speaking, this is a 
parameterized {\em many-one} reduction as an analog of the classical Karp reduction. 
Other variations such as parameterized Turing reductions are possible. The function 
$g$ can be arbitrary, rather than computable, for other non-uniform versions. 
We give the reduction most commonly met.}} from a parameterized
language $L$ to a parameterized language $L'$ (symbolically $L\leq_{\FPT}L'$) is an algorithm that
computes, from input consisting of a pair $(I, k)$, a pair
$(I', k' )$ such that:

\begin{enumerate}
\item $(I, k) \in L$ if and only if $(I', k' )
\in L'$, \item $k' = g(k)$ is a computable function depending only $k$, and \item
the computation is accomplished in time $f(k)\cdot  n^{c}$, where $n$ is the size of the main part of 
the input $I$, $k$ is the parameter, $c$ is a 
constant (independent of both $n$ and $k$), and $f$ is an arbitrary function dependent only on $k$.
\end{enumerate}
If $A\leq_{\FPT} B$ and $B\leq_{\FPT} A$, then we say that $A$ and $B$ are {\em \FPT-equivalent
 ,}
and write $A\equiv_{\FPT} B$.
\end{definition}

A simple example of an \FPT\ reduction is the fact that $k$-{\sc Independent Set}
$\equiv_{\FPT}$ $k$-{\sc Clique}. 
(Henceforth, we will usually drop the 
parameter $k$ from the name of problems 
and will do so when the parameter is implicit from the 
context.) Namely, $G$ has a clique of size 
$k$ iff the complement of $G$ has an independent set of size $k$. 
A simple {\em non-example} is the classical reduction of 
{\sc Independent Set} to {\sc Vertex Cover}: $G$ will have a 
size $k$ independent set iff $G$ has a size $n-k$ vertex cover,
where $n$ is the number of vertices of $G$. The point of this last example is 
that the {\em parameter} is not fixed.

\subsection{An analog of the Cook-Levin Theorem}
\label{analog}

We need the final component for our
program to establish the apparent parameterized intractability 
of computational problems: the identification of  a ``core'' problem to reduce from.

In classical {\sf NP}-completeness this is the heart of the Cook-Levin Theorem: 
the argument that a nondeterministic Turing machine is such an opaque object that
it does not seem reasonable that we can determine in polynomial time if it has an accepting 
path from amongst the exponentially many possible paths. 
Building on earlier
work of Abrahamson, Ellis, Fellows  and Mata 
\cite{AbrahamsonFEM89}, the idea of Downey and Fellows
was to define reductions and certain core problems 
which have this property.
In the fundamental papers \cite{DowneyF95fixe-I,DowneyF95fixe-II},
a parameterized version of 
{\sc Circuit Acceptance}. The classic version of this problem
has as instance a boolean circuit
and the question is whether  some value assignment to the input variables
leads to a yes. As is well known, this 
corresponds to Turing Machine acceptance, at least classically.
Downey and Fellows \cite{DowneyF95fixe-II} 
combined with Cai, Chen, Downey and Fellows \cite{CaiCDF97onth}
allows for a Turing Machine core problem:
\\

\noindent{\sc Short Non-Deterministic Turing Machine Acceptance}\\
\noindent {\sl Instance:}~~A nondeterministic Turing machine $M$ (of arbitrary 
degree of non-determinism).  \\
{\sl Parameter:} A positive integer $k$.\\
{\sl Question:}~~Does $M$ have a computation path accepting the empty string in
at most $k$ steps?\\

In the same sense that {\sf NP}-completeness of the {\sc $q(n)$-Step
Non-deterministic Turing Machine  Acceptance}, where $q(n)$ is a polynomial in the size of the
 input, provides us with very strong evidence that no {\sf NP}-complete problem is likely to
 be solvable in polynomial time,
using {\sc Short Non-Deterministic Turing Machine Acceptance} as a hardness core 
 provides us with
very
strong evidence that no 
parameterized 
language $L$, for which {\sc Short Non-Deterministic Turing Machine Acceptance}$\leq_{\FPT}L$,
is likely to be
fixed-parameter tractable.
That is, if we accept the idea behind the basis
of {\sf NP}-completeness,  then we should also accept that the  {\sc Short Non-deterministic Turing
 Machine Acceptance} problem is not solvable in time $O(|M|^c)$ for some  fixed $c$.
Our intuition would again be that all computation paths would need to be tried.

We remark that the hypothesis ``{\sc Short Non-Deterministic Turing Machine Acceptance}
is not in \FPT'' is somewhat stronger than {\sf P}$\neq${\sf NP}. Furthermore,
connections between this hypothesis and classical complexity have recently become apparent.
If {\sc Short Non-Deterministic Turing Machine Acceptance} {\em is} in \FPT,
then we know that the {\sc Exponential Time Hypothesis},
which states that $n$-variable {\sc 3Sat} is not in 
subexponential time (DTIME$(2^{o(n)})$), fails. See Impagliazzo, Paturi and Zane 
\cite{ImpagliazzoPZ01whic}, Cai and Juedes \cite{CaiJ03onth}, and 
Estivill-Castro, 
Downey,
Fellows,  Prieto-Rodriguez and   Rosamond \cite{DowneyCFPR03cutt}
(and our later discussion of ${\sf M}[1]$) for more details.
As we will later see the ETH is a bit stronger than the 
hypothesis that {\sc Short Non-Deterministic Turing Machine Acceptance}
is not in \FPT, but is equivalent to an apparently stronger 
hypothesis that ``${\sf M}[1]\neq \FPT$''.
The precise definition of ${\sf M}[1]$  will be given later, but
the idea here is that, as most researchers believe, not only is 
${\sf NP}\ne {\sf P}$,
but ${\sf NP}$ problems like 
{\sc Non-Deterministic Turing Machine Acceptance}
require substantial search of the available 
search space, and hence do not  
have algorithms running in deterministic subexponential time such as 
$O(n^{\log n})$.

The class of problems that are \FPT-reducible to {\sc Short Non-Deterministic Turing Machine Acceptance} 
is called ${\sf W}[1]$, for reasons discussed below. 
The parameterized analog of the classical Cook-Levin theorem (that {\sc CNF Sat}
is {\sf NP}-complete) uses the following parameterized version of 
{\sc 3Sat}:\\

\noindent{\sc Weighted CNF Sat}\\
\noindent {\sl Instance:}~~A CNF formula $X$ (i.e., a formula in Conjunctive Normal Form).  \\
{\sl Parameter:} A positive integer $k$.\\
{\sl Question:}~~Does $X$ have a satisfying assignment of weight $k$?\\ 

\noindent Here the {\em weight} of an assignment is its Hamming weight, that is, the 
number of variables set to be true.\\

  Similarly, we can define {\sc Weighted $n$CNF Sat}, where the clauses have only $n$ 
variables
and $n$ is some number fixed in advance.
{\sc Weighted $n$CNF Sat}, for any fixed $n \geq 2$, is complete for ${\sf W}[1]$.

\begin{theorem}[Downey and Fellows \cite{DowneyF95fixe-II} and Cai, Chen, Downey
and Fellows \cite{CaiCDF97onth}] For any fixed $n\geq 2$,
{\sc Weighted $n$CNF Sat}$\equiv_{\FPT}
$ \label{ccdf} {\sc Short Non-Deterministic Turing Machine Acceptance}.\end{theorem}

As we will see, there are many well-known problems hard for 
${\sf W}[1]$. For example, {\sc Clique} and {\sc Independent set}
are basic ${\sf W}[1]$-complete problems. An example 
of a ${\sf W}[1]$-hard problem {\sc Subset Sum} which classically has as input 
a set $S$ of integers, a positive integers and an integer $s$
and asks if there is a set of members of $S$ which add to $s$.
In parametric form, the question asks where there exist $k$ members of 
$S$ which add to $s$. A similar problem is {\sc Exact Cheap Tour}
which asks for a weighted graph whether there is a tour 
through $k$ nodes of total weight $s$. 
Another example is the {\sc Finite State Automata Intersection}
which has parameters $m$ and $k$ and asks 
for  a set $\{A_1,\dots,A_k\}$ of finite state automata over an alphabet 
$\Sigma$ whether there is a string $X$ of length $m$ 
accepted by each of the $A_i$, for $i=1,\dots,k.$
There are a number of problems related to the {\sc Least Common Subsequence}
which are hard according to various parameterizations, and notably
a ${\sf W}[1]$-hard parameterized version of {\sc Steiner Tree}.
Here, and later, we will refer the reader to the various 
monographs and compendia 
such as the one currently maintained by Marco Cesati:

\centerline{\href{http://bravo.ce.uniroma2.it/home/cesati/research/compendium/}{{\tt http://bravo.ce.uniroma2.it/home/cesati/research/compendium/}}}
\medskip

The original proof of Theorem \ref{ccdf} involves a generic simulation of 
a Turing machine by a circuit and then in the other direction 
involves combinatorial arguments to have parametric reductions from 
certain kinds of circuits (``weft 1'', see below) to {\sc Independent Set}.
 
Since the original Downey-Fellows work, hundreds of problems 
have been shown to be ${\sf W}[1]$-hard or ${\sf W}[1]$-complete. A classical graph theoretical 
problem which is ${\sf W}[1]$-complete (see~\cite{DowneyF99para,DowneyF95fixe-II}) is the following.\\

\noindent $d$-{\sc Red-Blue Nonblocker}\\
{\em Instance:} A 2-coloured graph $(V_{Red}\cup V_{Blue},E)$ of (fixed) maximum degree $d\ge 2$, and a positive integer $k$.
\\
{\em Parameter:} A positive integer $k$.
\\
{\em Question:} Is there a set of red vertices $V'\subseteq V_{Red}$ 
of cardinality $k$ such that each blue vertex has at least one neighbour not
belonging to $V'$?.\\

This problem seems to be useful for proving hardness results such as 
the proof that various coding problems are hard for ${\sf W}[1]$ in
Downey, Fellows, Vardy, and Whittle \cite{DowneyFVW99thep}. 
In that problem, using a reduction from
 {\sc Red-Blue Nonblocker}
 Downey et. al. show the following natural 
problem is hard for ${\sf W}[1]$.\\

\noindent{\sc Weight Distribution}
\\
{\sl Instance:} A binary $m\times n$ matrix $H$.\\
{\sl Parameter:} An integer $k>0$.\\
{\sl Question:} Is there a set of at most $k$ columns of 
$H$ that add to the all zero vector?\\

As a consequence,
the related problem of {\sc Maximum Likelihood Decoding } where there is a
target vector $s$ is also 
${\sf W}[1]$ hard. Two notorious open questions
remain in this area.\\

\noindent{\sc Shortest Vector}
\\
{\sl Instance:} A basis $X=\{x_1,\dots,x_n\}\in {\mathbb Z}^n$ for a lattice
${\mathbb L}$.\\
{\sl Parameter:} An integer $k>0$.
\\
{\sl Question:}  Is there a nonzero vector $x\in {\mathbb L}$ such that $||x||^2\leq k$?\\

\noindent{\sc Even Set} (also known as {\sc Minimum Distance})
\\
{\sl Instance:} A red/blue graph $G=(R,B,E)$.\\
{\sl Parameter:} A positive integer $k$.\\
{\sl Question:} Is there a set of at most $k>0$ red vertices 
all of which have an even number of blue neighbours?\\

Both of these are conjectured in \cite{DowneyF99para} as being ${\sf W}[1]$ hard. The 
unparameterized version of the latter 
is known to be {\sf NP}-complete and of the former is 
famously open, a question of Peter van Emde Boas from 1980. We refer the reader to
\cite{DowneyFVW99thep} for more details and other related problems such as {\sc Theta 
Series}.

\subsection{The ${\sf W}$-hierarchy}
\label{W_hiera}
The original theorems 
and hardness classes were first characterized in terms of boolean circuits of 
a certain structure. These characterizations lend themselves to easier 
{\em membership} proofs, as we now see.
This uses 
the model of a {\em decision circuit}. This has boolean variables as inputs,
and a single output. It has {\em and} and {\em or}  gates and 
{\em inverters}. We designate a gate as {\em large} or {\em small} depending
on the fan-in allowed, where small will be some fixed number.
For example  a 3CNF formula can be modeled by a circuit consisting of $n$ 
input variables  (of unbounded fanout) one for each formula variable, possibly inverters
below the variable, and  a  large {\em and} of small
{\em or}'s (of size 3) with a single output line. 
For  a decision circuit, the
{\em depth} is the maximum number of gates on any path
from the input variables to the output line, and the {\em weft} is the
``large-gate depth''.  More precisely, the weft is defined
to be the maximum number of large gates on any path from the input
variables to the output line, where a  gate is called large if its
fan-in exceeds some pre-determined bound.

The {\em weight} of an assignment to the input variables of a decision circuit is the Hamming weight, i.e., the number of variables set to true by the assignment.

Let ${\cal F}=\{C_1,\ldots,C_n,\ldots\}$ be a family of decision
circuits.  Associated with ${\cal F}$ is a basic parameterized language
$$L_{\cal F}=
\{\langle C_i,k \rangle: C_i {\rm ~has~a~weight~}k{\rm ~satisfying
~assignment }\}\;.$$

We will denote by $L_{{\cal F}(t,h)}$
the parameterized language associated with the family of weft $t$,
depth $h$, decision circuits.

\begin{definition} [${\sf W}\lbrack t\rbrack$ -- Downey and  Fellows \cite{DowneyF92}]
\label{W_1}
We define a language $L$ to be in the class ${\sf W}[t], t\geq 1$ if
there is a parameterized reduction from
$L$ to $L_{{\cal F}(t,h)}$, for some $h$.
\end{definition}

We think of the inputs of {\sc Weighted CNF Sat} as circuits consisting of  conjunctions  of disjunctions of literals. 
Hence  {\sc Weighted CNF Sat}  is in ${\sf W}[2]$. Extending this 
idea,
a typical example of a formula in ${\sf W}[3]$ would be a conjunction of disjunctions of 
conjunctions of literals. 
More generally, 
we can define {\sc Weighted $t$-Normalized Sat} as the weighted satisfiability
 problem for a formula $X$ where $X$ is a conjunction of disjunctions of conjunctions of disjunctions $\ldots$
 with $t$ alternations. 
 
This allows for the following basic result.

\begin{theorem}[Downey and Fellows \cite{DowneyF95fixe-I}]
For all $t\ge 1$, {\sc Weighted $t$-Normalized Sat} is complete for 
${\sf W}[t]$.\end{theorem}

There are problems complete for other ${\sf W}[t]$ levels such as {\sc Dominating Set}
being complete for ${\sf W}[2]$, but as with many situations in logic and computer science,
natural problems at levels above $3$ tend to be rare. Cesati \cite{Cesati03thet}
gave 
machine-based problems for various levels 
and his approach (via Turing machines) can allow for easier proofs,
such as his proof of the ${\sf W}[1]$-completeness of {\sc Perfect Code}
(\cite{Cesati02perf}).

As an illustrative example, we will give one of the basic  reductions.

\begin{theorem}[Downey and Fellows \cite{DowneyF95fixe-I}] {\sc Dominating Set}
$\equiv_{\FPT} 
$ {\sc Weighted CNF Sat}.\end{theorem}

\begin{proof}
We sketch the proof of the hardness direction that {\sc Weighted CNF Sat} $\leq_{\FPT}$ {\sc Dominating Set}. 

Let $X$ be a Boolean expression in conjunctive normal form 
consisting of $m$ clauses $C_{1},...,C_{m}$ over the set of 
$n$ variables $x_{0},...,x_{n-1}$.  We 
show how to produce in polynomial-time by local replacement, 
a graph $G=(V,E)$ that has a dominating set 
of size $2k$ if and only if $X$ is satisfied by a truth assignment 
of weight $k$.

\begin{figure}[t]
\begin{center}
{\includegraphics[scale=0.15]{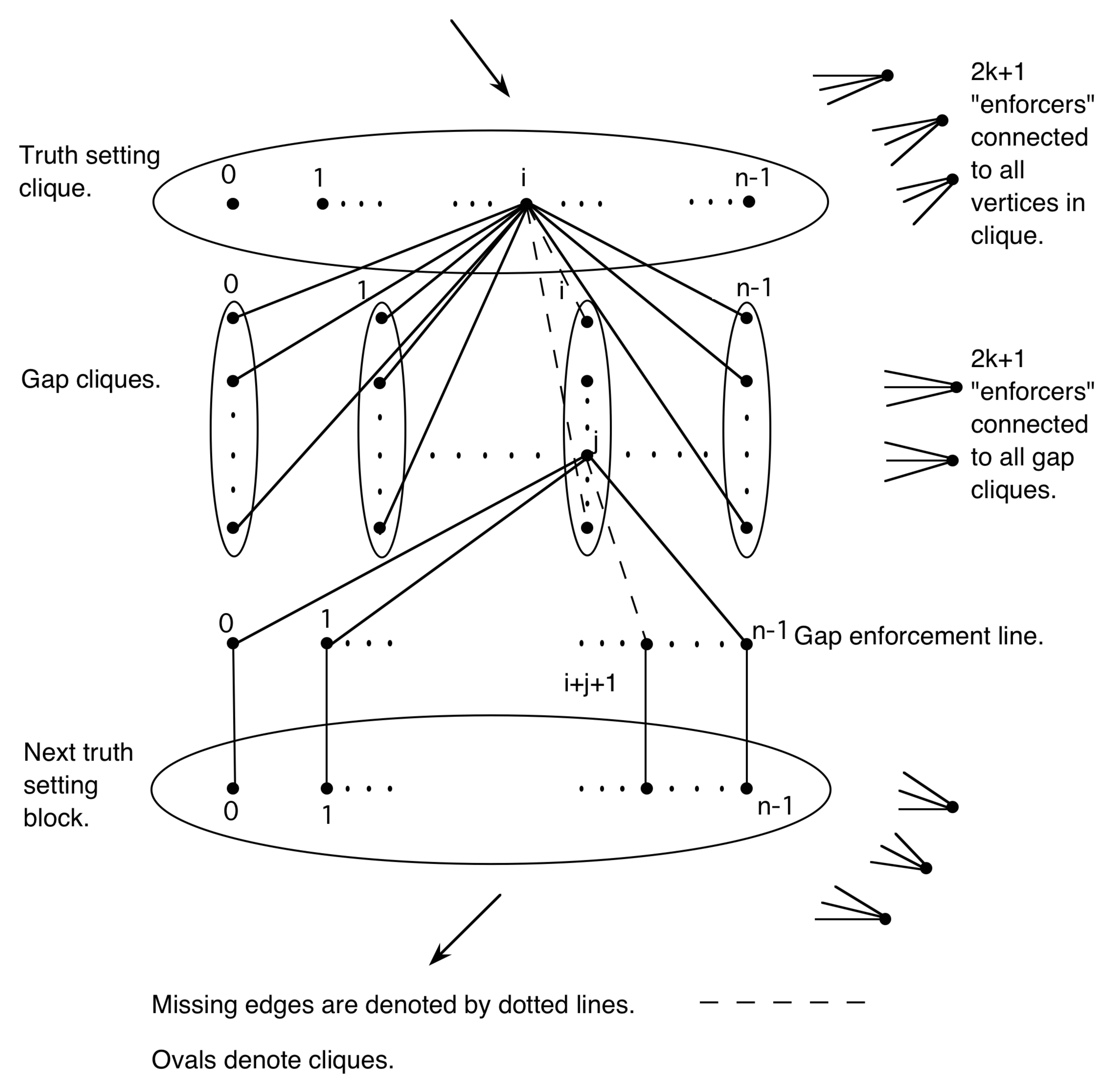}}
\end{center}
\caption{Gadget for {\sc Dominating Set}}
\label{dominating}
\end{figure}

A diagram of the gadget used in the reduction is given in Figure~\ref{dominating}. The idea of the proof is as follows. 
There are $k$ of the gadgets arranged in a circle,
where we regard them as ordered from first to last. Each of the gadgets has 3 main parts. Taken clockwise from top to bottom, these are the {\em truth setting
clique,} the
{\em gap selection} part (achieved by the gap selection cliques)
and the {\em gap enforcement} part (achieved by the gap enforcement line). 

The pigeonhole principle
combined with the so-called 
\emph{\sl enforcers}  are used  
to force one vertex from each of the truth cliques and 
one vertex from each of the next set of cliques, which form the gap 
enforcement part. The 
intuition is that the truth selection cliques represent 
a choice of a particular vertex to be selected to be true, and 
the gap selection represents the gap till the next 
selected true vertex. The interconnections between the 
truth setting cliques and the gap selection means that
they align, and the gap enforcement line makes all
the selections consistent. Finally because of  the \emph{\sl clause variables}
which also need to be dominated, we  will ensure that the
dominating set corresponds to a truth assignment.

In more detail,
the truth selection component is a  clique and the gap selection consists of $n$ cliques which we call {\em columns}.
Our first action is to ensure that in {\em any}
 dominating set of $2k$ elements, 
we must pick one vertex from each of these two components. This goal is achieved by
the $2k$ sets of $2k+1$ enforcers (which are independent sets).
For example, for each truth selection clique, $2k+1$ enforcers  are connected to
every vertex  in 
this clique and nowhere else and then it will
follow by the pigeonhole principle that in any size $2k$ dominating set
for the final graph, to dominate 
these enforcers, {\em some vertex} in the truth selection clique must 
be chosen.
Similarly, it follows that we must pick {\em exactly}
one vertex from each of the truth selection cliques,
and one of the gap selection cliques
to dominate the enforcers.

The truth selection component of the $r$-th gadget  is denoted by $A(r)$, $r=0,\ldots,k-1$.
%
Each of these $k$ components consists of a clique of 
$n$ vertices labeled $0,\ldots,n-1$. The intention is  that if  the vertex labeled $i$ is picked, 
this represents  variable $i$ being set to  true in the formula $X$.
We denote by $B(r)$ the gap selection part of the $r$-th gadget, $r=0,\ldots,k-1$.
As explained above, this part consists of $n$ columns (cliques)
where we index the columns by $i=0,\ldots,n-1$.
%
%
The intention is that column $i$ corresponds to the choice of variable $i$ in the 
preceding $A(r)$. The idea then is the following. We join the vertex
$a[r,i]$ corresponding to variable $i, $ in $A(r)$, to all vertices in $B(r)$ 
{\em except} those in column $i$. This means that the choice of 
$i$ in $A(r)$ will cover all vertices of $B(r)$ except those in this column.
It follows 
since we have only
$2k$ to spend,
that we {\em must} choose the dominating element from this column and nowhere else. (There are no connections from column to column.)
The columns are meant to be the gap selection saying how many $0$'s there will be till the next positive choice for a variable.
We finally need to ensure that if we choose variable $i$ in $A(r)$ and 
gap $j$ in column $i$ from $B(r)$ then we need to pick $i+j+1$ in $A(r+1)$. This
is fulfilled by  the gap enforcement component which consists of a set of 
$n$ vertices. 
We denote by $d[r,s]$, $s=0,\ldots,n-1$ the set of  vertices in this gap-enforcement
 line in the $r$-th gadget $r=0,\ldots,k-1$.

For $r<k-1$,
the method is to connect vertex $j$ in column $i$ of $B(r)$ to all 
of the $n$ vertices $d[r,s]$ $except$ to $d[r,i+j+1]$
\emph{\sl provided that $i+j+1\le n$}. (For $r=k-1$
simply connect vertex $j$ in column $i$
of $B(r)$ to all of the $n$ vertices $d[r,s]$ except to 
$d[r,i+j+1]$  since this will need to ``wrap around''
to $A(0)$.)  The first point of this process is 
that if we choose vertex $j$ in column $i$ with $i+j+1>n$
then \emph{\sl none} of the vertices in the enforcement line are 
dominated. Since there is  only a single edge to the corresponding vertex in 
$A(r+1)$, there cannot possibly be a size $2k$ dominating set 
for such a choice. It follows that we must choose 
some $j$ with $i+j+1\le n$ in any dominating set of size $\le 2k$. 
The main point is 
that if we choose $j$ in column $i$ we will dominate all of the $d[r,s]$ 
except $d[r,i+j+1]$. Since 
we will only connect $d[r,s]$ additionally to $a[r+1,s]$ and nowhere else, 
to choose an element of $A[r+1]$ and still dominate all of the $d[r,s]$ we 
must actually choose $a[r+1,i+j+1]$.


Thus the above provides a selection gadget that chooses $k$ true variables
with the gaps representing false ones. We enforce that the selection is 
consistent with the clauses of $X$ via  {\em clause vertices}
$c_i$ one for each clause $C_{i}$.
These are connected in the obvious ways. One connects a choice in $A[r]$ or $B[r]$ corresponding to making a clause $C_q$ true to the vertex $c_q$. Then if we dominate 
all the clause variables too, we must have either chosen in some $A[r]$ a 
positive occurrence of a variable in $C_q$ or we must have chosen in $B[r]$ 
a gap corresponding to a negative occurrence of a variable in $C_q$, and conversely. 
The formal details can be found in Downey and Fellows \cite{DowneyF95fixe-I,DowneyF99para}.\end{proof}

There are notable problems which are ${\sf W}[t]$-hard for all $t$,
such as {\sc Bandwidth}, below. \\

\noindent{\sc Bandwidth}\\
\noindent{\sl Instance:} A graph $G=(V,E)$.\\
{\sl Parameter:}
A positive integer $k$.\\
{\sl Question:}  
Is there a 1-1 layout $f:V\to \{1,\dots,|V|\}$ such that
$\{u,v\}\in E$ implies $|f(u)-f(v)|\leq k$?\\


The $W[t]$ (for all $t$) hardness of 
{\sc Bandwidth} 
 was proven by Bodlaender, Fellows and Hallett \cite{BodlaenderFH04beyo} via a rather complex
reduction. It is unknown what, say, {\sc Bandwidth} $\in {\sf W}[1]$ (or even 
${\sf W}[{\sf P}]$) would imply
either parametrically or classically.
On general grounds, it seems unlikely that such a containment is possible.

At even higher levels,
we can define {\sc Weighted Sat}  to be the weighted satisfiability
 problem where inputs correspond to unrestricted  Boolean formulas
 and finally 
{\sc Weighted Circuit Sat} to be the most general problem whose inputs are all polynomial sized circuits.

Notice that, 
in Theorem \ref{ccdf}, we did {\em not} say that 
{\sc Weighted CNF Sat} is ${\sf W}[1]$-complete. 
The reason for this is that we do not believe that 
it is!
In fact,  we believe that ${\sf W}[2]\neq {\sf W}[1].$

That is,
classically, using a padding argument, we know that {\sc CNF Sat} $\equiv_{m}^{P}$ {\sc $3$CNF Sat}. However, 
the classical reduction {\em does not} define a parameterized reduction from {\sc Weighted CNF Sat} to 
{\sc Weighted $3$CNF Sat}, it is not structure-preserving enough to ensure that parameters map to parameters. 
In fact, it is conjectured \cite{DowneyF95fixe-I} that there is {\em no} parameterized reduction at all from {\sc 
Weighted CNF Sat} to {\sc Weighted $3$CNF Sat}.
If the conjecture is correct, then {\sc Weighted CNF Sat} is {\em not} in the class ${\sf W}[1]$.

The point here is that parameterized reductions are more refined than classical ones, and
hence we believe that we get a wider variety of apparent hardness behaviour 
when intractable problems are classified according to this 
more fine-grained analysis.

These classes form part of the basic hierarchy of parameterized problems
below.
$$\FPT \subseteq {\sf W}[1] \subseteq {\sf W}[2] \subseteq \cdots \subseteq {\sf W}[t]
\subseteq {\sf W}[SAT] \subseteq {\sf W}[{\sf P}] \subseteq {\sf AW}[t]\subset {\sf AW}[{\sf P}] \subseteq {\sf {\sf XP}}$$

 This sequence is commonly termed ``the {\em ${\sf W}$-hierarchy}''. The complexity
class ${\sf W}[1]$ can be viewed as  the parameterized analog of {\sf NP}, since it 
suffices for the purpose of establishing  likely 
parameterized intractability.

The classes ${\sf W}[SAT]$, ${\sf W}[{\sf P}]$ and the ${\sf AW}$ classes were
introduced by Abrahamson, Downey and Fellows in \cite{AbrahamsonDF95fixed}.
The class ${\sf W}[SAT]$ is the collection of parameterized languages $\FPT$-reducible to {\sc Weighted Sat}.  
The class ${\sf W}[{\sf P}]$ is the collection of parameterized languages \FPT-equivalent to {\sc Weighted Circuit Sat}, the weighted satisfiability problem for a decision circuit  $C$ that is unrestricted.
A standard translation of Turing machines into circuits shows that {\sc $K$-Weighted Circuit Sat} is the same as the problem of deciding whether or not a deterministic Turing machine accepts an input of weight $k$. It is conjectured that the containment ${\sf W}[SAT] \subseteq {\sf W}[{\sf P}]$ is proper  \cite{DowneyF99para}.

Another way to view ${\sf W}[{\sf P}]$ is the following. 
Consider the problem {\sc Short Circuit Sat} defined as follows.\\

\noindent{\sc Short Circuit Sat}\\
\noindent{\em Instance:} A decision circuit $C$ with at most $n$ gates and 
$k\log n$ inputs.\\
{\em Parameter:} A positive integer $k$.\\
{\em Question:} Is there a setting of the inputs making $C$ true?\\

\begin{thm}[Abrahamson, Downey and Fellows \cite{AbrahamsonDF95fixed}]
{\sc Short Circuit Sat} is ${\sf W}[{\sf P}]$-complete.\end{thm}

\begin{proof} The proof of this result uses the ``$k\cdot \log n$'' trick
introduced by Abrahamson, Downey and Fellows \cite{AbrahamsonDF95fixed}.
To see that the problem is ${\sf W}[{\sf P}]$-hard, take an instance $I$ of 
{\sc Weighted Circuit Satisfiability} with parameter $k$ and inputs $x_1,\dots,
x_n$. Let $z_1,\dots,z_{k\log n}$
be new variables. Using lexicographic order and in polynomial time 
we have a surjection from this set to the $k$-element subsets of 
$x_1,\dots,x_n$. Representing this as a circuit and putting this 
on the top of the circuit for $I$ defines our new circuit $C$.
The converse is equally easy.\end{proof}

${\sf AW}[t]$ captures the notion of {\it alternation}. ${\sf AW}[t]$ is the 
collection of parameterized languages $\FPT$-reducible to  {\sc Parameterized Quantified Circuit Sat$_t$}, the weighted satisfiability problem for an unrestricted decision circuit that applies {\em alternating quantifiers} to the inputs, defined here.\\

\noindent {\sc Parameterized Quantified Circuit Sat$_t$}

\noindent {\sl Instance:}~~A weft $t$ 
decision circuit $C$ whose inputs correspond to a sequence $s_1, \ldots s_r$ of pairwise disjoint sets  of variables.  \\
{\sl Parameter:} $r, \; k_1, \ldots , k_n$.\\
{\sl Question:}~~Is it the case that there exists a size $k_1$ subset $t_1$ of $s_1$, 
such that for every size $k_2$ subset $t_2$ of $s_2$, 
there exists a size $k_3$ subset $t_3$ of $s_3$, 
such that $\ldots$ (alternating quantifiers) 
such that, when $t_1 \cup t_2 \cup \ldots \cup t_r$ are set to true, 
and all other variables are set to false, $C$ is satisfied?\\

The idea here is to look at the analog of {\sf PSPACE}. The problem is that 
in the parameterized setting there seems no natural analog of 
Savitch's Theorem or the proof that 
{\sc QBFSat} is {\sf PSPACE}-complete, and it remains an interesting 
problem to formulate a true analog of parameterized space.

The approach taken by \cite{AbrahamsonDF95fixed} was to look at 
the  parameterized analog of {\sc QBFSat} stated above.

One of the fundamental theorems proven here is that the choice of $t$
is irrelevant:

\begin{theorem}[Abrahamson, Downey and Fellows \cite{AbrahamsonDF95fixed}]
${\sf AW}[t]={\sf AW}[1]$ for  all $t\geq 1$.\end{theorem}

Many parameterized analogs of game problems are complete for the ${\sf AW}[1]$,
such as the parameterized analog of {\sc Geography}.\\

\noindent 
{\sc Short Geography}
\\
{\sl Instance:} A directed graph $D=(V,E)$ and a specified vertex $v_0$.\\
{\sl Parameter:} A positive Integer $k$.
\\
{\sl Question:} Does player 1 have a winning strategy in the following game?
Each player alternatively chooses a new arc from $E$. The first arc must have its 
tail as $v_0$, and each vertex subsequently chosen must have as its tail
the head of the previously chosen arc. The first player unable to choose
loses.\\

The class 
${\sf XP}$, introduced in \cite{DowneyF99para},
 is the collection of parameterized languages $L$ such that the {\em $k$-th slice} of $L$ (the instances of $L$ having  parameter $k$) is  a member of ${\sf DTIME}(n^k)$. ${\sf XP}$ is provably distinct from $\FPT$ and seems to be the parameterized class corresponding to the classical class ${\sf EXP}$ (exponential time).
There are problems complete for ${\sf XP}$ including the game of $k$-{\sc Cat and Mice}.
The problem here is played on a directed graph $G=(V,E)$ and begins 
with a distinguished vertex $v_0$ called the cheese, a token $c$ on one vertex
called the cat, and $k$
tokens 
(called mice) on other vertices. Players cat and the team of mice 
play alternatively moving one token at a time.
A player can move a token along an arc
at each stage. Team of mice wins if one mouse can reach the cheese (by occupying it
even if the cat is already there) before 
the cat can eat one of the mice by occupying the same vertex.

It is conjectured that all of the containments given so far  are proper, but all that is currently known is
that \FPT\ is a proper subset of ${\sf XP}$.

There are hundreds of complete problems for the levels of the hierarchy.
Here is a short list. 
The reader is referred to \cite{Downey03para} for references and definitions.
As stated ${\sf XP}$ has {\sc $k$-Cat and Mouse Game} and 
many other games; ${\sf W}[{\sf P}]$ has
           {\sc Linear Inequalities},
{\sc Short Satisfiability}, 
           {\sc Weighted Circuit Satisfiability} 
and
  {\sc Minimum Axiom Set}. There are a number of quite important
problems from combinatorial pattern matching which are ${\sf W}[t]$ hard for all
$t$:
      {\sc Longest Common Subsequence} ($k=$ number of sequences, $|\Sigma|$-two
parameters), 
  {\sc Feasible Register Assignment},
 {\sc Triangulating Colored Graphs},
                        {\sc Bandwidth},
{\sc Topological Bandwidth},
                        {\sc Proper Interval Graph Completion},
 {\sc Domino Treewidth} and {\sc Bounded Persistence Pathwidth}.
Some concrete problems complete for ${\sf W}[2]$ include
         {\sc Weighted $\{0,1\}$ Integer Programming},
  {\sc Dominating Set},
           {\sc Tournament Dominating Set,}
          {\sc Unit Length Precedence Constrained Scheduling} (hard),
        {\sc Shortest Common Supersequence} (hard), 
        {\sc Maximum Likelihood Decoding} (hard),
        {\sc Weight Distribution in Linear Codes} (hard),
      {\sc Nearest Vector in Integer Lattices} (hard),
        {\sc Short Permutation Group Factorization} (hard).
Finally a collection of ${\sf W}[1]$-complete problems: {\sc $k$-Step Derivation for Context Sensitive Grammars},
{\sc Short NTM Computation},
  {\sc Short Post Correspondence}, {\sc Square Tiling},
       {\sc Weighted $q$--CNF Satisfiability},
        {\sc Vapnik--Chervonenkis Dimension}, 
        {\sc Longest Common Subsequence ($k$, $m=$ length of common subseq.)},
 {\sc Clique},
        {\sc Independent Set},
and
        {\sc Monotone Data Complexity for Relational Databases}. 
          This list is merely representative, and 
new areas of application are being found all the time.
There are currently good compendia of hardness and completeness results
as can be found at the web

 \centerline{\href{http://bravo.ce.uniroma2.it/home/cesati/research/compendium/}{{\tt http://bravo.ce.uniroma2.it/home/cesati/research/compendiuma/}}}
 
\noindent  For older material, see  the appendix of
the monograph by Downey and Fellows \cite{DowneyF99para}, as well as the many
surveys and the recent issue of the {\em Computer Journal}
\cite{Downey08thec}.

There remain several important structural questions associated 
with the ${\sf W}$-hierarchy such as how it relates to the ${\sf A}$-hierarchy 
below, and whether any collapse may propagate.\\

\noindent {\em Open questions :} Does ${\sf W}[t]={\sf W}[t+1] 
$ imply ${\sf W}[t]={\sf W}[t+2]$? Does ${\sf W}[t]=\FPT$ imply ${\sf W}[t+1]=\FPT?$

\subsection{The ${\sf A}$-hierarchy and the Flum-Grohe approach}
\label{A_hier}
There have been several attempts towards simplification of this material,
notably by Flum and Grohe \cite{FlumG04para}. Their method is to 
try to make the use of logical depth and logic more
explicit. To do this Flum and Grohe take a detour through
the logic of 
finite model theory. Close inspection of their proofs reveals
that similar combinatorics are hidden. Their view is that \emph{\sl model checking} 
should be viewed as the fundamental viewpoint of complexity.

To this end, for a class of first order formula $\varphi$
with
$s$-ary 
 free relation variable,
we can define
$p\mbox{-\sc WD}(\varphi)$ as the problem:\\

\noindent {$p\mbox{-\sc WD}(\varphi)$}\\
\noindent {\em Instance:}
A structure ${\mathcal A}
$ with domain $A$ and an integer $k$.\\
{\em Parameter:} $k$.\\
{\em Question:}
Is there a relation $R\subseteq A^s$ with $|R|=k$ such that
${\mathcal A}\vDash \varphi(R)?$\\

This idea naturally extends to {\em classes} $\Phi$  of formulae $\Phi$.
Then we define $p\mbox{-\sc WD}(\Phi)$ to be the class of 
parameterized problems $p\mbox{-\sc WD}(\varphi)$ for  $\varphi\in \Phi$.
It is easy to show that $p\mbox{-\sc WD}({\Sigma_1})=\FPT.$ (Strictly, we should write
$\Sigma_1^0$ but the formulae considered here are first order.)
Then to recast the classical 
${\sf W}$-hierarchy at the finite levels, the idea is to  {\em define}
for $t\geq 1$,
$${\sf W}[t]=[p\mbox{-\sc WD}({\Pi_t})]^{\FPT},$$
where, given a parameterized problem $X$,  $[X]^{\FPT}$ denotes the parameterized problems that are  $\FPT$-reducible to $X$.

Flum and Grohe have similar logic-based formulations of the 
other {{\sf W}-hi\-e\-rar\-chy} classes. We refer the reader to \cite{FlumG04para} 
for more details. 

We remark that the model checking  approach leads to other hierarchies. One important hierarchy
found by Flum and Grohe is the ${\sf A}$-hierarchy which is also 
based on alternation like  the {\sf AW}-hierarchy but works differently.
For a class $\Phi$ of formulae, 
we can define the following parameterized problem.\\

\noindent $p\mbox{\sc -MC}(\Phi)$
\\
{\em Instance:} A structure ${\mathcal A}$ and a formula $\varphi \in \Phi$.\\
{\em Parameter:} $|\varphi|$.
\\
{\em Question:} Decide if $\phi({\mathcal A}
)\neq \emptyset,$ 
where this denotes the evaluation of $\phi$ in ${\mathcal A}$.\\

Then Flum and Grohe define 
$${\sf A}[t]=[p\mbox{\sc -MC}(\Sigma_t)]^{\FPT}.$$
For instance, for $k\geq 1$, $k$-{\sc Clique}
can be defined by 
$$\mbox{clique}_k=\exists x_1,\dots x_k(\bigwedge _{1\le i<j\le k}x_i\neq x_j\land 
\bigwedge_{1\le i<j\le k}Ex_ix_j)$$ in the language of graphs, and the interpretation
of the formula in a graph $G$ would be that $G$ has a clique of size $k$.
Thus the mapping $(G,k)\mapsto (G,\mbox{clique}_k)$ is a \FPT\ reduction
showing that parameterized {\sc Clique} is in ${\sf A}[1].$
Flum and Grohe populate various levels of the $A$-hierarchy and 
show the following.

\begin{theorem}[Flum and Grohe \cite{FlumG04para}] The following hold:
\begin{itemize}
\item[(i)] ${\sf A}[1]={\sf W}[1].$
\item[(ii)] ${\sf A}[t]\subseteq {\sf W}[t]$.
\end{itemize}
\end{theorem}
Clearly ${\sf A}[t]\subseteq {\sf XP}$, but no other containment 
with respect to other classes of the ${\sf W}$-hierarchy
is known. 
It is conjectured by Flum and Grohe that no other containments than those given 
exist, but this is not apparently related to any other conjecture.

If we compare classical and parameterized complexity it is evident
that the framework provided by parameterized complexity theory
allows for more finely-grained complexity analysis of
computational problems. 
It is deeply connected with algorithmic heuristics and 
exact algorithms in practice. We refer the reader to 
either the survey  \cite{FlumG04para}, or 
those in two recent issues of {\em The Computer Journal} \cite{Downey08thec} for further insight.

We can consider many different
parameterizations of a single classical problem, each of which
leads to either a tractable, or (likely) intractable, version in the
parameterized setting. This allows for an extended
dialog with the problem at hand. This idea towards
the solution of algorithmic problems is explored in, 
for example, \cite{DowneyFS99thev}. 
A nice example of this extended dialog can be found in the work of Iris van Rooij
and her co-authors, 
as discussed in van Rooij and Wareham \cite{vanRooijW08para}.

\subsection{Connection with {\sf PTAS}'s}
\label{ptasc}
The reader may note that parameterized complexity is addressing intractability 
{\sl within polynomial time}. In this vein, the parameterized framework can be used to demonstrate 
that many classical problems that admit a {\sf PTAS} do not, 
in fact, admit any {\sf PTAS} with a practical running time,
unless ${\sf W}[1]=\FPT$. 
The idea here is that if a ${\sf PTAS}$ has a running time such as 
 $O(n^{\frac{1}{\epsilon}})$, where $\epsilon$ is the 
error ratio, then the {\sf PTAS} is unlikely to be useful. For example if
$\epsilon=0.1$ then the running time is already $n$ to the 10th power 
for an error of $10\%$. What we could do is regard $\frac{1}{\epsilon}$ as a parameter
and show that the problem is ${\sf W}[1]$-hard with respect to that prameterization.
In that case {\sl there would likely be no method of
removing the $\frac{1}{\epsilon}$ from the 
exponent in the 
running time and hence no {\em efficient} {\sf PTAS}},
a method first used by Bazgan \cite{Bazgan95sche}.
For many more details of the method we refer the reader to the survey 
\cite{Downey03para}.

A notable application of this technique is due to 
Cai et. al.~\cite{CaiFJR07thec}
who showed that the  method of using planar formulae 
tends to give {\sf PTAS}'s that are never practical. 
The exact calibration of {\sf PTAS}'s and parameterized complexity
comes through yet another hierarchy called the ${\sf M}$-hierarchy.
The breakthrough was the realization by
Cai and Juedes \cite{CaiJ01sube} that the collapse of   a basic sub-hierarchy
of the ${\sf W}$-hierarchy 
was deeply connected with approximation. 

The base level of the hierarchy is the problem ${\sf M}[1]$ defined by the core problem
below.\\

\noindent
{\em Instance:} A CNF circuit $C$ (or, equivalently, a CNF formula) of size $k\log n$.\\
{\em Parameter:} A positive integer $k$.\\
{\em Question:} Is $C$ satisfiable? \\

That is, we are parameterizing the {\em size} of the problem
rather than some aspect of the problem. 
The idea naturally extends to higher levels for that, for example, {\sf M}[2] would 
be a product of sums of product formula of size $k\log n$ and 
we are asking whether it is satisfiable.

The basic result is that ${\sf FPT}\subseteq {\sf M}[1]\subseteq {\sf W}[1]$.
The hypothesis ${\sf FPT}\ne {\sf M}[1]$ is {\em equivalent to} ETH.
In fact, as proven in \cite{ChenG07anis}, there is an isomorphism, the so-called
\emph{miniaturization}, between exponential time complexity (endowed
with a suitable notion of reductions) and {\sf XP} (endowed with {\sf FPT}
reductions) such that the respective notions of tractability correspond,
that is, subexponential time on the one and \FPT\ on the other side.
The reader might wonder why we have not used a Turing Machine 
of size $k\log n$ in the input, rather than a CNF circuit. The reason is that
whilst we can get reductions for small version of problems, such as  $k\log n$-{\sc Vertex Cover}
and the like,  to have the same complexity as the circuit problem above, 
we do not know how to do this 
for Turing machines. It remains an open question whether
the $k\log n$-sized circuit problem and the $k\log n$-sized 
Turing Machine problem have the same complexity.

For more on this topic and other connections with 
classical {exponential} algorithms we refer the reader to the 
survey of Flum and Grohe \cite{FlumG04para}.

\subsection{{\sf FPT} and {\sf XP} optimality}
\label{optimality}

Related to the material of the last section 
are emerging programmes devoted to proving tight lower bounds on 
parameterized problems, assuming various non-collapses of the 
parameterized hierarchies.
For this section, it is useful to use the $O^*$ notation for 
parameterized algorithms. This is the part of the running time which 
is exponential. 
For example, a running time of $2^{k^2}|G|^5$ would
be written as $O^*(2^{k^2}).$

One example of a lower bound was the original paper of Cai and Juedes
\cite{CaiJ01sube,CaiJ03onth} who proved the following definitive result.

\begin{thm}
\label{ethvc}

$k$-{\sc Planar Vertex Cover}, $k$-{\sc Planar Independent Set},
$k$-{\sc Planar Dominating Set}, and $k$-{\sc Planar Red/Blue Dominating Set}
 cannot be in
$O^*(2^{o(\sqrt{k})})$-$\FPT$ unless
${\sf FPT}={\sf M}[1]$ (or, equivalently, unless ETH fails).
\end{thm}

The optimality of Theorem~\ref{ethvc} follows from the fact that  all above problems have been classified in $O^*(2^{O(\sqrt{k})})$-{\sf FPT} as proved in~\cite{AlberBFKN02,KloksLL02newa,FominT06domi,KoutsonasT10plan} (see also Subsection~\ref{subexpalg}).

We can ask for similar optimality results for any \FPT\ problem.
See for example, Chen and Flum \cite{ChenF06onmi}.
We will meet another approach to \FPT\ optimality 
in Subsection \ref{boundedx} where we look at 
classes like {\sf EPT} meant to capture how problems are placed
in \FPT\ via another kind of completeness programme.

Another example of such optimality programmes 
can be found in exciting resent work on {\sf XP} optimiality.
This programme represents a major step forward 
in the sense that it regards the classes like ${\sf W}[1]$ as artifacts of the 
basic problem of proving hardness 
under reasonable assumptions, and strikes at 
membership of ${\sf XP}$.

Here are some  examples. We know that {\sc Independent Set} and 
{\sc Dominating Set} are in ${\sf XP}$.

\begin{thm}[Chen et. al \cite{ChenCFHJKX05tigh}] The following hold:
\begin{itemize}
\item[(i)] {\sc Independent Set} cannot be solved in
time $n^{o(k)}$ unless ${\sf FPT}={\sf M}[1].$
\item[(ii)] {\sc Dominating Set} cannot be solved in
time $n^{o(k)}$ unless ${\sf FPT}={\sf M}[2]$.
\end{itemize}
\end{thm}

A beautiful development in this area is the resent paper by Marx 
on the {\sc Closest Substring} problem. We refer 
to Marx \cite{Marx08clos} for more details, and 
to Chen and Meng \cite{ChenM08onpa} for other related  results.

There remains a lot of work to be done here and these 
programmes appear to be exciting developments, see e.g., \cite{FedorGLS09cliq,FominGLS10algo,LokshtanovMS11know,LokshtanovMS10know}.

\subsection{Other classical applications}
\label{ocapp}
Similar techniques  have  been used to 
solve a significant open question about techniques for 
formula evaluation when they were used to
show ``resolution is not automizable'' unless ${\sf W}[{\sf P}]=\FPT$ (Alekhnovich and Razborov
\cite{AlekhnovichR01reso}, Eickmeyer,  Grohe and  Gr{\"u}ber \cite{EickmeyerGG08appr}.) 
Parameterized complexity assumptions  can also be used to show 
that the large hidden constants (various towers of two's)
in the running times of generic algorithms obtained though the use of algorithmic meta-theorems cannot be improved upon (see \cite{FlumG06para}.)
One illustration is obtained through the use of 
{\em local treewidth} (which we will discuss in Subection \ref{ftpfol}). 
The 
 notions of treewidth and branchwidth  are  by now ubiquitous
in algorithmic graph theory (the definition of branchwidth is given in  Subsection~\ref{cource}). 
Suffice to say is that it is a method of decomposing 
graphs to measure how ``treelike'' they are, 
and if they have small treewidth/branchwidth, 
as we see in Subection \ref{cource}, we can run dynamic 
programming algorithms upon them.
The {\em local} treewidth of a {\em class} ${\mathcal C}$ of graphs is called 
bounded iff there a function $f$ such that for all 
graphs $G\in {\mathcal C}$ and all vertices $v\in V(G)$
the neighborhood of $v$ of distance $k$ from $v$ has treewidth $f(k)$ (see Subsection \ref{ftpfol}).
Examples include planar graphs and  graphs of bounded maximum degree. 
The point is that a class of graphs of bounded local treewidth 
is automatically \FPT\ for a wide class of properties.

\begin{theorem}[Frick and Grohe \cite{FrickG99deci}] 
Deciding first-order statements is \FPT\ for every fixed class 
of graphs of bounded local treewidth.\end{theorem}

One problem  with this algorithmic meta-theorem is that 
the algorithm obtained for a fixed first-order statement $\varphi$ 
can rapidly have towers of twos depending on the quantifier complexity
of the statement, in the same way that this happens for Courcelle's 
theorem on decidability of monadic second order statements 
(as discussed in Subsection \ref{cource})
for graphs of bounded treewidth. What Frick and Grohe \cite{FrickG04thec} 
showed is 
that 
such towers of two's cannot be removed unless ${\sf W}[{\sf P}]=\FPT.$

Another use of parameterized complexity is to give an indirect approach to
proving likely intractability of problems which are not known to be 
{\sf NP}-complete. A classic example of this is the following problem.\\

\noindent {\sc Precedence Constrained $k$-Processor Scheduling}
\\
{\em Instance:} A set $T$ of unit-length jobs and a partial ordering 
$\preceq$ on $T$, a positive deadline $D$ and a number of processors $k$.\\
{\sl Parameter:} A positive integer $k$.\\
{\sl Question:} Is there a mapping $f:T\to \{1,\dots,D\}$
such that for all $t\prec t'$, $f(t)\prec f(t')$, and for all $i$, $1\leq i\leq D$,
$|f^{-1}(i)|\leq k?$.\\

In general this problem is {\sf NP}-complete and is known to be
in {\sf P} for 2 processors. The question is what happens 
for $k>2$ processors. For us the question becomes 
whether the problem is  in {\sf XP} for $k>2$ processors.
This remains one of the open questions from Garey and Johnson's famous book
\cite{GareyJ79comp}
(Open Problem 8), but we have the following.

\begin{theorem}[Bodlaender, Fellows and Hallett \cite{BodlaenderFH04beyo}]
{\sc Precedence Constrained $k$-Processor Scheduling} is ${\sf W}[2]$-hard.
\end{theorem}

The point here is that {\sl even if} \ {\sc Precedence Constrained $k$-Processor Sche\-duling} is in {\sf XP}, there seems no way that it will be  feasible 
for large $k$.
Researchers in the area of parameterized 
complexity have long wondered whether  this approach might 
be applied to other problems like {\sc Composite Number} or 
{\sc Graph Isomorphism}.
For example, Luks \cite{Lucks82isom} has shown that 
{\sc Graph Isomorphism} can be solved in $O(n^k)$ for  graphs 
of maximum degree  $k$, but any proof that the problem was ${\sf W}[1]$
hard would clearly imply that the general problem was not 
feasible. We know that {\sc Graph Isomorphism} is 
almost certainly not {\sf NP}-complete, since proving that
would collapse the polynomial hierarchy to 2 or fewer levels
(see the work of Schöning in~\cite{Schoning87grap}). Si\-mi\-lar comments can be applied to
graphs of bounded treewidth by Bodlaender \cite{Bodlaender90poly}.

\subsection{Other parameterized classes}
\label{opcla}

There have been other parameterized analogs of classical complexity 
analyzed. For example McCartin \cite{McCartin06para} and Flum and Grohe \cite{FlumG02thep}
each analyzed parameterized counting complexity. Here we 
can define the class $\#{\sf W}[1]$ for instance
(with the core problem being counting the number of 
accepting paths of length $k$ in a nondeterministic
Turing machine), and show that
counting the number of $k$ cliques is $\#{\sf W}[1]$-complete.
Notably Flum and Grohe proved the following analog to Valiant's Theorem on the 
permanent.

\begin{theorem}[Flum and Grohe \cite{FlumG02thep}] Counting the number of 
cycles of size $k$ in a bipartite graph is $\#{\sf W}[1]$-complete.\end{theorem}

One of the hallmark theorems of classical complexity is Toda's 
Theorem which states that ${\sf P}^{\#{\sf P}}$ contains the polynomial time hierarchy.
There is no known analog of this result for parameterized complexity.
One of the problems is that all known proofs of Toda's Theorem 
filter through probabilistic classes. Whilst there are known  
analogs of Valiant's Theorem (Downey, Fellows and Regan \cite{DowneyFR98para}, and M{\"u}ller
\cite{Muller06rand,Muller08para}), there is no known method of ``small'' probability
amplification. (See Montoya \cite{Montoya08thep} for a thorough discussion of this problem.)
This seems the main problem and there is really no 
satisfactory treatment of  probability amplification
in parameterized complexity.
For example, suppose we wanted an analog of the operator  calculus for 
parameterized complexity.  For example,
consider ${\sf BP}\cdot\oplus {\sf W}[{\sf P}]$, as an analog
of ${\sf BP}\cdot \oplus {\sf P}$. We can define $L\in \oplus {\sf W}[{\sf P}]$
to mean that  $(x,k)\in L$ iff 
$(x,k)\mapsto (x',k')$ where $x'$ is a circuit with (e.g.)
$k\log |x'|$ inputs and the number of accepting 
inputs is odd. We need that
there is a language $L\in {\sf BP}\cdot {\mathcal C}$
iff there is a language $L'\in {\mathcal C}$ such that for all $(x,k)$,
$$(x,k)\in L\mbox{ iff a randomly chosen }(x,k,k')\in L'.$$
A problem will appear as soon as we try to prove the 
analog basic step in Toda's Theorem: ${\sf W}[{\sf P}]\subseteq {\sf BP}\cdot \oplus {\sf W}[{\sf P}].$ 
The first step in the usual proof of Toda's Theorem,  which {\em can}
be emulated,  is to use  some kind of random hashing 
to result in  a $\oplus {\sf W}[{\sf P}]$ circuit with either 
no accepting inputs, or exactly one accepting input, the latter with {\em nonzero}
probability. So far things work out okay.
However, the  next step in the usual proof is 
to amplify this probability: that is, do this a polynomial number of times
independently 
to amplify by majority to get the result in the ${\sf BP}$ class.
The problem is that if this amplification uses
\emph{\sl many} independently chosen instances,
then the number of input variables goes up and the 
the result is no longer a ${\sf W}[{\sf P}]$ circuit since 
we think of this 
as a polynomial sized circuit with {\em only $k\log n$ many inputs.}
There is a fundamental question: {\em Is it possible to 
perform amplification with only such limited nondeterminism?}

Notable here are the following:

\begin{theorem}[Downey, Fellows and Regan \cite{DowneyFR98para}]
For all $t$, there is a randomized {\sf FPT}-reduction from
${\sf W}[t]$ to unique ${\sf W}[t]$. (Analog of of the Valiant-Vaziarini 
Theorem)\end{theorem}

\begin{theorem}[M{\"u}ller \cite{Muller08vali}]
${\sf W}[{\sf P}]\cdot {\sf BPFPT}$ (an analog of ${\sf BPP}$) has 
weakly uniform derandomization (``weakly uniform'' is an 
appropriate technical condition) iff 
there is a polynomial time computable unbounded function $c:
{\mathbb N}\to {\mathbb N}$
with ${\sf BPP}(c)={\sf P}$, where ${\sf BPP}(c)$
denotes {\sf BPP} with only access to $c(n)\log n$ nondeterministic bits.
\end{theorem}

Moritz M{\"u}ller's result says that, more or less, 
parameterized derandomization implies nontrivial
classical derandomization.
Other interesting work   on randomization in parameterized complexity
is to be found in the work of M{\"u}ller.
For instance, in \cite{Muller06rand}, he showed that
there is a Valiant-Vaziarini type lemma for 
most ${\sf W}[t]$-complete problems, including e.g.\
{\sc Unique Dominating Set}. 
The only other work in this area is
in the papers of Montoya, such as \cite{Montoya08thep}.
In \cite{Montoya08thep}, Montoya showed that it
is in a certain sense unlikely that BP$\cdot${\sf W}[{\sf P}], an analogue of the
classical Arthur-Merlin class, allows probability amplification.
(That is, amplification with $k\log n$ bits of nodeterminism
is in a sense unlikely.)
Much remains to do here.

Perhaps due to the delicacy of material, or because of the 
focus on practical computation, there is only a little work on 
what could be called parameterized \emph{structural} complexity.
By this we mean analogs of things like Ladner's Theorem (that if
$P\ne NP$ then there are sets of intermediate complexity)
(See \cite{DowneyF93fixe}), Mahaney's 
Theorem that if there is a sparse $NP$ complete set then $P=NP$ (Cesarti
and Fellows \cite{CesatiF96spar},
the PCP Theorem, Toda's theorem etc. 
There is a challenging research agenda here.


\subsection{Parameterized approximation}
\label{parapp}
One other fruitful area of research has been the area of {\em parameterized 
approximation},
beginning with three papers 
at the same conference! (Cai and  Huang
\cite{CaiH06fixe}, Chen, Grohe and Gr{\"u}ber \cite{ChenGG06onpa},
and Downey, Fellows and McCartin
\cite{DowneyFM06para}). 
Parameterized approximation was  part of the folklore for 
some time originating with 
the dominating set question 
originally asked by Fellows.
For parameterized approximation,
 one inputs a problem and asks for either a solution of size 
$f(k)$ or a statement that there is no solution of size $k$. This idea
was originally suggested by Downey and Fellows, inspired by earlier work
of Robertson and Seymour on approximations to treewidth. 
Of course, we need to assume that $\FPT\ne {\sf W}[1]$ for this to  
make sense.
A classical example taken from Garey and Johnson 
\cite{GareyJ79comp}
is {\sc Bin Packing} where the First Fit algorithm  
either says that no packing of size $k$ exists or gives one 
of size at most $2k$.
As observed by Downey, Fellows and McCartin~\cite{DowneyFM06para} 
most ${\sf W}[t]$ hard problems do not have approximations with an additive 
factor (i.e.  $f(k)=k+c$) unless ${\sf W}[t]={\sf FPT}.$ One surprise from that paper is the following.

\begin{theorem}[Downey, Fellows, McCartin \cite{DowneyFM06para}] The problem
{\sc Independent Dominating Set} which asks if there is a dominating set
of size $k$ which is also an independent set, has no 
parameterized  approximation algorithm for {\em any} computable function
$f(k)$ unless ${\sf W}[2]=\FPT.$\end{theorem}

Subsequently, Eickmeyer,  Grohe  Gr{\"u}ber showed the following in \cite{EickmeyerGG08appr}.

\begin{theorem}
If $L$ is a ``natural'' ${\sf W}[{\sf P}]$-complete language
then $L$ has no 
parameterized approximation algorithm unless 
${\sf W}[{\sf P}]={\sf FPT}.$ \end{theorem}

The notion of ``natural'' here is indeed quite natural and covers all of the 
known ${\sf W}[{\sf P}]$-complete problems, say, in the appendix 
of Downey and Fellows \cite{DowneyF99para}. We refer the reader to \cite{EickmeyerGG08appr}
for more details.

One open question asks whether there is 
any multiplicative \FPT\ approximation for {\sc Dominating Set}.
This question of Mike Fellows
has been open for nearly 20 years and asks in its oldest incarnation
whether there is an algorithm which, on input $(G,k)$ either says 
that there is no size $k$ dominating set, or that there is one of size $2k$.

\subsection{Limits on kernelization}
\label{limker}

There has been important recent work concerning limitations of 
parameterized techniques. One of the most important techniques is that 
of {\em kernelization.} We will see in Subsection \ref{kernelization}
 that this is one of the basic techniques 
of the area, and remains one of the most practical techniques
since usually kernelization is based around simple reduction rules
which are both local and easy to  implement.
(See, for instance 
Abu-Khzam, et. al. \cite{Abu-KhzamCFLSS04kern}, Flum and Grohe \cite{FlumG04para},
or Guo and Niedermeier \cite{GuoN07invi}.)
The idea, of course, is that we shrink the 
problem using some polynomial time reduction to a small one whose size depends only on the parameter, and then do exhaustive  search on that kernel.
Naturally the latter step is the most time consuming and 
the problem becomes how to find  small kernels.
An important question is therefore: When is it possible to show that a problem
has {\em no polynomial kernel?}
A formal definition of kernelization is the following:

\begin{definition}[Kernelization]
\label{Definition:Kernelization}%
A \emph{kernelization algorithm}, or in short, a \emph{kernel} for a
parameterized problem $L \subseteq \Sigma^* \times \mathbb{N}$ is an
algorithm that given $(x,k) \in \Sigma^* \times \mathbb{N}$, outputs
in $p(|x|+ k)$ time a pair $(x',k') \in \Sigma^* \times \mathbb{N}$
such that
\begin{itemize}
\item $(x,k)\in L \Leftrightarrow (x',k')\in L$,
\item $(|x'|,k') \leq f(k)$,
\end{itemize}
where $f$ is an arbitrary computable function, and $p$ a polynomial.
Any function $f$ as above is referred to as the \emph{size} of the
kernel. 
We frequently use the term ``kernel'' for the outputs $(|x'|,k')$ of the 
kernelization algorithm and, in case $f$ is a polynomial (resp. linear) function, we say that we have a {\em polynomial  (resp. linear) kernelization algorithm} or, 
simply, a {\em polynomial (resp. linear) kernel}.
\end{definition}

Clearly, if ${\sf P}={\sf NP}$, then all problems have constant size kernels so some kind of complexity theoretical hypothesis is needed to show that something does not have a 
small kernel.
The key thing that the reader should realize is that 
the reduction to the kernel is a \emph{polynomial time} reduction for 
both the problem \emph{and} parameter, and not an FPT algorithm.
Limiting results on kernelization 
mean that this often used practical method cannot be used to get FPT algorithms.
Sometimes showing that problems do not have small kernels (modulo 
some hypothesis) can be extracted
from work on approximation, since a small kernel is, itself, an approximation.
Thus, whilst we know of a $2k$ kernel for $k$-{\sc Vertex Cover}~\cite{NemT75vert}, we know that
unless ${\sf P}={\sf NP}$,
we cannot do better than size $1.36\cdot k$
using Dinur and Safra \cite{DinurS02thei}, 
since the PCP theorem 
provides a 1.36-lower bound on approximation (see also~\cite{GuoN07invi}).

To show that certain problems do not have 
polynomial kernels modulo a reasonable hypothesis,
we will need the following
definition, which is in classical complexity. 
(It is similar to another definition of Harnik and Noar \cite{HarnikN06onth}).

\begin{definition}[Or-Distillation]
\label{Definition: Distillation}%
An \emph{Or-distillation algorithm} for a classical problem $L \subseteq
\Sigma^*$ is an algorithm that
\begin{itemize}
\item receives as input a sequence
$(x_1, \ldots, x_t)$, with $x_i \in \Sigma^*$ for each $1\leq i\leq
t$,
\item uses time polynomial in $\sum_{i=1}^t |x_i|$,
\item and outputs a string $y \in \Sigma^*$ with
\begin{enumerate}
\item $y \in L \iff x_i \in L$ for some $1\leq i\leq t$.
\item $|y|$ is polynomial in $\max_{1\leq i\leq t} |x_i|$.
\end{enumerate}
\end{itemize}
\end{definition}

We can similarly define {\em And-distillation} by replacing the second last
item by ``$y \in L \iff x_i \in L$ for {\em all} $1\leq i\leq t$.''
Bodlaender, Downey, Fellows and Hermelin \cite{BodlaenderDFH09onpr}
showed that an {\sf NP}-complete problem has an Or-distillation
algorithm iff all of them do.
On general Kolmogorov complexity grounds, it seems very unlikely
that {\sf NP}-complete problems have either distillation algorithms.  
Following a suggestion of 
Bodlaender, Downey, Fellows and Hermelin,
Fortnow and Santhanam related Or-distillation to 
the polynomial time hierarchy as follows.

\begin{lemma}[\cite{FortnowS11infe}]
\label{Lemma: No distillation for NP_complete}%
If any {\sf NP}-complete problem has an Or-distillation algorithm then
{\sf co}-${\sf NP}\subseteq {\sf NP}\backslash${\sf Poly} and hence 
the polynomial time hierarchy 
collapses to 3 or fewer levels.
\end{lemma}

At the time of writing, there is no known 
version of  Lemma \ref{Lemma: No distillation for NP_complete}
for And-distillation and it remains an
important open question, whether {\sf NP}-complete problems
having And-distillation implies any classical collapse of the 
polynomial time hierarchy.
This material all relates to 
kernelization as follows.

\begin{definition}[Composition]
\label{Definition: Composition}%
An \emph{Or-composition algorithm} for a parameterized problem $L
\subseteq \Sigma^* \times \mathbb{N}$ is an algorithm that
\begin{itemize}
\item receives as input a sequence
$((x_1,k),\ldots,(x_t,k))$, with $(x_i,k) \in \Sigma^* \times
\mathbb{N}^+$ for each $1\leq i\leq t$,
\item uses time polynomial in $\sum_{i=1}^t |x_i|+k$,
\item and outputs $(y,k') \in \Sigma^* \times \mathbb{N}^+$ with
\begin{enumerate}
\item $(y,k') \in L \iff (x_i,k) \in L$ for some $1\leq i\leq t$.
\item $k'$ is polynomial in $k$.
\end{enumerate}
\end{itemize}
\end{definition}

Again we may similarly define {\em And-composition.} The key lemma relating the
two concepts is the following, which has an analogous 
statement for the And-distillation case:

\begin{lemma}[Bodlaender, Downey, Fellows and Hermelin \cite{BodlaenderDFH09onpr}]
\label{Lemma: Combining composition and kernelization}%
Let $L$ be a Or-compositional parameterized problem whose
unparameterized version $\widetilde{L}$ is {\sf NP}-comple\-te. If $L$ has a
polynomial kernel, then $\widetilde{L}$ is also Or-distillable.
\end{lemma}

Distillation
of one problem within another has also been
discussed in 
Chen, Flum and M{\"u}ller \cite{BodlaenderDFH09onpr}.

Using Lemma \ref{Lemma: Combining composition and kernelization},
 Bodlaender, Downey, Fellows and Hermelin \cite{BodlaenderDFH09onpr}
proved that a large class of graph-theoretic \FPT\ problems 
including $k$-{\sc Path}, $k$-{\sc Cycle}, various problems for 
graphs of bounded treewidth, etc., all have  no polynomial-sized 
kernels unless the polynomial-time hierarchy collapses to
three or fewer levels.

For And-composition, tying the distillation to something like 
the Fortnow-Santha\-nam material would be important since
it would say that important \FPT\ problems like
{\sc Treewidth} and {\sc Cutwidth} would likely not have polynomial size 
kernels, and would perhaps suggest why algorithms such 
as Bodlaender's linear time \FPT\ algorithm \cite{Bodlaender96alin}
(and other treewidth algorithms)
is so hard to run. 
Since the original  paper~\cite{BodlaenderDFH09onpr}, there 
have been a number of developments such as 
the Bodlaender, Thomass{\'e}, and Yeo \cite{BodlaenderTY08anal}
  use of reductions to extend the 
arguments above to wider classes of problems such as {\sc Disjoint Paths},
{\sc Disjoint Cycles}, and {\sc Hamilton Cycle Parameterized by Treewidth}.
In the same direction, Chen,  Flum and M{\"u}ller  \cite{ChenFM09lowe}
used this methodology to further explore the 
possible sizes of kernels.
 One fascinating application was by Fernau et. al. \cite{FernauFLRSV09kern}. 
They showed that
the following problem was in~\FPT, by showing a kernel of size $O(k^3)$.\\

\noindent{\sc Rooted $k$-OutLeaf Branching}\\
{\sl Instance:} A directed graph $D=(V,E)$ with exactly one vertex of indegree $0$
called the {\em root}.\\
{\sl Parameter:} A positive integer $k$.\\
{\sl Question:} Is there an oriented subtree $T$ of 
$D$ with exactly $k$ leaves spanning $D$?\\

However, for the {\em unrooted version}
they used the  machinery from ~\cite{BodlaenderDFH09onpr} to demonstrate that it 
has {\sl no} polynomial-size kernel unless some collapse occurs, but 
clearly by using $n$ independent versions of the rooted version 
the problem has a {\em polynomial-time Turing kernel}. We know of no method
of dealing with the non-existence of polynomial time Turing kernels and it seems an important programmatic question.

All of the material on lower bounds for Or distillation tend to 
use reductions to things like large disjoint unions of graphs having Or composition.
Thus, problems sensitive  to this seemed not readily approachable
using the machinery. Very recently, Stefan Kratsch \cite{Kratsch11cono} discovered a way around
this problem using the work Dell and van Melkebeek \cite{DellM10sati}.

\begin{theorem}[Kratsch \cite{Kratsch11cono}] The FPT problem $k$-{\sc Ramsey}
which asks if a graph $G$ has either an independent set or a clique of
size $k$,
has no polynomial kernel unless 
NP$\subseteq co-NP\backslash $poly.\end{theorem}

The point here is that a large 
collection of disjoint graphs would have a large independent set, hence
new ideas were definitely needed.
The ideas in Kratsch's proof are quite novel and use co-nondetermism
in compositions and communication complexity in a quite novel way.
Perhaps they might allow us to attack other such problems.

Finally, we remark in passing that as an important spinoff 
of the question from parameterized complexity, 
the Fortnow-Santhanam Lemma has been used by
Burhmann and Hitchcock \cite{BuhrmanH08npha} to show that unless the 
polynomial-time hierarchy collapses,
{\sf NP}-hard languages must  
be exponentially dense
(meaning that hard instances must occur very often), perhaps suggesting a connection between 
parameterized complexity and density of hard instances.

For further techniques on lower bounds on kernels, see \cite{DomLS09inco}.

\subsection{Bounded parameterized complexity}
\label{boundedx}
Another direction exploring the fine structure of \FPT\ was taken by 
Flum, Grohe and Weyer \cite{FlumGW06boun}
 who suggested that the important problems
were those where the constants were small. If $L$ is \FPT\
then the question ``$(x,k)\in L?$'' is decidable in time $f(k)|x|^c$, for some computable function $f$. But as we have seen
$f$ could be anything. Flum, Grohe and Weyer argue that the correct 
classes to focus upon 
are those with membership algorithms of the form $2^{o(k)}|x|^c$,
$2^{O(k)}|x|^c$,
and  $2^{\mbox{\footnotesize poly}(k)}|x|^c.$ These classes are called 
{\sf SUBEPT}, {\sf EPT} and {\sf E{\sf XP}T}, respectively. 
As a group they are referred to as {\em bounded} \FPT.
Interestingly, the reductions may be different for the different classes,
because the idea is to give the maximum power possible to the reductions and yet still
remain within the class. 
For putting this idea in a  general framework, we 
mention Downey, Flum, Grohe and Weyer \cite{FlumGW06boun}.

Of course, any reduction $(x,k)\mapsto (x',k')$ which is polynomial 
in both variables 
will keep the problem within the relevant class. 
As an example, one of the most useful reductions here is the
{\em {\sf EPT} reduction} which asks that there is a function $f\in 2^{O(k)}$
so that $x'$ is computable in time $f(k)|x|^c$ and 
there is a constant $d$ such that $k'\leq d(k+\log |x|)$. 
It is easy to see that {\sf EPT} is closed under {\sf EPT} reductions.

Concentrating on {\sf EPT} as a representative example,
the next step is to introduce classes akin to the ${\sf W}$-hierarchy 
for the bounded theory. For example, a {\em nondeterministic $(2^{O(k)},k)$-restricted 
Turing machine} is one for which there is a function
$f\in  2^{O(k)}$ that for each 
run on input $x$ the machine performs at most $f(k)\cdot |x|^{c_1}$ many steps 
such that at most $c_2(k+\log |x|)\cdot \log |x|$ are
nondeterministic. Using this as the core problem, and closing under {\sf EPT}
reductions defines the class ${\sf EW}[{\sf P}]$. More generally, 
it is possible to extend this definition to define another 
hierarchy
akin to the ${\sf W}$-hierarchy called the ${\sf EW}$-hierarchy. 
It is, of course, easy to show that the classes {\sf EPT}, {\sf E{\sf XP}T}, {\sf SUBEPT} are all
distinct by diagonalization. But what the new hierarchy allows 
for is to show that various problems which are all ${\sf W}[1]$-hard, say,
differ on the new hierarchy, aligning with our intuition that
they should be. As an illustration, there is a problem from 
computational learning theory called the {\sc Vapnik-Chervonenkis Dimension}
which is proven to be ${\sf W}[1]$-complete by combining 
Downey, Evans and Fellows \cite{DowneyEF93para} and Downey and Fellows \cite{DowneyF95survey}. The hardness  proof in~\cite{DowneyEF93para} 
used explicitly non-polynomial \FPT\ reductions. It turns out that 
{\sc Vapnik-Chervonenkis Dimension} is complete for ${\sf EW}[3]$,
and yet the problem {\sc Tournament Dominating Set} which is 
${\sf W}[2]$-complete is complete for the class ${\sf EW}[2]$.

Little is known here. It would be fascinating if a  central  \FPT\ problem
with infeasible algorithms at present,
such as {\sc Treewidth} could be shown to be complete for, say, ${\sf EW}[2]$,
which would suggest no possible {\sl reasonable} \FPT\ algorithm.
The current algorithms only put {\sc Treewidth} into {\sf EXPT}.

From \cite{FlumGW06boun} one example of this phenomenon is 
{\sc First-Order Model-Checking over Words}: it is in  \FPT\ but even {\sf EAW}[*]-hard (a bounded analog of ${\sf AW}$).

\subsection{Things left out} In such a short article, we do not really
have space to devote to the many areas of applications for this hardness
theory. Suffice to say, it has been applied in relational databases, 
phylogeny, linguistics, VLSI design, graph structure theory,
cognitive science,
Nash equilibria, voting schemes, operations research, etc. 
We can only point at the various survey articles, the books by Niedermeier's
\cite{Niedermeier06invi}, Fernau 
\cite{Fernau05para}, and  Flum and Grohe~\cite{FlumG06para}  
as well as the {\em Computer Journal} issues mentioned earlier~\cite{Downey08thec}.

Also we have not really had space to devote neither to many other
natural hierarchies based on various aspects of 
logical depth such as the ${\sf S}$, ${\sf W}^*$ and other hierarchies,
nor to issues like parameterized parallelism. Here we
refer the reader to Flum and Grohe \cite{FlumG06para}. Finally, we have not really developed 
all of the models which have been used.

What we hope the reader has gotten is the general 
flavor of the parameterized complexity  subject and some appreciation of the techniques.

%

\section{Parameterized algorithms}
\label{algorithms}

 The diversity of problems and areas where parameterized algorithms have been developed is enormous.
In this section we make an, unavoidably incomplete,  
presentation of the main techniques and results on the design of parameterized algorithms. To facilitate our description, we are mainly focusing on  problems on graphs. 

Some basic notational conventions on graphs follow: Given a graph, $G$ 
we denote the vertex and edge set of $G$ by $V(G)$ and $E(G)$ respectively. Given a vertex $v\in V(G)$ we denote the set of its neighbors by $N_{G}(v)$.  When a
graph $G$ is the  input of a parameterized problem, we always denote by $n$ the number of its vertices.

\subsection{De-nondeterminization}
\label{derand}

Parameterized algorithm design requires a  wider viewpoint that
the classic one of polynomial algorithm design. The reason is that 
we now permit time complexities that are super-polynomial. However, we have 
to make sure that this super-polynomial overhead depends 
only on the parameter. A wide family of techniques in  
parameterized algorithm design can be seen as 
ways to turn some polynomial non-deterministic algorithms 
to a deterministic one where the resulting super-polynomial
overhead is independent of  the main part of the problem.

\subsubsection{Bounded search trees}
\label{subs_boundedst}
The most ubiquitous  de-nondeterminization technique is the {\sl bounded search tree technique}.
We present it on one of the most extensively  studied problems in parameterized algorithms and complexity: \\

\noindent{\sc $k$-Vertex Cover}

\noindent {\sl Instance:}~~A graph $G$ and a non-negative  integer $k$.  \\
{\sl Parameter:} $k$.\\
{\sl Question:}~~Does $G$ have a vertex set $S$ of size at most $k$ that intersects all the edges of $G$?\\

This problem can be solved by the following non-deterministic algorithm:

\begin{tabbing}
{\bf 1.}  set $S\leftarrow \emptyset$ and $i\leftarrow k$,\\
{\bf 2.} while \= $E(G)\neq \emptyset$ and $i>1$, \\
\>  consider some edge $e=\{v,u\}$ of $G$,\\
\>  guess {\em non-deterministically} one, say $x$, of the two endpoints of $e$, \\
\> set $G\leftarrow G\setminus x$, $i\leftarrow i-1$\\
{\bf 3.}  if $E(G)=\emptyset$ then return  {\tt YES} \\
{\bf 4.}  If ${k}=0$, then return  {\tt NO}
\end{tabbing}

Clearly,  this algorithm is based on the fact that for each edge, one of its endpoints should be a vertex of every vertex cover. It makes at most $k$ non-deterministic choices each requiring $O(n)$ deterministic steps. This polynomial non-deterministic 
algorithm can be reverted to an (exponential) deterministic one  as follows.

\begin{tabbing}
{\sl Algorithm} {\sf {algvc}}$({G},{k})$\\
{\bf 1.} {\bf if} $|E({G})|=0$, then return {\tt YES}\\
{\bf 2.} {\bf if} ${k}=0$, then return  {\tt NO}\\
{\bf 3.} choose\= \  (arbitrarily)  an edge  $e=\{v,u\}\in E({G})$ and\\
\> {\bf return} {\sf {algvc}}$({G}-v,{k}-1)\bigvee${\sf algvc}$({G}-u,{k}-1)$ 
\end{tabbing}

Notice that the above algortihm makes $2$ recursive calls and the depth of the recursion is $\leq k$. Therefore it takes $O(2^{k}\cdot n)$ steps and is an {\sf FPT}-algorithm. This implies that 
{\sc $k$-Vertex Cover}$\ \in{\sf FPT}$. Notice that the algorithm is based on the transformation of a non-deterministic algorithm to a deterministic one in a way that the exponential blow-up (i.e., the size of the search tree) 
depends  exclusively  on the parameter $k$. This idea traces back to the paper of Buss and Goldsmith in~\cite{BussG93nond} and the corresponding technique is called the {\em bounded search tree technique}. 

Can we do better? A positive answer requires a better combination of non-deterministic guesses. As an example, instead of an edge one may pick a path {\sf P} with 
vertices $v_{1},v_{2},v_{3}$, and $v_{4}$. Then every  vertex cover of $G$ will contain some of the pairs  $\{v_{1},v_{2}\}$, $\{v_{2},v_{3}\}$,  $\{v_{2},v_{4}\}$.
That way, each recursive call has now $3$ calls but also guesses 2 vertices and therefore the depth of the recursion is at most $\lceil k/2\rceil$. In case $G$ does not contain
a path of length $3$, then $G$ is a forest of stars and in such a case the {\sc $k$-Vertex Cover}  can be solved in linear time. Summing all this together, 
we have a  $O(3^{k/2}\cdot n)$ step {\sf FPT}-algorihtm for the {\sc $k$-Vertex Cover} which improves the previous one, as $3^{1/2}<1.733$.
An even faster algorithm can be designed if we exploit  the fact  that the {\sc $k$-Vertex Cover} can be solved in linear time for graphs of maximum degree $2$.
Then, as long as there is a vertex $v$ with at least  3 neighbors, we know that a vertex cover should contain $v$ or all its neighbors. 
An elementary analysis implies that the size $T(k)$ of the search tree satisfies the relation $T(k)\leq T(k-1)+T(k-3)$. As the biggest root 
of the characteristic polynomial $a^{k}=a^{k-1}+a^{k-3}$ is less than $1.466$, we have an {\sf FPT}-algorithm for the the {\sc $k$-Vertex Cover} that runs in $O(1.466^{k}\cdot n)$ steps. 

Especially for  {\sc $k$-Vertex Cover}, there is a long sequence of improvements of the running time, based on even more refined search trees. 
The first non-trivial results dates back to the $O(1.3248^{k}\cdot n)$ step algorithm of Balasubramanian,  Fellows, and Raman~\cite{BalasubramanianFR98anim}. This result was improved in~\cite{ChenKJ01vert, NiedermeierR99uppe} and, 
currently,  the fastest parameterized algorithm for  {\sc $k$-Vertex Cover} runs in $O(1.2738^k + kn)$ steps  by Chen, Kanj, and Xia~\cite{ChenKX10impr}. For applications of the same technique on restricted versions of the same problem, see~\cite{ChenLJ00impr,NiedermeierR03onef}.

We stress that the bounded search tree technique is strongly linked  to 
the design of exact algorithms, as {\sf FPT}-algorithms can be seen as exact algorithms where the parameter is not any more restricted. 
For instance, the $O(1.2738^k + kn)$ step algorithm of  \cite{ChenKX10impr} implies an $O^{*}(1.2738^n)$ step
exact algorithm for 
{\sc Vertex Cover}. However, the existence of fast exact algorithms does not imply directly the existence of a parameterized one with the 
same parameter dependence. For example {\sc Vertex Cover} can be solved in $O^{*}(1.189^{n})$ steps by an algorithm of Robson~\cite{Robson86algo} (see also~\cite{FominGK06meas} for a simpler one running in $O({1.221}^{k})$ steps).\\

The bounded search tree technique is an unavoidable ingredient of most parameterized algorithms. Typical examples of problems where this technique has been applied are {\sc Feedback Vertex Set}  \cite{ChenFLSV08impr} ($O(5^{k}\cdot kn^{2})$ steps -- see also~\cite{DehneFLRS07anft,GuoGHNW06comp}), {\sc Cluster Editing} \cite{BockerBBT08goin} ($O((1.82)^{k}+n^{3})$ steps -- see also~\cite{GrammGHN04auto}), {\sc Dominating Set}~\cite{AlonG09line}, and partial covering problems \cite{AminiFS08impl}.
%
%
\subsubsection{Greedy localization}
%

A step further in  the branching approach is to combine it with some non-determini\-stic guess of some part of the solution. This idea is known as the   {\em  greedy localization technique} and  was introduced in~\cite{ChenFJK04usin} and~\cite{JiaZC04anef}.
We will briefly present the idea using the following problem.\\

\noindent{\sc $k$-Triangle Packing}

\noindent {\sl Instance:}~~A graph $G$ and a non-negative  integer $k$.  \\
{\sl Parameter:} $k$.\\
{\sl Question:}~~Does $G$ have at least $k$ mutually vertex disjoint triangles?\\

Again, we describe the {\sf FPT}-algorithm for this problem in its non-deterministic version.  We use the term {\em partial triangle}
for any vertex or edge of $G$ and we will treat it as a potential part of a triangle in the requested triangle packing. The first step is to find in $G$ a maximal set of disjoint triangles ${\cal T}=\{T_{1},\ldots,T_{r}\}$. This greedy step justifies the name of the technique and can be done in polynomial time.  If $r\geq k$ then return {\tt YES} and stop. If not, let $G_{{\cal T}}$ be the graph formed by the disjoint union of the triangles in ${\cal T}$ and and 
observe that $|V(G_{{\cal T}})|\leq 3(k-1)$. The key observation is that if there exists a solution ${\cal T'}=\{T_{1}',\ldots,T_{r}'\}$ to the problem ($r'\geq k$), then each $T_{i}'$ should intersect some triangle in ${\cal T}$, (because of the  maximality of the choice of ${\cal T}$). These intersections define a partial solution ${\cal P}=\{A_{1},\ldots,A_{p}\}$ of the problem. The next step of the algorithm is to {\sl guess} them: let $S$ be a subset of $V(G_{{\cal T}})$ such that $G_{{\cal T}}[S]$ has at least $k$ connected components (if such a set does not exist, the algorithm returns {\tt NO} and stops). Let ${\cal P}=\{P_{1},\ldots,P_{p}\}$  be these connected components and observe that each $P_{i}$ is a partial triangle of $G$. 
We also denote by $V({\cal P})$ the set of vertices in the elements of ${\cal P}$.
Then, the algorithm intends to 
extend this partial solution to a full one using the following non-deterministic procedure {\sf branch}$({\cal P})$:

\begin{tabbing}
{\sl Procedure} {\sf {branch}}$(P_{1},\ldots,P_{p})$\\
{\bf 1.} $A\leftarrow B\leftarrow \emptyset$, $i\leftarrow 1$\\
{\bf 2.} while \ \= $i\leq p$\\
\> {\bf if} there exists \=a partial triangle $B$ of $G\setminus (V({\cal P})\cup A)$\\
\> \>  such that $G[P_{i}\cup B]$ is a triangle, \\
\>{\bf then} $A\leftarrow A\cup B$, $i\leftarrow i+1$\\
\>{\bf otherwise} goto step {\bf 4}\\
{\bf 3.} {\bf return} {\tt YES} and stop.\\
{\bf 4.} {\bf let }\=$A'$ be the set containing each vertex  $v\in A$ \\
\>  such that $P_{i}\cup \{v\}$ is a (partial) triangle.\\
{\bf 5.} {\bf if} $A'=\emptyset$, \={\bf then} \={\bf return} {\tt NO} and stop,\\
   \> {\bf otherwise} {\sl guess} $v\in A'$  and\\
\> \>  {\bf return} {\sf {branch}}$(P_{1},\ldots,P_{i}\cup\{v\}\ldots,P_{p})$.
\end{tabbing}

\noindent In Step {\bf 2}, the procedure checks greedily for a possible extension of ${\cal P}$ to a 
full triangle packing. If it succeeds, it enters Step {\bf 3} and returns {\tt YES}.
If it fails at the $i$-th iteration of Step {\bf 2}, this means that one of the vertices used for the 
current extension (these are the vertices in $A'$ constructed in Step {\bf 4}) should belong to the extension of $P_{i}$ (if such a solution exists). So we further {\sl guess}
which vertex of $A$ should be added to $P_{i}$ and we recurse with this new collection of partial triangles
(if such a vertex does not exist, return a negative answer). This algorithm has two non deterministic steps: one is the initial choice of the set $S$ and the other is the one in Step {\bf 5} of {\sf {branch}}$(P_{1},\ldots,P_{p})$. The first choice implies that  {\sf {branch}}$(P_{1},\ldots,P_{p})$ is called $k^{O(k)}$ times (recall that $|V(G_{{\cal T}})|=O(k)$). Each call of {\sf {branch}}$(P_{1},\ldots,P_{p})$ has $|A'|\leq 2k$ recursive calls and the depth of the recursion is upper bounded by $2k$, therefore the size of the search tree is bounded by $k^{O(k)}$. Summing up, we have that {\sc $k$-Triangle Packing} is in  $2^{O(k\log k)}\cdot n$-{\sf FPT}. 
 
The greedy localization technique has been used to solve a wide variety of problems. 
It was applied for the standard parameterizations of problems such as  $r$-{\sc Dimensional Matching} in~\cite{ChenFJK04usin}, $r$-{\sc  Set Packing} in~\cite{JiaZC04anef}, $H$-{\sc Graph Packing} in~\cite{FellowsHRST04find}. Also, similar ideas, combined with other techniques, have been used for
various problems such as  the (long-standing open)  {\sc Directed Feedback Vertex Set}~\cite{ChenLOR08afix} by and others such as the  \textsc{Bounded Length Edge Disjoint Paths} in \cite{GolovachT11path}. The main drawback of this method is that the parameter dependence of the derived {\sf FPT}-algorithm is typically $2^{O(k\log k)}$.
\medskip

\subsubsection{Finite automata and sets of characteristics}
\label{fasoc}

We proceed now with a more elaborate use of nondeterminism in the design of parameterized algorithms. 
The problem that will drive our presentation is the following:\\

\noindent{\sc $k$-Cutwidth}

\noindent {\sl Instance:}~~A graph $G$ and a non-negative  integer $k$.  \\
{\sl Parameter:} $k$.\\
{\sl Question:}~~Does $G$ have cutwidth at most $k$? i.e., is there a vertex layout $L=\{v_{1},\ldots,v_{n}\}$ of $G$ such that for every $i\in\{0,\ldots,n\}$
the number of edges of $G$ with one endpoint in $\{v_{1},\ldots,v_{i}\}$ and one in $\{v_{i+1},\ldots,v_{n}\}$ is at most $k$?\\

In what follows we denote the cutwidth of a graph $G$ by $\cw(G)$.
This problem belongs to the category of {\em layout problems} and the technique we present applies to many of them.
The standard approach uses the parameter of pathwidth defined as follows.
Given a graph $G$, we say that a sequence ${\cal X}=(X_{1},\ldots,X_{r})$ of vertex sets from $G$ is a {\em path decomposition} of $G$ if 
every vertex or edge of $G$ belongs entirely in some $X_{i}$ and if for every vertex $v\in G$ the set $\{i \mid v\in X_{i}\}$ 
  is a set of consecutive 
indices from $\{1,\ldots,r\}$. The {\em width} of a path-decomposition is equal to $\max\{|X_{i}|-1\mid i=1,\ldots,r\}$.
The {\em pathwidth} of a graph is the minimum $k$ for which $G$ has a path-decomposition of width at most $k$.
It follows from the results in~\cite{BodlaenderT04comp,BodlaenderK96effi,Bodlaender96alin} that the following problem is in $f(k)\cdot n$-\FPT:\\

\noindent{\sc $k$-Pathwidth}\\
\noindent {\sl Instance:}~~A graph $G$ and a non-negative  integer $k$.  \\
{\sl Parameter:} $k$.\\
{\sl Question:}~~Does $G$ have pathwidth at most $k$?\\

\noindent Moreover, the aforementioned  \FPT-algorithm for  {\sc $k$-Pathwidth}  is able to construct the corresponding path-decomposition, in case of a positive answer. 

The design of an \FPT-algorithm for {\sc $k$-Cutwidth} will be reduced 
to the design of an \FPT-algorithm for the following problem.\\

\noindent{\sc $pw$-$k$-Cutwidth}

\noindent {\sl Instance:}~~A graph $G$ of pathwidth $\leq l$ and a non-negative  integer $k$.  \\
{\sl Parameter:} $k+l$.\\
{\sl Question:}~~Does $G$ have  cutwidth at most $k$?\\

An \FPT-algorithm for {\sc $pw$-$k$-Cutwidth} is based on  the construction, for any $k,l\geq 0$, of
a {\sl non-deterministic finite state automaton} that receives $G$ as a word encoded in terms of a path decomposition
of width at most $l$ and  decides whether a graph $G$ has cutwidth at most 
$k$.  The number of states of this automaton depends exclusively on the parameters $k$ and $l$ and, of course,  this is also the 
case when we revert it to a deterministic one. Before we proceed with the definition of the automaton we need some more definitions. 


A path decomposition is {\em nice} if $|X_{1}|=1$, and for every $i\in\{2,\ldots,r\}$, the symmetric difference of 
$X_{i-1}$ and $X_{i}$ contains only one vertex. Moreover, if this vertex belongs to $X_{i}$ we call it {\em insert vertex} for $X_{i}$ and if it 
belongs to $X_{i-1}$ we call it {\em forget vertex} for $X_{i}$. 
Given a nice path decomposition ${\cal X}=(X_{1},\ldots,X_{r})$ of $G$ of width at most $l$, it is easy to  construct an $(l+1)$-coloring  $\chi: V \rightarrow \{{\bf 1},\ldots, {\bf l+1}\}$ of
 $G$ such that vertices in the same $X_{i}$ have always distinct colors (we refer to these colors using bold characters).
If $X_{i}$ has an introduce vertex then set $p(i):={\tt ins}({\bf t},{\bf S})$ where 
$\{{\bf t}\}=\chi(X_{i}-X_{i-1})$ and ${\bf S}=\chi(N_{G[X_{i-1}]}(v))$ (if $i=1$ then ${\bf S}=\emptyset$). 
If $X_{i}$ has a forget vertex then set $p(i):={\tt del}({\bf t})$ where 
$\{{\bf t}\}=\chi(X_{i-1}-X_{i})$.
We encode ${\cal X}$ as a string $w({\cal X})$  whose $i$-th letter is $p(i)$, $i=1,\ldots,r$.

Let $\Sigma$ be an alphabet consisting of all possible ${\tt ins}(v,S)$ and ${\tt del}(v)$ (notice that the size of $\Sigma$ depends only on $l$) and denote $K=\{0,\ldots,k\}$. We also set up a set of labels ${\bf T}=\{\mbox{\bf --},{\bf 1},\ldots,{\bf
l+1}\}$ where bold numbers represent colors and ``\mbox{\bf --}''
represents the absence of a vertex.

We define the automaton $A_{G,l,k}=(Q,\Sigma,\delta,q_{s},F)$ that receives $w({\cal X})$  as input (provided that  a path decomposition ${\cal X}$ of $G$
is given). We set  $$Q=\{ w\mid w\in K({\bf T}K)^{*}, |w|\leq 2n+1\mbox{, \small  and each label of
${\bf T}$ appears at most once in $w\}$},$$

$q_{s}=[0]$, $F=\{w\in Q\mid
w=2n+1\}$, and for any $q=n_{0}{\bf t}_{1}n_{1}{\bf t}_{2}\cdots{\bf
t}_{r}n_{r}\in Q$ and ${\tt ins}({\bf t},{\bf S})\in\Sigma$ or ${\tt
del}({\bf t})\in \Sigma$, the value of the transition function $\delta$ is defined as follows:
\begin{eqnarray}
\delta(q,{\tt ins}({\bf t},{\bf S})) & = & \{q'\mid q'=n_{0}'{\bf t}_{1}'n_{1}'{\bf t}_{2}'\cdots {\bf t}_{r+1}'n_{r+1}'
\mbox{~where~} q'\in Q\nonumber \\
 \mbox{~and~}  \exists_{i,0\leq i\leq r}:                          &     & \hspace{-.7cm}\forall_{h=1,\ldots,i}{\bf t}'_{h}={\bf t}_{h}\wedge {\bf t}_{i+1}'={\bf t} \wedge \forall_{h=i+1,\ldots,i}{\bf t}'_{h+1}={\bf t}_{h}\nonumber \wedge\\
                           & & \hspace{-.7cm} \forall_{h=1,\ldots,i} n'_{h}=n_{h}+|\{{\bf t}_{j}\mid {\bf t}_{j}\in {\bf S} \wedge j\leq h\}|\wedge\nonumber \\     
                           & & \hspace{-.7cm} \forall_{h=j+1,\ldots,r+1} n'_{h}=n_{h+1}+|\{{\bf t}_{j}\mid {\bf t}_{j}\in {\bf S} \wedge j\geq h\}|\}\nonumber 
\end{eqnarray}
and 
\begin{eqnarray}
\delta(C,{\tt del}({\bf t})) & = & \{C'\mid C'=n_{0}{\bf t}_{1}n_{1}{\bf t}_{2}\cdots n_{i-1}\mbox{\bf --}n_{i} \cdots{\bf t}_{r+1}n_{r+1}\nonumber \\
& & \mbox{~where~} {\bf t}={\bf t}_{i}\}\nonumber 
\end{eqnarray}
In the definition of $\delta(q,{\tt ins}({\bf t},{\bf S}))$,  $A_{G,l,k}$ {\sl guesses} the position $i$ where the vertex colored by ${\bf t}$
is being inserted in the already guessed layout.  Bodlaender, Fellows, and  Thilikos proved the following in~\cite{BodlaenderFT09deri}.

 \begin{lemma}
Let $G=(V,E)$ be a graph and ${\cal X}$ be a path decomposition of $G$ 
of width $\leq l$.  Then $\cw(G)\leq k$ iff  $w({\cal X})\in {\sf L}(A_{G,l,k})$.
 \end{lemma}

In the above lemma, ${\sf L}(A)$ is the language recognized by the automaton $A$.
Notice that in $A_{G,l,k}$ the definition of $\delta$ does not require any 
knowledge of $G$. This is because the adjacency information  codified in the string $w({\cal X})$
is enough for the definition of $\delta$. However, 
the number  $|Q|$ of the states still depends on the size of $G$.
To explain how this problem can be overcome
we  give an example.

Suppose we have a substring ${\bf 5}4\mbox{\bf --}7\mbox{\bf --}9{\bf 2}$ in some state of $Q$.
Then notice the following: If the automaton accepts the string $w({\cal X})$ and during its
operation enters this state and proceeds by ``spliting'' the number $7$ then 
it will also accept the same string if it ``splits'' instead the number $4$.
This essentially means that it is not a problem if we ``forget'' 7 from this state.
As a consequence of this observation, we can reduce the length of the strings in
$Q$ by suitably ``compressing'' them (see~\cite{BodlaenderK96effi, BodlaenderT97cons,BodlaenderT00cons, ThilikosSB00cons,ThilikosSB05cutw1}). 
To explain this we first need some definitions.

Let $w\in Q$. We say that a $z$ is a {\em portion} of $w$ if it is a maximal 
substring of $w$ that does not contain symbols from ${\bf T}$.
We say that a portion of $w$ is  {\em compressed} if it does 
not contain a sub-sequence $n_{1}\mbox{\bf --}n_{2}\mbox{\bf --}n_{3}$ such that
either $n_{1}\leq n_{2}\leq n_{3}$ or $n_{1}\geq n_{2}\geq n_{3}$.
The operation of replacing in a portion 
any such a subsequence by $n_{1}\mbox{\bf --}n_{3}$ is called {\em compression} of the portion.
We also call {\em compression of} $w\in Q$ the string that appears if we compress all 
portions of $w$. We define $\tilde{Q}$ as the set occurring from $Q$ if we replace 
each of its strings by their compressions. This replacement naturally defines 
an equivalence class in $Q$ where two strings are equivalent if they have the same 
compression. That, way, $\tilde{Q}$ defines a ``{\sl set of characteristics}'' of $Q$ in the sense 
that it contains all the useful information that an automaton needs in order to operate. The good news is that
the size of $\tilde{Q}$ is now depending {\sl only} on $l$ and $k$. Indeed, in~\cite{BodlaenderK96effi,ThilikosSB05cutw1} it was proved 
that $|\tilde{Q}|\leq (l+1)!\frac{8}{3}2^{2k}$ and this makes it possible to
construct an automaton that does not depend on the input graph $G$.
In particular define the non-deterministic finite state automaton
$\tilde{A}_{l,k}=(\tilde{Q},\Sigma,\tilde{\delta},q_{s},\tilde{Q})$. Here
$\tilde{\delta}=\beta\circ\delta$ where $\beta$ is the function that receives
a set of members of $Q$ and outputs the set of all their
compressions. The following holds:

 \begin{lemma}
Let $G=(V,E)$ be a graph and ${\cal X}$ be a path decomposition of $G$ 
of width $\leq l$.  Then $\cw(G)\leq k$ iff  $w({\cal X})\in {\sf L}(\tilde{A}_{l,k})$.
 \end{lemma}

As mentioned, the construction of the automaton $\tilde{A}_{l,k}$ 
depends  only on $l$ and $k$. Moreover,
as the input is a string of length $O(n)$ the decision can be made non-deterministically 
in linear time. As for every non-deterministic finite state automaton one can construct an equivalent 
deterministic one  (notice that the state explosion will be an exponential function on $k$ and $l$), we deduce the following:

\begin{theorem}
\label{pw-cut-check}
{\sc $pw$-$k$-Cutwidth}$\ \in f(k)\cdot n$-\FPT.
\end{theorem}

In~\cite{BodlaenderFT09deri} Bodlaender, Fellows, and  Thilikos described how to turn the decision algorithm to an algorithm that, in
case of a positive answer, also outputs the corresponding layout.

An {\sf FPT}-algorithm 
for {\sc $k$-Cutwidth} follows easily from Theorem~\ref{pw-cut-check} by taking into account that 
if the input graph $G$  has cutwidth at most $k$ then  it also has a path decomposition of width at most $k$ (see~\cite{ThilikosSB05cutw1}).
Recall that {\sc $k$-pathwidth} belongs to $f(k)\cdot n$-\FPT \cite{BodlaenderT04comp}. Using this algorithm,  we check 
whether the pathwidth of $G$ is at most $k$ and, in case of a negative answer, we report a negative answer for cutwidth as well. Otherwise, the same algorithm outputs a path decomposition of $G$ of  width at most $k$. Then we may solve  {\sc $k$-Cutwidth} by  applying the {\sf FPT}-algorithm of Theorem~\ref{pw-cut-check}. 

The above technique is known as the {\em characteristic sequences technique} and appeared for the 
first time in~\cite{BodlaenderK96effi} (for pathwidth and treewidth) and has further been developed in~\cite{BodlaenderK96effi, BodlaenderT97cons,BodlaenderT00cons, ThilikosSB00cons,ThilikosSB05cutw1,ThilikosSB05cutw2} (for linear-width, branchwidth, cutwidth, and carving-width). Its automata interpretation appeared in~\cite{BodlaenderFT09deri} and derived {\sf FPT}-algoroithms for weighted and directed variants of pathwidth,  cutwidth, and modified-cutwidth (see also~\cite{SernaT05para} for a survey).  Finally, an extension of the same technique for more general parameters appeared by Berthomé and Nisse in~\cite{BerthomeN08auni}.


\subsection{Meta-algorithmic techniques}
\label{matech}

By the term {\em algorithmic meta-theorem} we mean a theorem than provides generic conditions concerning the problem itself that enable us to guarantee the existence (constructive or not) of an algorithm. In a sense, (constructive) meta-algorithmic results provide algorithms that receive as input the problem and output an algorithm to solve it. Such results usually 
require a combinatorial condition on the inputs of the problem and a logical 
condition on the statement of the problem. 
This is a ``{\sl dream ideal algorithm-design technique}'' since we automatically obtain an algorithm  
given just the specification of the problem.

We remark that these methods are so general that, in most of the cases,
 they not produce very good algorithms, especially in terms of the constants.
As we mentioned in Subsection~\ref{ocapp}, sometimes this 
impracticality is provable. Paradoxically, this 
impracticality is often proved using parameterized complexity 
classes like $W[1]$. 
However, at least for Courcelle's Theorem, 
the methods work well for problems of small quantifier depth,
and the principle impediment here is a really good algorithm for 
finding tree (or other) decompositions. 
We will also discuss efforts towards speeding up those algorithms 
 in Subsection~\ref{ffpta}.

We also remark that the ideas have been applied beyond graphs, for example to 
\emph{\sl matrices} and \emph{\sl matroids} via the work of, for example, Petr Hliňený \cite{Hlineny06}
or Geelen, Gerards, and Whittle \cite{GeelenGW06}.
We do not have the space or expertise to present their ideas, and refer
the reader to their work for more details.

\subsubsection{Courcelle's theorem}
\label{cource}

The logical condition is related to the Monadic Second Order Logic on graphs (MSOL). Its syntax requires an infinite supply of individual variables, usually denoted by lowercase letters x,y,z and an infinite supply of set variables, usually denoted by uppercase letters X,Y,Z. 

Monadic second-order formulas in the language of graphs are built up from 
\begin{itemize}
\item atomic formulas $E(x,y)$, $x=y$ and $X(x)$ (for set variables $X$ and individual variables $x$) by using the usual Boolean connectives $\neg$ (negation), $\land$ (conjuction), $\lor$(disjunction), $\rightarrow$ (implication), and $\leftrightarrow$ (bi-implication)  and
\item existential quantification $\exists x, \exists X$ and universal quantification $\forall x, \forall X$ over individual variables and set variables.
\end{itemize}

Individual variables range over vertices of a graph and set variables are interpreted by sets of vertices. The atomic formula $E(x,y)$ express adjacency, the formula $x=y$ expresses equality  and $X(x)$ means that the vertex $x$ is contained in the set $X$. The semantics of the monadic-second order logic  are defined in the obvious way. 
As an example of an MSOL property on graphs, we give the following formula 
\begin{eqnarray*}
\phi & =& \exists {R}\ \exists {Y} \ \exists {B}\ (\forall x\  (({R(x)}\vee Y(x)\vee {B(x)})\wedge\\
& & \neg ({R(x)}\wedge {Y(x)}) \wedge \neg ({B(x)}\wedge {Y(x)}) \wedge \neg ({R(x)}\wedge {B(x)})))\\
& & \wedge \neg (\exists x \ \exists y\  {(} E(x,y)\ \wedge\\
& &  {(}({R(x)}\wedge {R(y)})\vee ({Y(x)}\wedge {Y(y)})\vee({B(x)}\wedge {B(y)}){)}{)})
\end{eqnarray*}
and observe that $G\models \phi$ ($G$ {\em models} $\phi$) if and only if $G$ is 3-colorable (the two first lines ask for a partition of the vertices of $G$ into three sets $R$, $Y$, and $B$ and the two last lines ask for the non-existence of an edge between vertices of the same part of this partition).

The second (combinatorial) ingredient we need is a measure of the topological 
resemblance of a graph to the structure of a tree.

A {\em branch decomposition} of a graph $G$ is a  pair $(T,\tau)$, where $T$ is a ternary tree 
(a tree with vertices of degree one or three) and $\tau$ is a bijection from
$E({G})$ to the set of leaves of $T$. 
Clearly, every edge $e$ of $T$ defines a partition $(\tau^{-1}(L_{e}^{1}),\tau^{-1}(L_{e}^{1}))$ 
of $E(G)$ where, for $i=1,2$,  $L_{e}^{i}$ contains the leaves 
of the connected components of $T\setminus e$. 
The {\em middle set
function} $\mids: E(T)\rightarrow 2^{V(G)}$ of a branch
decomposition maps every edge $e$ of $T$ to the set  (called {\em middle set}) containing every vertex of $G$
that is incident {\sl both} to some edge in $\tau^{-1}(L_{e}^{1})$ and to some edge in $\tau^{-1}(L_{e}^{2})$.
The {\em branchwidth} of $(T,\tau)$ is equal to $\max_{e\in E(T)}|{\bf mid}(e)|$
and the {\em branchwidth} of ${G}$, $\bw({G})$, is
the minimum width over all branch decompositions of ${G}$.

Now we are ready to define the following parameterized meta-problem.\\

\noindent{\sc $bw$-$\phi$-Model Checking for $\phi$}\\
\noindent {\sl Instance:}~~A graph $G$.  \\
{\sl Parameter:} $\bw(G)+|\phi|$.\\
{\sl Question:}~~$G\models \phi$?\\

At this point we must stress that branchwidth can be interchanged with its  ``twin''
parameter of treewidth,  $\tw(G)$, for  most of their applications in parameterized algorithm design.
The reason  is that for every graph $G$, it holds that $\bw(G)\leq \tw(G)+1\leq \frac{3}{2}\bw(G)$, where 
by $\tw(G)$ we denote the treewidth of $G$~\cite{RobertsonS91-X}. 
Therefore, there is not much difference  in the {\sc $bw$-$\phi$-Model Checking for $\phi$} problem if we use treewidth instead of branchwidth

We can now present the following meta-algorithmic result proved by Courcelle.

\begin{theorem}[Courcelle's Theorem~\cite{Courcelle90them}]
\label{courcelle:th}
If $\phi$ is an  MSOL formula, then  {\sc $bw$-$\phi$-Model Checking for $\phi$} belongs to $O(f(k,|\phi|)\cdot n)$-\FPT.
\end{theorem}

%
A proof of the above theorem can be found in~\cite[Chapter 6.5]{DowneyF99para} and 
\cite[Chapter 10]{FlumG06para} and similar results 
appeared  by Arnborg, Lagergren, and Seese in~\cite{ArnborgLS91easy} and Borie, Parker, and Tovey  in~\cite{BoriePT92auto}.
 Also, an alternative game theoretic-proof has appeared recently in~\cite{KneisL09apra,KneisLR10cour}.
In fact, Theorem~\ref{courcelle:th} holds also for the more general monadic second order logic, namely MSOL$_{2}$, where 
apart from quantification over vertex sets we also permit quantification over edge sets.
It also has  numerous generalizations for other extensions of MSOL (see~\cite{Courcelle92them,CourcelleM93mona,Courcelle97thee}).
Theorem~\ref{courcelle:th} and its extensions can be applied to a very wide family of problems due to the 
expressibility  power of MSOL. 
As a general principle, ``most normal'' problems in  ${\sf NP}$  can be expressed  in MSOL$_{2}$, and hence 
they are classified in  {\sf FPT} for graphs of bounded width.

\emph{\sl However,}
 there are still graph problems that remain parameterized-intractable even when 
parameterized in terms of the branchwidth of their inputs. As examples, we mention {\sc List Coloring}, 
{\sc  Equitable Coloring}, {\sc Precoloring Extension}~\cite{FellowsFLRSST11othe},  and  {\sc Bounded Length Edge-Disjoint Paths}~\cite{GolovachT09para}.

In~\cite{CourcelleMR00line}, Courcelle, Makowski, and Rotics extended Theorem~\ref{courcelle:th} for the more general parameter of rankwidth (see also~\cite{LangerRS11line}).
The definition of rankwidth is quite similar to the one of branch-width: a {\em rank decomposition} of a graph $G$ is again
a  pair $(T,\tau)$ with the difference that now $\tau$ is a bijection from
$V({G})$ to the set of leaves of $T$. For each edge $e$ of $T$, the partition $(V_{e}^{1},V_{e}^{2})=(\tau^{-1}(L_{e}^{1}),\tau^{-1}(L_{e}^{1}))$
is a partition of $V(G)$ which, in turn, defines a $|V_{e}^{1}|\times |V_{e}^{2}|$ matrix $M_{e}$  over the field $GF(2)$ with entries $a_{x,y}$ for $(x,y)\in V_{e}^{1}\times V_{e}^{2}$ where $a_{x,y}=1$ if $\{x,y\}\in E(G)$, otherwise $a_{x,y}=0$. The {\em rank-order} $\rho(e)$ of $e$ is defined as the row rank of $M_{e}$. Similarly to the definition of branch-width,  the {\em rankwidth} of $(T,\tau)$ is equal to $\max_{e\in E(T)}|\rho(e)|$
and the {\em rankwidth} of ${G}$, $\rw({G})$, is
the minimum width over all rank decompositions of ${G}$. 
Rank-width is a more general  parameter than branchwidth in the sense that there exists a function $f$ were 
$\bw(G)\leq f(\rw(G))$ for every graph $G$, as proved by Oum and Seymour in~\cite{OumS06appr}.

Consider now the following parameterized meta-problem.\\

\noindent{\sc $rw$-$\phi$-Model Checking for $\phi$}\\
\noindent {\sl Instance:}~~A graph $G$.  \\
{\sl Parameter:} $\rw(G)+|\phi|$.\\
{\sl Question:}~~$G\models \phi$?\\

We can now present
the following extension of Theorem~\ref{courcelle:th}.

\begin{theorem}[Courcelle, Makowski, and Rotics~\cite{CourcelleMR00line}]
\label{ext:more}
If $\phi$ is an  MSOL formula, then  {\sc $rw$-$\phi$-Model Checking for $\phi$} belongs to $O(f(k,|\phi|)\cdot n)$-\FPT.
\end{theorem}

Actually the original statement of Theorem~\ref{ext:more} used 
yet another width metric, called cliquewidth,
 instead of rankwidth. However, we
avoid the definition of cliquewidth and we use rankwidth instead, in order to maintan uniformity with the definition of branchwidth. In what concerns the statement of the theorem, this makes no difference as rankwidth and cliquewidth are parametrically equivalent: as proved in~\cite{OumS06appr}, for every graph $G$, $\rw(G)\leq \clw(G)\leq 2^{\rw(G)}-1$ (where $\clw(G)$ is the cliquewidth of $G$).

It is interesting to observe that Theorem~\ref{ext:more} is not expected to hold for MSOL$_{2}$ formulas. Indeed,  Fomin,  Golovach,  Lokshtanov, and Saurabh,  proved in~\cite{FedorGLS09cliq} that there are problems,
such as {\sc Graph Coloring}, {\sc Edge Dominating Set}, and {\sc Hamiltonian
Cycle} that can be expressed in MSOL$_{2}$ but not in MSOL, that are {\sf {\sf W}[1]}-hard for graphs of bounded cliquewidth (or rankwidth).

\subsubsection{{\sf FPT} and {\rm FOL}}
\label{ftpfol}

While the parameterized tractability horizon of MSOL-expressible problems is restricted to 
graphs of bounded branchwidth, one can do much 
more for problems that can be expressed in First Order Logic (FOL). 
First Order formulas  are defined similarly to MSOL formulas with the (essential) restriction that 
no quantification on sets of variables is permitted any more. A typical example of a parameterized FOL-formula 
is the following:
\begin{eqnarray*}
\phi_{k} & =& \exists_{x_{1}}\cdots\exists_{x_{k}}(\bigwedge_{1\leq i<j\leq k} x_{i}\neq x_{j}\wedge\forall y\bigvee_{i\in\{1,\ldots,k\}}(y=x_{i}\vee E(y,x_{i})))
\end{eqnarray*}
The formula $\phi_{k}$ expresses the fact that $G$ has a dominating set of size  $k$.\\

Let ${\cal C}$ be a class of graphs. We consider the following parameterized meta-problem.\\

\noindent{\sc $\phi$-Model Checking for FOL  in ${\cal C}$}

\noindent {\sl Instance:}~~A graph $G\in{\cal C}$ and a FOL-formula $\phi$.  \\
{\sl Parameter:} $|\phi|$.\\
{\sl Question:}~~$G\models \phi$?\\

A  graph $H$ is a {\em minor} of a graph $G$, we denote it by $H\leq_{\rm m} G$,  if there exists a partial function $\sigma: V(G)\rightarrow V(H)$ that is surjective and such that 
\begin{itemize}
\item [a)] for each $a\in V(H)$, $G[\sigma^{-1}(a)]$ is connected and
\item [b)] for each edge $\{a,b\}\in E(H)$ 
there exists an edge $\{x,y\}\in E(G)$ such that $x\in \sigma^{-1}(a)$ and $y\in\sigma^{-1}(b)$.
\end{itemize}
Given a graph $H$ we say that a graph class ${\cal G}$ is {\em $H$-minor free} if none of the graphs in ${\cal G}$ contain $H$ as a minor.  Intuitively, $H$ is a minor of $G$ if $H$ can be obtained from 
a subgraph of $G$ after applying a (possibly empty) sequence of edge contractions.

If in  the definition of the minor relation we further restrict each $G[\sigma^{-1}(a)]$ to have radius at most $r$, for some  $r\in \Bbb{N}$, then we say that
$H$ is  a {\em $r$-shallow minor} of $G$ and we denote it by $H\leq_{{\rm m}}^{r} G$.

The {\em grad} (greatest reduced average density) {\em of} {\em rank} $r$ of 
a graph $G$ is equal to the maximum average density over all $r$-shallow minors of $G$
and is denoted by $\nabla_{r}(G)$. We say that a graph class ${\cal G}$
has {\em bounded expansion} if there exists a function $f:\Bbb{N}\rightarrow \Bbb{R}^{+}$ such that 
$\nabla_{r}(G)
\leq f(r)$ for every integer $r\geq 0$ and for every graph in ${\cal G}$. The notion of bounded expansion was 
introduced by Nešetřil
 and Ossona de Mendez in~\cite{NesetrilO11onno} and was studied in~\cite{NesetrilO08grad-I,NesetrilO08grad-II,NesetrilO08grad-III}.
 Examples of classes of graphs with bounded expansion include proper minor-closed classes of graphs\footnote{A graph class ${\cal G}$ is {\em (proper) minor-closed} if it (does not contain all graphs
 and) is closed under taking  minors.}, classes of graphs with bounded maximum degree, classes of graphs excluding a subdivision of a fixed graph, classes of graphs that can be embedded on a fixed surface with bounded number of crossings per each edge and many others (see~\cite{NesetrilMW09char}). Many meta-algorithmic  results from proper minor-closed classes of graphs to classes of graphs with bounded expansion, see~\cite{Kreutzer08algo,DawarGK07loca,DvorakK09algo,Grohe08algo,Grohe08logi}.
These include the classes of graphs of bounded local treewidth.
The reader will 
recall
that a class of graphs 
 ${\mathcal C}$ of graphs 
has
bounded local treewidth iff there a function $f$ such that for all
graphs $G\in {\mathcal C}$ and all vertices $v\in V(G)$
the neighbourhood of $v$ of distance $k$ from $v$ has treewidth $f(k)$.

The following theorem was proved by  Dvorak, Kral, and Thomas in~\cite{DvorakKR10deci} and is one of the most general meta-algorithmic results in parameterized complexity.

\begin{theorem}
If ${\cal G}$ is a  graph class of bounded expansion, then {\sc $\phi$-Model Checking for FOL  in ${\cal C}$} belongs to {\sf FPT}. In particular, 
there exists a computable function $f: \Bbb{N}\rightarrow \Bbb{N}$ 
and an algorithm that, given a $n$-vertex graph $G\in{\cal G}$ and a FOL-formula $\phi$ as input, decides whether $G\models \phi$ in 
$f(|\phi|)\cdot n^{}$ steps.
 \end{theorem}

Actually, there is an extension of the  above for the even more general case where ${\cal G}$ has {\em locally bounded expansion}, that is there exists a function   $f:\Bbb{N}\times \Bbb{N}\rightarrow \Bbb{R}^{+}$ such that 
$\nabla_{r}(G)
\leq f(d,r)$ for every integer $r\geq 0$ and for every graph 
that is a subgraph of the $d$-neighborhood of a vertex of a graph from in ${\cal G}$. In this case, according to~\cite{DvorakKR10deci}, the algorithm runs in $f(|\phi|)\cdot n^{1+\epsilon}$ steps, for every $\epsilon>0$.

\subsubsection{Graph Minors}
\label{graphminors}

The Graph Minor Theory has been developed by Robertson and Seymour   towards proving the long-standing Wagner’s Conjecture: {\sl graphs are well quasi-ordered  under the minor relation} (see Theorem~\ref{rste} below).  The contribution of the results and the methodologies derived by this theory to the design and analysis of parameterized algorithms, was fundamental.
In this subjection and the next one we present some of the main contributions in this direction.

A {\em graph parameter} is a recursive function ${\bf p}$ mapping graphs to non-negative integers. A graph parameter ${\bf p}$ is {\em minor-closed} if $H\leq_{\rm m}  G$ implies ${\bf p}(H)\leq {\bf p}(G)$.
Given a graph parameter  ${\bf p}$ and $k\geq 0$, we define 
its {\em $k$-th obstruction} ${\obs}({\bf p},k)$ as the set of all minor-minimal graphs 
of the set $\{G\mid {\bf p}(G)>k\}$. The following lemma is a direct consequence of the definitions.

\begin{lemma}
\label{easy}
If ${\bf p}$ is a  minor-closed parameter ${\bf p}$ and $G$ is a graph, then ${\bf p}(G)\leq k$ if and only if  
none of the graphs in $\obs({\bf p},k)$ is a minor of $G$.
\end{lemma}

The following theorem is one of the deepest results in  Graph theory and was the main 
result of the Graph Minor series of Robertson and Seymour.

\begin{theorem}[Robertson and Seymour theorem -- \cite{RobertsonS04-XX}]
\label{rste}
Every infinite sequence of graphs contains two distinct elements $G_{i}$, $G_{j}$ such that $G_{i}\leq_{m}  G_{j}$.
\end{theorem}

By  definition,  the graphs in ${\obs}({\bf p},k)$ are incomparable with respect to the minor relation. Therefore,   a direct corollary of Theorem~\ref{rste} is the following.

\begin{corollary}
\label{finitecor}
${\obs}({\bf p},k)$  is finite for every parameter ${\bf p}$ and $k\geq 0$.
\end{corollary}

Part of the Graph Minors series was the study of the parameterized complexity 
of the following problem and its generalizations.\\

\noindent{\sc $H$-Minor Checking}\\
\noindent {\sl Instance:}~~Two graphs $G$ and $H$. \\
{\sl Parameter:} $k=|V(H)|$.\\
{\sl Question:}~~Is $H$ a minor of $G$?

\begin{theorem}[Robertson and Seymour \!\!\cite{RobertsonS95-XIII}]
\label{checking}
{\sc $H$-Minor Checking}$\ \in f(k)\cdot n^{3}$-$\FPT$. 
\end{theorem}

To demonstrate the importance of the above results to  \FPT-algorithm design, we consider the 
following meta-problem.\\

\noindent{{\sc {$k$-Parameter Checking for }}\mbox{\rm\bf p}}\\
\noindent {\sl Instance:}~~ a graph ${G}$ and an integer  ${k}\geq 0$.\\
{\sl {Parameter:} ${k}$}\\
{\sl Question:}~~ {\bf p}({G})$\leq {k}$?\\

It is now easy to check that Lemma~\ref{easy}, Corollary~\ref{finitecor}, and Theorem~\ref{checking}  imply the following meta-algorithmic theorem.

\begin{theorem}
\label{metaena}
If ${\bf p}$ is a minor-closed graph parameter, then {{\sc {$k$-Parameter Checking for }}\mbox{\rm\bf p}}
 belongs to $ f(k)\cdot n^{3}$-\FPT.
\end{theorem}

Notice that the above theorem implies the {\em existence} of an \FPT-algorithm for  {{\sc {$k$-Parameter Checking for }}\mbox{\rm\bf p}} and {\sl not} its constructibility.
The construction of the algorithm for some specific parameter ${\bf p}$ is possible only when $\obs({\bf p},k)$ can be constructed  for all $k\geq 0$. Moreover, the general  problem of computing $\obs({\bf p},k)$ is not recursively solvable, as  has been noticed in~\cite{FellowsL94}  (see also ~\cite{Leeuwen90grap}); see also the results of  Friedman,  Robertson, and Seymour~\cite{FriedmanRS87them} on the non-constructibility of the Robertson and Seymour theorem. 

Our first example of the  applicability of Theorem~\ref{metaena} concerns {\sc Vertex Cover}. 
We define $\vc(G)$ as the minimum 
size of a vertex cover in $G$. Notice that  $\vc$ is minor-closed and therefore, Theorem~\ref{metaena} can be applied. 
Moreover, in~\cite{DinneenL07prop}, Dinneen and Lai proved that each graph in $\obs({\bf p},k)$ has at most $2k+1$ vertices and this implies the constructibility of $\obs({\bf p},k)$  for a given $k$: just enumerate all graphs of at most $2k+1$ vertices and filter those that are minor minimal. As a consequence,
we can construct an \FPT-algorithm for  {\sc Vertex Cover}.

Our next example concerns the following problem.\\

\noindent{\sc {$k$-Vertex Feedback Set}}\\
\noindent {\sl Instance:}~~ A graph ${G}$ and an integer  ${k}\geq 1$.\\
{\sl {Parameter:} ${k}$}\\
{\sl Question:}  Does $G$ contain a feedback vertex set of size $\leq k$ (i.e., a set of at most $k$ vertices meeting all its cycles)?\\

We define $\fvs(G)$ as the minimum size of a feedback vertex set of $G$. It is easy to observe that $\fvs$ is minor-closed, therefore {\sc {$k$-Vertex Feedback Set}} belongs to $ \FPT$. Unfortunately,
this approach is not able to construct the algorithm as there is no known upper bound on the size of $\obs(\fvs,k)$ for all $k\geq 3$ (for 
$k=0,1,2$,  $\obs(\fvs,k)$ has been found in~\cite{DinneenCF01forb}).
To make the result constructive, we need some more results.

We call a parameter ${\bf p}$ {\em treewidth-bounded}, if $\tw(G)\leq f({\bf p}(G))$ for some recursive function $f$. We also call ${\bf p}$ 
{\em MSOL-definable} if, for every $k\geq 0$, the class of graphs $\{G\mid {\bf p}(G)\leq k\}$ is  definable by an  MSOL formula $\phi_{k}$ (the length of the formula depends on $k$). 
The following lemma  was proved  by  Adler, Grohe, and 
Kreutzer in~\cite{AdlerGK08comp} (see also~\cite{CattellDDFL00onco}).
  
\begin{lemma}
\label{tbp}
If ${\bf p}$ is a parameter that is {MSOL-definable} and treewidth-bounded, then $\obs({\bf p},k)\leq f(k)$ for some recursive function $f$.
\end{lemma}

Lemma~\ref{tbp} and Theorem~\ref{metaena} imply the following:

\begin{corollary}
\label{metaduo}
If ${\bf p}$ is a treewidth-bounded minor-closed graph parameter, then {{\sc {$k$-Parameter Checking for }}\mbox{\rm\bf p}}
is constructively in $f(k)\cdot n$-{\sf FPT}.
\end{corollary}

It is easy to observe that  $\tw(G)\leq \fvs(G)$. Therefore, Corollary~\ref{metaduo} implies that one may construct an algorithm computing $\obs(\fvs,k)$ and this, in turn,  implies that   one can construct an \FPT-algorithm for  {\sc Feedback Vertex Set}.  
\\

Our next example is more complicated. 
A planar graph $H$ is a {\em $k$-fold planar cover} of a graph $G$ if there is a surjection $\chi: V(H)\rightarrow V(G)$ such that
\begin{itemize}
\item for every $e=\{x,y\}\in E(H)$, $\{\chi(x),\chi(y)\}\in E(G)$.
\item for every $v\in V(H)$, the restriction of $\chi$ to the neighbors of $v$ is a bijection.
\item for every $u\in V(G)$, $|\chi^{-1}(u)|\leq k$.
\end{itemize}

We consider the following parameterized problem.\\

\noindent{\sc {$k$-Planar Cover }}\\
\noindent {\sl Instance:}~~ A graph ${G}$ and an integer  ${k}\geq 1$.\\
{\em {Parameter:} ${k}$}\\
{\sl Question:}  Does $G$ have a $k$-fold planar cover?\\

Let ${\bf cov}(G)$ be the minimum $k$ for which $G$ has a   $k$-fold planar cover. For completeness we set ${\bf cov}(G)=\infty$
if such a cover does not exists. It is easy to observe (see e.g.~\cite{Adler08open}) that ${\bf cov}(G)$ is minor-closed, therefore, from Theorem~\ref{metaena}
{\sc {$k$-Planar Cover }}$\in\FPT$. However this result  still remains non-constructive as, so far, there is no known procedure to construct $\obs({\bf cov},k)$ for $k\geq 1$.
\medskip

A more general theorem on the constructibility horizon of the obstruction sets is the following.

\begin{theorem}[Adler, Grohe, and Kreutzer \cite{AdlerGK08comp}]
\label{simpletomore}
There is an algorithm that, given two
classes ${\cal C}_{1}$, ${\cal C}_{2}$ of finite graphs represented by their obstruction sets
 $\obs({\cal C}_{1})$ and $\obs({\cal C}_{2})$,
computes the set of excluded minors for the union ${\cal C}= {\cal C}_{1}\cup{\cal C}_{2}$.
\end{theorem}

The above theorem, can be useful  for making the meta-algorithm of Theorem~\ref{metaena} constructive.
For this consider two graph parameters ${\bf p}_{1}$ and ${\bf p}_{2}$. We define the {\em union}
of ${\bf p}_{1}$ and ${\bf p}_{2}$ as the parameter ${\bf p}$ where ${\bf p}(G)=\max\{{\bf p}_{1}(G),{\bf p}_{2}(G)\}$.
A direct consequence of Theorem~\ref{simpletomore} is the following:

\begin{corollary}
\label{corpars}
Let ${\bf p}_{1}$ and ${\bf p}_{2}$ be  minor-closed parameters for which 
 {{\sc {$k$-Parameter Checking for }}\mbox{\rm\bf p}$_{i}$}  is constructively in ${\sf FPT}$ for $i=1,2$. Then 
the {{\sc {$k$-Parameter Checking for }}\mbox{\rm\bf p}} is also constructively in {\sf FPT}, where ${\bf p}$ is the union of ${\bf p}_{1}$ and ${\bf p}_{2}$.
\end{corollary}

An interesting problem is to extend as much as possible the collection of parameters where Corollary~\ref{corpars}
can be made constructive.
Another running project is to find (when   possible) counterparts of the algorithm
of Theorem~\ref{checking} for other partial relations on graphs.  
An important step in this direction was done by Grohe, Kawarabayashi, Marx, and Wollan \cite{GroheKMW10find} for the relations of topological minor and immersion.

\subsubsection{The irrelevant vertex technique}
\label{irrelevant}

One of the most important  contributions of the Graph Minors project to 
parameterized algorithm design was the proof, in~\cite{RobertsonS-XIII}, 
that the  following problem is in $f(k)\cdot n^{3}$-$\FPT$ (a faster,
$f(k)\cdot n^{2}$
step algorithm appeared recently in~\cite{Kawarabayashi11thed}).\\

\noindent{\sc {$k$-disjoint Paths}}\\
\noindent {\sl Instance:}~~ A graph $G$ with  $k$ pairs of terminals $T=\{(s_1,t_1),\ldots, (s_k,t_k)\}\in V(G)^{2(k)}$.\\
{\em {Parameter:} ${k}$}\\
{\sl Question:}  Are there $k$ pairwise vertex disjoint paths $P_1,\ldots,P_k$ in $G$
		such that $P_i$ has endpoints $s_i$ and $t_i$?\\

The algorithm for the above problem is based on 
an idea known as the {\em irrelevant vertex technique}
and revealed strong links between structural graph theory and parameterized 
algorithms. In general the idea is described as follows.

Let $(G,T)$ be an input of the {\sc $k$-disjoint Paths}.
We say that a vertex $v\in V(G)$ is {\em solution-irrelevant}
when $(G,T)$ is a YES-instance if only if $(G\setminus v,T)$ is a YES-instance. 
If the input graph $G$ violates some structural condition, then it is possible
to find a vertex $v$ that is solution-irrelevant.
One then repeatedly removes
such vertices until the structural condition 
is met which means that the graph has been simplified
and a solution is easier to be found.  

The structural conditions used the in algorithm from~\cite{RobertsonS-XIII} 
are two: 
\begin{itemize}
\item[{\sf (i)}] $G$ excludes a clique, whose size depends on $k$,
as a minor and 
\item[{\sf (ii)}] $G$ has  treewidth bounded by some function of $k.$ 
\end{itemize}
In~\cite{RobertsonS-XIII} Robertson and Seymour proved, for some specific   function $h: \Bbb{N}\rightarrow \Bbb{N}$, that if the input graph $G$ contains 
some $K_{h(k)}$ as a minor, then 
$G$ contains some solution-irrelevant vertex that can be found in $h(k)\cdot n^{2}$ steps.
This permits us to assume that every input of the problem is $K_{h(k)}$-minor free, thus enforcing 
structural condition (i). Of course, such graphs may still be complicated enough and do not necessarily 
meet condition~(ii). Using a series of (highly non-trivial) structural results~\cite{RobertsonS-XXI,RobertsonS-XXII}, Robertson and Seymour
proved that there is some  function $g: \Bbb{N}\rightarrow \Bbb{N}$ such that 
every $K_{h(k)}$-minor free graph with treewidth at least $g(k)$ contains 
a solution-irrelevant vertex that can be found in $g(k)\cdot n^{2}$ steps.
This enforces structural property (ii) and then {\sc $k$-disjoint Paths}
can be solved in $f(k)\cdot n$ steps, 
using Courcelle's theorem (Theorem~\ref{courcelle:th}).

Actually the above idea was used in~\cite{RobertsonS-XIII} to solve a more 
general problem that contains both  {$k$-disjoint Paths} and {\sc $H$-Minor Checking}.
The only ``drawback'' of this algorithm is that its parametric dependence, i.e.,
the function $f$ is immense. Towards lowering the  contribution of $f$,
better combinatorial bounds where provided by Kawarabayashi and 
Wollan in~\cite{KawarabayashiW10asho}. Also, for planar graphs a better  upper bound was  given in~\cite{AdlerKKLST12tigh}.
The irrelevant vertex technique has been used extensively
in parameterized algorithm design. For a sample of  results that use this technique, see~\cite{DawarGK07loca,DawarK09domi,GolovachKPT09indu,KawarabayashiK08integ,Kawarabayashi2010oddc,GroheKMW10find,KobayashiK09algo}.

\subsection{Faster {\sf FPT}-algorithms}
\label{ffpta}

Clearly the {\sf FPT}-algorithms presented in the previous subsection are far from being practical as their parameter dependence 
may be huge.  A big part of research has been devoted to the reduction of their parameter dependence 
or at least to the detection of classes of their instances where such a simplification is possible.

\subsubsection{Dynamic Programming}
\label{dynamicpr}

While the parameter dependence of the algorithms derived by Theorem~\ref{courcelle:th} are huge (see~\cite{FrickG04thec}), they usually can be considerably improved with the use of  dynamic programming. In what follows, we describe the general approach and 
provide two  simple examples.

Suppose that $G$ is a graph and let $(T,\tau)$ be a branch decomposition of it of width at most $k$. 
For applying dynamic programming on $(T,\tau)$ we consider the tree $T$ to be {\em rooted} at one of its leaves.
Let $v_{r}$ be this leaf and let $e_{r}$ be the edge of $T$ that contains it.
Also, we slightly enhance the definition of a branch decomposition so that 
no edge of $G$ is assigned to $v_{r}$ and thus
$\mids(e_{r})=\emptyset$ (for this, we take any edge of the branch decomposition, subdivide it and
then connect the subdivision vertex with a new (root) leaf $t_{r}$, using the edge $e_{r}$). The edges of $T$ can be
oriented towards the root $e_{r}$, and for each edge $e\in E(T)$ we denote by $E_{e}$ the edges of
$G$ that are mapped to leaves of $T$ that are descendants of $e$. We also set $G_{e}=G[E_{e}]$ and
we denote by $L(T)$ the edges of $T$ that are incident to leaves of $T$  that are different than $v_{r}$. Given an edge $e$ heading
at a non-leaf vertex $v$, we denote by $e_L,e_R \in E(T)$ the two edges with tail $v$. 

%

We give two examples. We first present how to to do dynamic programming for solving the following problem.\\

\noindent{\sc $bw$-Vertex Cover}

\noindent {\sl Instance:}~~A graph $G$ and a non-negative  integer $\ell$.  \\
{\sl Parameter:} $k=\bw(G)$.\\
{\sl Question:}~~Does $G$ have a vertex set $S$ of size at most $\ell$ that intersects all the edges of $G$?\\

Let $G$ be a graph and $X,X'\subseteq V(G)$ where $X\cap X'=\emptyset$. We say that
$\vc(G,X,X')\leq \ell$ if $G$ contains a vertex cover $S$ where $|S|\leq \ell$ and $X\subseteq S\subseteq
V(G)\setminus X'$.
%
 Let 
 \begin{eqnarray}
 {\cal R}_{e} & = & \{(X,\ell)\mid X\subseteq \mids(e) \mbox{~and~} \vc(G_{e},X,\mids(e)\setminus X)\leq \ell\}. \nonumber
 \end{eqnarray}
The set ${\cal R}_{e}$ book-keeps  all pairs $(X,\ell)$ certifying the existence, in $G_{e}$,
of a vertex cover of size $\leq \ell$ whose restriction in $\mids(e)$ is $X$.
Observe that  $\vc(G)\leq \ell$ iff $(\emptyset,\ell)\in {\cal R}_{e_{r}}$.
 For each $e\in E(T)$ we can compute ${\cal R}_{e}$ by using the following dynamic programming formula:
 \begin{eqnarray*}
 {\cal R}_{e} & = &
\begin{cases}
\{(X,\ell)\mid X\subseteq e \mbox{~and~} X\neq \emptyset \wedge \ell\geq |X|\}  & \text{if $e\in L(T)$}\\
\{(X,\ell)\mid \exists (X_{1},\ell_{1})\in {\cal R}_{e_{1}},  \exists (X_{2},\ell_{2})\in {\cal R}_{e_{2}}:   & \\
(X_{1}\cup X_{2})\cap \mids(e)=X \wedge   \ell_{1}+\ell_{2}-|X_{1}\cap X_{2}|\leq \ell \} & \text{if
$e\not\in L(T)$}
\end{cases}
\end{eqnarray*}
Note that for each $e\in E(T)$, $|{\cal R}_{e}|\leq 2^{|\mids(e)|}\cdot \ell$. Therefore, the above
algorithm can check whether $\vc(G)\leq \ell$ in $O(4^{\bw(G)}\cdot \ell^2\cdot |V(T)|)$ steps.
Using the fact that every $n$-vertex graph has at most $O(\bw(G)\cdot n)$ edges, we obtain that 
 $|V(T)|=O(\bw(G)\cdot n)$. As $l\leq n$, the above 
 dynamic programming  algorithm implies that {\sc $bw$-Vertex Cover} belongs to  $2^{O(k)}\cdot n^{3}$-\FPT.\\
 

%
%
%
%

Our second example is a dynamic programming algorithm for the following problem:\\

\noindent{\sc $bw$-3-Coloring}\\
\noindent {\sl Instance:}~~A graph $G$.  \\
{\sl Parameter:} $k=\bw(G)$.\\
{\sl Question:}~~Does $G$ have a proper 3-coloring?\\

We will consider 3-coloring functions of the type $\chi: S\rightarrow \{1,2,3\}$ and, for $S'\subseteq S$,
we use the notation $\chi\!\!\mid_{S'}=\{(a,q)\in \chi\mid a\in S'\}$ and $\chi(S')=\{\chi(q)\mid q\in S'\}$.

 Given a rooted branch decomposition $(T,\tau)$ an edge $e\in V(T)$, we use the notation
 ${\cal X}_{e}$ for all   functions $\chi: \mids(e)\rightarrow \{1,2,3\}$  and the notation
  $\bar{\cal X}_{e}$ for all proper 3-colorings of $G_{e}$. We define 
  \begin{eqnarray}
 \alpha_{e} & = & \{\chi\in{\cal X}_{e}\mid \mbox{~$\exists \bar{\chi}\in\bar{\cal X}_{e}: \chi\mid_{\mids(e)} =  \bar{\chi}$}\}\nonumber
\end{eqnarray}
The set $\alpha_{e}$ stores the restrictions in $\mids(e)$ of  all
proper 3-colorings of $G_{e}$.
 Notice that for each $e\in E(T)$, $|\alpha_{e}|\leq 3^{\mids(e)}\leq 3^{k}$ and observe that $G$ has 
 a 3-coloring iff $\alpha_{e_{r}}\neq\emptyset$ (if $\alpha_{e}\neq\emptyset$, then it contains the empty function).
  For each $e\in E(T)$ we can compute ${\cal R}_{e}$ by using the following dynamic programming formula:
 \begin{eqnarray*}
 \alpha_{e} & = &
\begin{cases}
\{\chi\in{\cal X}_{e}\mid |\chi(e)|=2\} & \text{if $e\in L(T)$}\\
\{\chi\in{\cal X}_{e}\mid \exists \chi_{L}\in {\cal X}_{e_{L}},   \exists \chi_{R}\in {\cal X}_{e_{R}},:   & \\
\chi_{L}\!\!\mid_{\mids(e_{L})\cap \mids(e_{R})}=\chi_{R}\!\!\mid_{\mids(e_{L})\cap \mids(e_{R})} \mbox{\ and}& \\
(\chi_{L}\cup \chi_{R})\!\!\mid_{\mids(e)}=\chi\}
& \text{if
$e\not\in L(T)$}
\end{cases}
\end{eqnarray*}

Clearly, this simple algorithm  proves that  {\sc $bw$-3-Coloring} belongs to  $2^{O(k)}\cdot n $-\FPT. 
A straightforward extension implies that  {\sc $bw$-}q{\sc -Coloring} belongs to  $q^{O(k)}\cdot n $-\FPT.

In both above examples, we associate to each edge $e\in E(T)$  some {\em characteristic structure},  
that, in case of {\sc $bw$-Vertex Cover} and {\sc $bw$-3-Coloring},  is ${\cal R}_{e}$ and $\alpha_{e}$  respectively.
This structure is designed so that its value for $e_{r}$ is able to determine the answer to the 
problem. Then, it remains to give this structure for the leafs of $T$ and then provide 
a recursive procedure to compute bottom-up all characteristic structures from the leaves to the root.
This  dynamic programming machinery has been 
used many times in parameterized algorithm design and for much 
more complicated types of problems. In this direction, the algorithmic 
challenge is to reduce as much as possible the information that is book-kept 
in the characteristic structure associated to each edge of $T$. Usually, for simple 
problems as those examined above, where the structure encodes subsets (or a bounded number of subsets)
of $\mids(e)$, it is  easy to achieve a single-exponential parametric dependence.
Typical examples of such  problems are {\sc Dominating Set}, {\sc Max Cut} and {\sc Independent Set},
parameterized by treewidth/branchwidth, 
where the challenge is to reduce as much as possible the constant hidden in the $O$-notation 
of their $2^{O(k)}$-parameter dependence\footnote{As here we care about the exact parameter dependence the constants may vary depending on whether we parameterize by treewidth or branchwidth.}. 
Apart from tailor-made improvements for specific problems such as  $tw$-{\sc Dominating Set} and $tw$-{\sc Vertex Cover} (see e.g.~\cite{AlberN02impr,AlberFN05expe,BetzlerNU06tree}),
substantial progress in this direction has  been done using the Fast Matrix Multiplication technique, introduced by Dorn in~\cite{Dorn06fast} and the results of Rooij, Bodlaender, and Rossmanith in~\cite{RooijBR09dyna}, where they used the 
Generalized Subset Convolution Technique (introduced by of Björklund, Husfeldt, Kaski, and Koivisto in~\cite{BjorklundHKK07four}).

Recently, some lower bounds  on the parameterized complexity of problems parameterized by treewidth
were given by Lokshtanov, Marx, and Saurabh in~\cite{LokshtanovMS11know}. According to~\cite{LokshtanovMS11know}, unless SAT is solvable in $O^{*}((2-\delta)^{n})$ steps,
{\sc $bw$-Independent-Set} does not belong to  $(2-\epsilon)^{k}\cdot n^{O(1)}$-\FPT,
{\sc $bw$-Max Cut} does not belong to  $(3-\epsilon)^{k}\cdot n^{O(1)}$-\FPT,
{\sc $bw$-Dominating-Set} does not belong to  $(3-\epsilon)^{k}\cdot n^{O(1)}$-\FPT,
{\sc $bw$-Odd Cycle Transversal} does not belong to  $(3-\epsilon)^{k}\cdot n^{O(1)}$-\FPT,
and 
{\sc $bw$-{\rm q}-Coloring} does not belong to  $(q-\epsilon)^{k}\cdot n^{O(1)}$-\FPT.
The assumption that  ${\sc SAT}\not\in O^{*}((2-\delta)^{n})$-{\sf TIME}, is known as the {\sl Strong Exponential Time Hypothesis} (SETH) and was introduced by  Impagliazzo and Paturi in~\cite{ImpagliazzoP99thec}. 

For more complicated problems, where the characteristic
encodes pairings, partitions, or packings of $\mids(e)$  the parametric dependence 
of the known \FPT\ algorithms is of the type $2^{O(k\log k)}\cdot n^{O(1)}$ or worst. Usually, 
these are problems involving  some global constraint such as connectivity on the certificate of their solution. 
Recent complexity results of Lokshtanov, Marx, and Saurabh~\cite{LokshtanovMS11slig} show that for problems such as 
the {\sc Disjoint Paths Problem} no $2^{o(k\log k)}\cdot n^{O(1)}$-algorithm exists unless ETH fails.
On the other hand, a technique was recently introduced by Cygan, Nederlof, M. Pilipczuk, M. Pilipczuk, van Rooij, and Wojtaszczyk in~\cite{CyganNPPRW11solv},
solved many problems of this type in $2^{O(\tw(G))}\cdot n^{O(n)}$ steps by  randomized Monte Carlo algorithms. This includes problem as~{\sc  Hamiltomian Path}, {\sc  Feedback Vertex Set} and {\sc  Connected Dominating Set}. 
For planar graphs it is still possible to 
design $O^{*}(2^{O(\bw)})$ step dynamic programming algorithms using the 
the {\em Catalan structures} technique introduced by Dorn, Penninkx, Bodlaender, and  Fomin in~\cite{DornPBF10effi}. This technique 
uses a special type of branch decomposition called {\em sphere cut decompositions} introduced in~\cite{SeymourT94call}.
In such decompositions, the vertices of  $\mids(e)$ are virtually arranged on a closed curve 
of the surface where the input graph is embedded. In case the characteristic structure encodes 
non-crossing pairings,  its size is proportional to  the $k$-th Catalan 
number that is single-exponential in $k$. This fact   yields the $O^{*}(2^{O(\bw)})$ time bound 
to the corresponding dynamic programming algorithms. The same technique 
was extended by Dorn,  Fomin, and Thilikos in~\cite{DornFT06fast} for bounded genus graphs and in~\cite{DornFT08cata} for every graph class that excludes some graph as a minor. Finally, Rué, Sau, and Thilikos~in~\cite{RueST10dyna}  extended this technique  to wider families of problems.

%
%
%
%

\subsubsection{Single-Exponential Algorithms}
\label{singlee}


%
In~\cite{CaiJ03onth} Cai and Juedes proved that $k$-{\sc Vertex Cover} cannot be solved in $2^{o(k)}\cdot n^{O(1)}$ steps unless 
ETH collapses. The linear parameter dependence of  the standard reductions from $k$-{\sc Vertex Cover}
 to
other problems such as {\sc $k$-Dominating Set} or {\sc $k$-Feedback Vertex Set} imply that the same lower bound
holds for these problems as well. In other words, for a wide class of parameterized problems 
membership in the parameterized class~{\sf EPT} (containing all problems solvable in $2^{O(k)}\cdot n^{O(1)}$
steps as defined in Subsection~\ref{boundedx})  is the best we may expect. 
However, to prove  membership in {\sf EPT} is not always easy.
In fact none of the the techniques described so far  guarantees that the derived 
algorithms will have single-exponential parameter dependence. We describe two general techniques  that yield such bounds.

\paragraph{Iterative compression.}
As a technique, iterative compression dates back to the parameterized algorithm given by Reed, Smith, and Vetta in~\cite{ReedSV04find}
for the problem {\sc Odd Cycle Transversal} that, given  a graph $G$ and an integer $k$, 
asks whether there is a set $S$ of at most $k$ vertices meeting all odd cycles of $G$ (see also~\cite{Huffner05inst}).
Before~\cite{ReedSV04find}, this was a popular open problem and iterative compression appeared 
to be just the correct approach for its solution. In what follows we will give a generic example of this technique.
Let $\Pi$ be a graph property  that is {\em hereditary} i.e., if $H$ is an induced subgraph of $G$
and $G\in\Pi$, then also $H\in\Pi$. We define the following meta-problem:\medskip

\noindent{\sc $k$-Vertex Deletion Distance From $\Pi$} \\
\noindent {\sl Instance:}~~A graph $G$ and a non-negative  integer $k$.  \\
{\sl Parameter:} $k$.\\
{\sl Question:}~~Is there an $S\subseteq V(G)$ of size at most $k$ such that $G-S\in\Pi$?
(Here we denote by $G-S$ the graph $G[V(G)\setminus S]$.)\\

Our intention is to solve this problem for several choices of $\Pi$. Clearly, if $\Pi$ is 
``being edgeless'', the above problem is {\sc Vertex Cover}. Also, 
if $\Pi$ is ``being acyclic" the above problem defines {\sc Feedback Vertex Set}, while 
if $\Pi$ is ``being bipartite", it defines {\sc Odd Cycle Transversal}.
 In fact, iterative compression   reduces algorithmically the problem 
to its annotated version below with the additional restrictions that (a) $H[Q]\in\Pi$ and (b) $H-Q\in\Pi$. \medskip

\noindent{\sc $k$-Annotated Vertex Deletion Distance From $\Pi$} \\
\noindent {\sl Instance:}~~A graph $H$, a set $Q\subseteq V(H)$, and a non-negative integer $k$.  \\
{\sl Parameter:} $k$.\\
{\sl Question:}~~Does $H$ have a vertex set $R\subseteq V(H)$ of size at most $k$ such that $R\subseteq Q$ and $H-R\in\Pi$?\\

We now present the following routine that, given a graph $G$ and  a non-negative integer $k$, returns 
either a vertex cover $S$ of size at most $k$ or  {\tt NO} which means that no such solution exists.

\begin{tabbing}
{\sl Procedure} {\sf {solveVDD$_\Pi$}}$(G,k)$\\
{\bf 1.} If $V(G)=\emptyset$, then return $S=\emptyset$.\\
{\bf 2.} Pick \= a vertex $v\in V(G)$ and\\ 
\> if \=   {\sf {solveVDD$_\Pi$}}$(G-v,k)=${\tt NO}, then return {\tt NO}.\\
{\bf 3.} let $S=        $ {\sf {solveVDD$_\Pi$}}$     (G-v,k)\cup\{v\}.$\\
{\bf 4.} If $|S|\leq k$ then return $S$.\\
{\bf 5.} for \=  all $F\subseteq S$ such that $G[F]\in \Pi$,\\
\> let $Q=V(G)-S$, $H=G[F\cup Q]$, $k'= k-|S-F|$, and check whether  \\
\> \> the {\sc $k$-Annotated Vertex Deletion Distance From $\Pi$}\\
\> \> with input $H$, $Q$, and $k'$ has a solution $R$ where   $R\subseteq Q$.\\
\> \> If this solution $R$ exists, then return $R\cup(S-F)$.\\
{\bf 6.} return {\tt NO}
\end{tabbing}

The above procedure considers a vertex ordering $(v_{1},\ldots,v_{n})$ of $G$
and solves the problem by considering the graphs $G_{i}=G[\{1,\ldots,i\}]$, $i=0,\ldots,n$.
If $i=0$, then, as returned in line {\bf 2}, the empty set is a solution.
Assume now that a solution $S'$ for $G_{i}$ is known, which, by the hereditarity 
of $\Pi$, implies that $S=S'\cup\{v_{i+1}\}$ is a solution for $G_{i+1}$ of size 
$\leq k+1$. If $|S|\leq k$, then $S$ is a solution also for $G_{i}$, as is decided in Step {\bf 4}.
In Step {\bf 5} the algorithm is trying to find a solution in $G$ that does not intersect $F\subseteq S$
and contains $S-F$ for all possible $2^{k+1}$ choices of $F$. Certainly, this is not possible 
if $H-Q=G[F]\not\in \Pi$ as indicated by the filtering condition  of this loop. In what remains,
one has to solve the annotated version of the problem of $H$ and for $k'$ with $Q$ as annotated set.
To find such a solution  one may use the additional properties that
 $H-Q\in \Pi$ and $H[Q]\in \Pi$ (recall that $Q=V(G)-S$ which means 
 that $H[Q]=G[Q]\in\Pi$). From now on, the technique is specialized to each particular problem.
 For instance, for {\sc Vertex Cover}, the graph $H$ is a bipartite graph with 
 parts $F$ and $Q$. The question is whether some subset $R, |R|\leq k'$ of $Q$ meets all the edges of $H$, 
 which essentially asks whether the non-isolated vertices of $Q$ are at most $k'$. As this can be decided 
 is polynomial time, the only non-polynomial load of the whole procedure is the one of Step {\bf 5}, and 
 this yields an $O(2^{k}\cdot n^{2})$ step algorithm. The analogous question for the case of {\sc Feedback Vertex Set} needs more effort as now $F$ and $Q$ 
 induce forests in $H$. In~\cite{DehneFLRS05anfp, GuoGHNW06comp}, it was shown that the annotated version of the {\sc Feedback Vertex Set} (where the requested solution is a subset of $Q$)
 can be reduced to an equivalent one where if there exists a solution, then  the annotated set has at most $c\cdot k$ vertices. As ${ck\choose k}=2^{O(k)}$, the solution of the annotated version adds a $2^{O(k)}$ overhead 
to the $2^{k+1}$ contribution of Step {\bf 5} which results to an $2^{O(k)}\cdot n^{2}$ step algorithm. The technique we describe  bellow is able to design \FPT-algorthms  for these three problems such as 
 {\sc Feedback Vertex Set} \cite{DehneFLRS05anfp, GuoGHNW06comp},
 {\sc Directed Feedback Vertex Set} \cite{ChenLOR08afix},
 {\sc Almost 2-SAT} \cite{RazgonO09almo}, and 
{\sc Cluster Vertex Deletion}  \cite{HKMN10fixe}
(see also the survey of Guo, H. Moser, and R. Niedermeier~\cite{GuoMN09iter} on results using the iterative compression  technique).

\paragraph{Color-coding.}
%

One of the most beautiful ideas in parameterized algorithm design is 
Color-coding, introduced by Alon, Yuster, and Zwick in~\cite{AlonYZ95colo}.  
This technique  was first applied to the following problem.\\

\noindent{\sc $k$-Path}

\noindent {\sl Instance:}~~A graph $G$ and a non-negative  integer $k$.  \\
{\sl Parameter:} $k$.\\
{\sl Question:}~~Does $G$ contain a path of at least $k$ vertices?\\

The above problem can be solved in $2^{O(k\cdot \log k)}\cdot n$ steps using dynamic programming techniques  (see Subsection~\ref{dynamicpr}).  
However, 
an {\sf EPT}-algorithm  for this problem was highly welcome. The main reason for this is that it resolved the conjecture of  
Papadimitriou and Yannakakis who conjectured, 
 in~\cite{PapadimitriouY96onli}, that it is possible to check in polynomial whether a $n$-vertex graph contains a path of length $\log n$.

The first step for solving 
{\sc $k$-Path}, is to consider a function $\chi: V(G)\rightarrow \{1,\ldots,k\}$ coloring the vertices of 
$G$ with $k$ distinct colors. Given a path {\sf P} of $G$, we denote by ${\bf col}(P)$ the set of colors of the vertices in {\sf P}. We call a path {\sf  P} of $G$ {\em $\chi$-colorful}, or simply {\em colorful}, if all its colors are pairwise distinct.
We also use the term {\em $i$-path} for a path of $i$ vertices.
 We now define the following variant of the original problem:\\

\noindent{\sc $k$-Colorful Path}

\noindent {\sl Instance:}~~A graph $G$, a positive integer $k$, and a coloring $\chi: V(G)\rightarrow \{1,\ldots,k\}$.  \\
{\sl Parameter:} $k$.\\
{\sl Question:}~~Does $G$ contain a colorful $k$-path?\\

%

{\sc $k$-Colorful Path} can be solved 
 by the following dynamic programing procedure. First we fix a vertex $s\in V(G)$ and then  
for any $v\in V(G)$, and $i\in\{1,\ldots,{k}\}$, we define \\
\begin{eqnarray*}
{\cal C}_{s}(i,v) & = & \{{R}\subseteq\{1,\ldots,{k}\}\mid \mbox{$G$ has a colorful $i$-path ${P}$ from}\\
& & \mbox{~~~~~~~~~~~~~~~~~~~~~~~  $s$ to $v$ such that    ${\bf col}(P)={R}$}\}.
\end{eqnarray*}
Notice that ${\cal C}_{s}(i,v)$ stores sets of colors in paths of length $i-1$ between 
$s$ and $v$, instead of the paths themselves. 
Clearly,  $G$ has a colorful ${k}$-path starting from $s$ iff $\exists v\in V(G): {\cal C}_{s}({k},v)\neq\emptyset$. The dynamic programming is based on the following 
relation:
$${\cal C}_{s}(i,v)=\bigcup_{v'\in N_{G}(v)}\{R\mid R\setminus \{\chi(v)\}\in {\cal C}_{s}(i-1,v')\}$$
{Notice} that  $|{\cal C}_{s}(i,v)|\leq {{k}\choose i}$ and, for all $v\in V(G)$,  ${\cal C}_{s}(i,v)$ can be computed in $O(m\cdot {{k}\choose i}\cdot i)$ steps (here, $m$ is the number of edges in $G$).
For all $v\in V(G)$, one can compute ${\cal C}_{s}({k},v)$ in  $O(\sum_{i=1,\ldots,{k}}m\cdot {{k}\choose i}\cdot i)=O(2^{{k}}\cdot {k}\cdot m)$ steps.
We conclude that one can check whether $G$ colored by $\chi$ has a colorful path of length $k$ in $O(2^{{k}}\cdot {k}\cdot m\cdot n)$ steps (just apply the above dynamic programming for each  possible starting vertex $s\in V(G)$).\\

A {\em family of  $k$-perfect hash functions} is a family ${\cal F}$  of functions from
$\{1,\ldots, n\}$  onto $\{1,\ldots,{k}\}$  such that for each ${S}\subseteq\{1,\ldots,n\}$ 
with $|{S}| = {k}$, there exists
an $\chi\in {\cal F}$ that is bijective when restricted to ${S}$. 

There is a lot of research on the construction of small size families of $k$-perfect hash functions dating back to the work of Fredman, J. Komlós, and E. Szemerédi~\cite{FredmanKS82stor} (see also~\cite{SlotB85onta,SchmidtS90thes,NaorSS95spli}).  A construction of such a family of size $2^{O(k)}\cdot \log n$ was given in~\cite{AlonYZ95colo}.
However,  the hidden constant in  the $O$-notation of this bound is quite big. Instead, 
we may  use the recent results of Chen, Lu, Sze, and Zhang in~\cite{ChenLSZ07impr} that 
give  a family of  $k$-perfect hash functions  ${\cal F}$  where $|{\cal F}|=O(6.4^{k}\cdot n)$. Moreover, according to~\cite{ChenLSZ07impr},
this collection  can be constructed in $O(6.4^{k}\cdot n)$ steps. 
\cite{ChenLSZ07impr}  gives also lower bounds on the size of  a family of  $k$-perfect hash functions indicating somehow the limits of the color-coding technique.

Clearly  $G$ contains a   $k$-path if and only if there is a $\chi\in {\cal F}$ such that $G$, colored by $\chi$, contains a $\chi$-colorful $k$-path.
This equivalence reduces the {\sc $k$-Path} problem  to the {\sc $k$-Colorful Path} problem: just run the above dynamic programming procedure for  {\sc $k$-Colorful Path}  for all colorings in ${\cal F}$.
As $|{\cal F}|=O(6.4^{k}\cdot n)$,  we conclude that 
{\sc $k$-Path}$\in O^{*}(12.8^{k})$-\FPT. \medskip

Color-coding has been used extensively in parameterized algorithm design.
Applications of the same technique can be found in~\cite{FellowsKNRRSTW08fast,DemaineHM09mini,MisraRSS09thei,BetzlerDKN08para,LiuLCS06gree} (see also~\cite{HuffnerWZ08algo,ShlomiSRRS06qpa, ScottIKS05effi,AlonYZ08colo}). Also novel developments 
of the color-coding idea appeared recently in~\cite{FominLRS10fast,AlonLS09fast}.

\subsubsection{Subexponential algorithms}
\label{subexpalg}

Our next step is to deal with the classification of parameterized 
problems in the class {\sf SUBEPT}.
We present  techniques that provide algorithms with subexponential 
parameter dependence for variants of parameterized problems where the inputs are 
restricted by some sparsity criterion.
The most common restriction for problems on graphs is to 
consider their {\em planar} variants where the input graph is embeddible in  a sphere 
without crossings. As mentioned in Theorem~\ref{ethvc}, in Subsection~\ref{optimality}, 
   {\sc $k$-Planar Dominating Set} does not belong to $2^{o(\sqrt{k})}\cdot n$-\FPT, unless ${\sc M}[1]=\FPT$
  and   the same holds 
   for several other problems on planar graphs such as  {\sc Planar Vertex Cover}, {\sc Planar Independent Set}, {\sc Planar Dominating
Set}, and {\sc Planar Red/Blue Dominating Set}~\cite{CaiJ03onth} (see also~\cite{ChenCFHJKX05tigh,ChenKPSX03genu}). 
This implies that for these problems, when the sparsity criterion includes planar graphs,  the best 
running time we may expect is  $2^{O(\sqrt{k})}\cdot n^{O(1)}$. 
The first sub-exponential parameterized algorithm  
on planar graphs appeared by Alber, H. Bodlaender, Fernau, and Niedermeier in~\cite{AlberBFN00} for {\sc Dominating Set}, {\sc Independent Dominating Set}, and 
{\sc Face Cover}. After that, many other problems were classified in  $2^{c\sqrt{k}}\cdot n^{O(1)}$-{\sf FPT}, while 
there was a considerable effort towards improving the constant $c$ for each one of them~\cite{AlberBFKN02,KloksLL02newa,DemaineFHT05talg,Fernau02grap,FernauJ04age,KanjPer02impr,FominT06domi,KoutsonasT10plan,DemaineHT05expo}.

Bidinensionality Theory was proposed in~\cite{DemaineFHT05sube} as a meta-algorithmic 
framework that describes when such optimal subexponential 
{\sf FPT} algorithms  are possible. Our first 
step is to illustrate the idea of bidimensionality for the following problem.\\

\noindent{\sc $k$-Planar  Longest Cycle}

\noindent {\sl Instance:}~~A planar graph $G$ and a nonnegative  integer $k$.  \\
{\sl Parameter:} $k$.\\
{\sl Question:}~~Does $G$ contain a cycle of length at least $k$?\\

We define the graph parameter  ${\bf lc}$ where $${\bf lc}(G)=\min \{k\mid \mbox{$G$ does not contain a cycle of length at least $k$}\}$$
Our subexponential algorithm for {\sc $k$-Planar  Longest Cycle} will be based on the Win/win technique whose main combinatorial
ingredient is the following result of Gu and Tamaki.

\begin{proposition}[\!\!\cite{GuoT10impr}] 
\label{planarexcl}
Every planar graph $G$ where $\bw(G)\geq 3k$ contains a $(k\times k)$-grid as a minor.
\end{proposition}

Actually, Proposition~\ref{planarexcl} is an improvement 
of the result of Robertson, Seymour, and Thomas in~\cite{RobertsonST94quic} where the lower bound to branchwidth was originally $4k$ instead of $3k$.

Notice that ${\bf lc}$ is closed under taking of minors;  the contraction/removal  of an edge in a graph will not cause a bigger cycle to appear. This means that if $G$ is a planar graph and ${\bf lc}(G)\leq l^{2}$, none of the minors of $G$ can contain a cycle of length $l^{2}$. As the $(l\times l)$-grid contains a cycle of length $l^{2}$, we conclude  that $G$ does not contain an $(l\times l)$-grid as a minor. From Proposition~\ref{planarexcl}, $\bw{(G)}< 3l$.
In other words, $\bw(G)< 3\cdot\sqrt{{\bf lc}(G)}$. 

To solve {\sc $k$-Planar  Longest Cycle} we apply the following steps:
\begin{itemize}
\item[(a)] we compute an optimal branch decomposition of $G$. According  to~\cite{GuT08opti}, this can be done in $O(n^{3})$ steps (the result in~\cite{GuT08opti} is an improvement of the algorithm of Seymour and Thomas in~\cite{SeymourT94call}). 
\item[(b)] we check whether $\bw(G)\geq 3\sqrt{k}$. If this is the case then ${\bf lc}(G)>k$ and we can safely return that 
$G$ contains a cycle of length $k$.
\item[(c)] if $\bw(G)< 3\sqrt{k}$, then we solve the following problem by using dynamic programming.
\end{itemize}

\noindent{\sc $bw$-Planar  Longest Cycle}

\noindent {\sl Instance:}~~A planar graph $G$ and a non-negative  integer $l$.  \\
{\em Parameters:} $\bw(G)$.\\
{\sl Question:}~~Does $G$ contain a cycle of length at least $l$?\\

What we have done  so far is to reduce the {\sc $k$-Planar  Longest Cycle} to its bounded branchwidth counterpart.
For general graphs, dynamic programming for this problem requires $O^{*}(2^{O(k\log k)})$ steps.
However, 
for planar graphs, one may use the technique introduced in~\cite{DornPBF10effi} yielding a $2^{O(k)}\cdot n^{O(1)}$ step dynamic programming.
In fact, according to~\cite{Dorn07desi}, {\sc $bw$-Planar  Longest Cycle} can be solved in $O(2^{2.75\cdot  \bw(G)}\cdot n+n^{3})$ steps. 
We conclude that {\sc $k$-Planar  Longest Cycle} belongs  to $O(2^{8.25\cdot  \sqrt{k}}\cdot n+n^{3})$-{\sf SUBEPT}.

Notice that, in order to apply the above machinery, we made the following two observations on the parameter ${\bf lc}$: (1)
it is closed under taking of minors and (2) its value with the $(l\times l)$-grid as input is $\Omega(l^{2})$, i.e. a certificate 
of the solution for the  {\sc $t$-Planar  Longest Cycle} problem spreads along {\sl both} dimensions  of a square grid (i.e. it virtually  spreads {\em bidimensionally} on the grid).
These two observations apply to  a wide range of parameters where the same approach can be used  
for the corresponding problems on planar graphs.  Bidimensionality theory was introduced in~\cite{DemaineFHT05sube} and used
the above observations in order to derive subexponential algorithms for graphs embeddible in surfaces of  bounded genus (see also~\cite{DemaineHT06theb}). 
Later, the same theory was developed for parameters that are closed under contractions~\cite{FominGT09cont} and for classes of graphs 
excluding specific graphs as a minor~\cite{DemaineH08line,DemaineFHT05bidi}. 

Consider the following parameterized meta-problem where ${\cal C}$ is a class of graphs.\\

\noindent{{\sc {$k$-Parameter Checking for }}\mbox{\rm\bf p} {\sc on ${\cal C}$}}\\
\noindent {\sl Instance:}~~ a graph ${G}\in{\cal C}$ and an integer  ${k}\geq 0$.\\
{\em {Parameter:} ${k}$}\\
{\sl Question:}~~ {\bf p}({G})$\leq {k}$?\\

Briefly, this theory can be summarized by the following  meta-algorithmic framework.

\begin{theorem}[Minor bidimensionality~\cite{DemaineFHT05sube,DemaineH08line}]
\label{thm:bidiminor}
Let ${\cal G}$ be an $H$-minor free  graph class.
Let also  ${\bf p}$ be a graph parameter which satisfies the following conditions:
\begin{itemize}
\item[{\bf (i)}] ${\bf p}$ is minor  closed, 
\item[{\bf (ii)}] if $A_{k}$ is the $(k\times k)$-grid,  then ${\bf p}(A_{k})=\Omega(k^{2})$, and  
\item[{\bf (iii)}] for graphs in ${\cal G}$, ${\bf p}$ is computable in time $2^{O(\bw(G))} \cdot n^{O(1)}$.
\end{itemize}
Then,  {{\sc {$k$-Parameter Checking for }}\mbox{\rm\bf p} {\sc on ${\cal C}$}} belongs to $O^{*}(2^{O(\sqrt{k})})$-{\sf SUBEPT}.
\end{theorem}

Notice that Theorem~\ref{thm:bidiminor} can be applied to maximization and minimization problems. 
We provide a typical example of its application on a maximization problem through  $k$-{\sc Path}.  Here,   the associated parameter is
$${\bf p}(G)=\min\{k\mid \mbox{$G$ does not contain a path of length $k$}\}.$$
It is easy to check that conditions {\bf (i)} and {\bf (ii)} are satisfied. Indeed, no bigger path
occurs if we remove or contract an edge and the $(k\times k)$-grid has a (Hamiltonian) path of length $k^{2}$.
Let ${\cal G}$ be any  $H$-minor free graph class.
Condition~{\bf (iii)} holds for  ${\cal G}$, because of the results in~\cite{DornFT08cata}. Therefore,  $k$-{\sc Path} restricted 
to  graphs in ${\cal G}$ belongs to $O^{*}(2^{O(\sqrt{k})})$-{\sf SUBEPT}.

Problems for  which Theorem~\ref{thm:bidiminor} proves {\sf SUBEPT}-membership for graphs excluding 
some graph as a minor are $k$-{\sc Vertex Cover}, $k$-{\sc Feedback Vertex Set}, {\sc {$k$-Almost Outerplanar}}, $k$-{\sc Cycle Packing}, $k$-{\sc Path}, $k$-{\sc Cycle}, $k$-{\rm d}-{\sc Degree-Bounded Connected Subgraph},  $k$-{\sc Minimum Maximal Matching}, and many others.\\

We say that a graph  $H$ is a {\em contraction}  of a graph $G$ if $H$ can be obtained from  $G$ after applying a (possibly empty) sequence of edge contractions\footnote{Formally, the contraction relation occurs
if  
in the definition of the minor relation (Subsection~\ref{ftpfol}) we additionally demand that 
\begin{itemize}
\item [c)]  for every $ \{x,y\}\in E(G)$,  either $\phi(x)=\phi(y)$
or $\{\phi(x),\phi(y)\}\in E(H).$
\end{itemize}}.
A parameter {\bf p} is {\em contraction closed} if $H\leq_{c} G$ implies that ${\bf p}(H)\leq {\bf p}(G)$.

Clearly, there are parameters escaping the applicability of Theorem~\ref{thm:bidiminor} due 
to the fact that they are not minor-closed.  The most typical example of such a parameter is the dominating set number: take a path $P$ of length $3k$ and connect all its vertices with a new vertex.
Clearly,  the resulting graph has a dominating set of size $1$, while it contains $P$ as a minor (actually it contains it as a subgraph) and every dominating set of $P$ has size at least $k$. 
However, the dominating set number is contraction closed and the good news is that 
 there is a counterpart of Theorem~\ref{thm:bidiminor}
for contraction closed  parameters. Before we state this result, we need two more definitions. An {\em apex} graph is  defined to be a 
graph that becomes planar after the removal of a vertex.  A graph class ${\cal G}$ is {\em apex-minor free} if it is $H$-minor free for some apex graph $H$.  An apex-minor free graph  can be seen 
as having some  big enough  ``flat region'', provided that the graph has big treewidth. 
Graph classes that are apex-minor free  are the graph classes of bounded genus and the
classes of graphs excluding some single-crossing graph as a minor~\cite{DemaineHT05expo}.

The graph $\Gamma_{k}$ is
the graph obtained from the  $(k\times k)$-grid by triangulating internal faces of the $(k\times k)$-grid such that all internal vertices become  of degree $6$, all non-corner external vertices are of degree 4, and one corner is joined by edges with all vertices
of the external face (the {\em corners} are the vertices that in the underlying grid have  degree two). For the graph
$\Gamma_6$, see Figure~\ref{fig-gamma-k}.
\begin{figure}[ht]
\begin{center}
\scalebox{.7}{\includegraphics{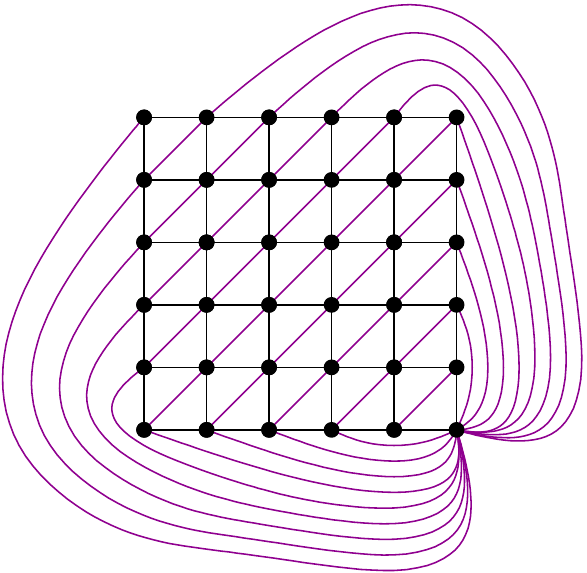}}
\end{center}
\caption{The graph $\Gamma_{6}$.}
\label{fig-gamma-k}
\end{figure}

\begin{theorem}[Contraction bidimensionality~\cite{FominGT09cont}]\label{thm:bidi_revised}
Let ${\cal G}$ be an apex-minor free graph class. Let also  ${\bf p}$ be a graph parameter which satisfies the following conditions:
\begin{itemize}
\item[{\bf (i)}] ${\bf p}$ is closed under contractions, 
\item[{\bf (ii)}] ${\bf p}(\Gamma_{k})=\Omega(k^{2})$, and  
\item[{\bf (iii)}] for graphs in ${\cal G}$, ${\bf p}$ is computable in time $2^{O(\bw(G))} \cdot n^{O(1)}$.
\end{itemize}
Then,  {{\sc {$k$-Parameter Checking for }}\mbox{\rm\bf p} {\sc on ${\cal C}$}} belongs to $O^{*}(2^{O(\sqrt{k})})$-{\sf SUBEPT}.\end{theorem}

Problems for  which Theorem~\ref{thm:bidi_revised} (but not Theorem~\ref{thm:bidiminor}) proves {\sf SUBEPT}-membership for apex-minor free graph classes   are $k$-{\sc Dominating set}, $k$-r-{\sc Distance Dominating Set}, 
$k$-{\sc Face cover} (for planar graphs), $k$-{\sc Edge Dominating set},  $k$-{\sc clique-transversal set}, $k$-{\sc Connected Dominating set},  and others.

For surveys on the application of Bidimensionality Theory
on the design of sub-exponential parameterized algorithms, see~\cite{DornFT08sube,DemaineH07theb} (see~\cite{DemaineH08bidi}). 
Finally, further applications of bidimensionality, appeared in~\cite{SauT10sube,DemaineHaj05bidi,FominLST10bidi,FominLRS11bidi}.
Finally, we should stress that bidimensionality is not the only technique to derive subexponential 
parameterized algorithms. Alternative approaches have been proposed in~\cite{AlonLS09fast,DornFLRS10beyo,FominV11sube,Tazari10fast,Thilikos11fast}.

\subsection{Kernelization}
\label{kernelization}

Kernelization (see Definition~\ref{Definition:Kernelization} in Subsection~\ref{limker}) appears to be one of the most rapidly growing fields of parameterized computation. In this section we will
describe some of the most characteristic techniques and results in this field. All kernels that we present are polynomial.  We start with an easy case.

\subsubsection{An easy example}

Consider the following problem.\\

\noindent{\sc {$k$-Non-Blocker}}\\
\noindent {\sl Instance:}~~ A graph ${G}$ and an integer  ${k}\geq 1$.\\
{\sl {Parameter:} ${k}$}\\
{\sl Question:}  Does $G$ contain a set $S$ of at least  $k$ vertices such that each vertex in $S$ has a neighbor that is not in $S$?\\

A kernelization algorithm for the above problem runs as follows. 
The first step is to remove from $G$ all isolated vertices (clearly, they cannot be part of a  solution). Then, in the remaining graph $G'$, pick some vertex $v_{i}$ from each of its connected components
and define a partition of $V(G)$ into two 
sets $V_{0}$ and $V_{1}$ such that $V_{0}$ (resp. $V_{1}$) contains all vertices of $G'$ whose distance from some vertex $v_{i}$ is even (resp. odd).
W.l.o.g. we assume that $|V_{0}|\geq |V_{1}|$ (otherwise,  swap their indices). If $k\leq |V_{0}|$ then we know that $V_{0}$ is 
a solution to the problem and we return a {\tt YES}-instance such as $(K_{2},1)$. Otherwise we 
return $(G',k)$ that is a kernel whose graph has  at most $2k-2$ vertices. \\

Currently the best  known kernel for   {\sc $k$-Non-Blocker} has size $5/3\cdot k$~\cite{DehneFFPR06nonb}. We remark that the {\em dual} of the {\sc {$k$-Non-Blocker}} problem is the same problem parameterized by $k'=|V(G)|-k$. This problem is equivalent 
to the {\sc Dominating Set} problem that is {\sf {\sf W}[2]}-complete and therefore the existence of a  kernel for this problem is unlikely.


%
%
%
%
%
%
%

In the above example the kernelization algorithm makes a very elementary preprocessing of  the input. However, in other cases things 
are more complicated and the kernelization algorithm needs first to discard some parts or simplify the input
in order to reduce it to an equivalent one of polynomial (on $k$)  size. Our next example concerns a non-trivial case of this type.

\subsubsection{Simplifying structure}

We now proceed with an example where the construction of a kernel requires some more work. Consider the following problem:\\

\noindent{\sc {$k$-3-Dimensional-Matching}}\\
\noindent {\sl Instance:}~~ Three disjoint  sets $A,B,C$, a collection of triples ${\cal T}\subseteq A\times B \times C$  and an integer  ${k}\geq 1$.\\
{\sl {Parameter:} ${k}$}\\
{\sl Question:}  If there a subset ${\cal M}\subseteq {\cal T}$ where $|{\cal M}|\geq k$ and such that no two triples in ${\cal M}$ share a common element (i.e. ${\cal M}$ is a matching of ${\cal T}$)?
\\

Let ${\cal M'}$ be a maximal matching of ${\cal T}$. Clearly, such a matching can be found greedily in polynomial time. If 
$|{\cal M'|}\geq k$ then we are done (in this case,  return a trivial {\tt YES}-instance). Therefore, we assume that $|{\cal M}'|\leq k-1$ and let $S$ be the set of all elements in the triples 
of ${\cal M}'$. Certainly, $|S|\leq 3k-3$ and, by the maximality of ${\cal M}'$, $S$ will intersect  every triple in ${\cal T}$.
Suppose now that ${\cal T}$ contains a sub-collection ${\cal A}$ with at least $k+1$ triples that agree in two of their coordinates. W.l.o.g. we 
assume that ${\cal A}=\{(x,y,v_{i}) \mid i=1\ldots \rho\}$ for some $r\geq k+1$. Our first reduction rule is to replace $({\cal T},k)$ by 
$({\cal T}',k)$ where  ${\cal T}'={\cal T}\setminus \{(x,y,v_{i})\mid i=k+1,\ldots,\rho\}$. We claim that this reduction rule is {\em safe}, i.e. 
the new instance  $({\cal T}',k)$ is equivalent to the old one. Obviously, every matching of ${\cal T'}$ is also  a matching of ${\cal T}$.
Suppose now that ${\cal M}$ is a matching of ${\cal T}$ that is not any more a matching of ${\cal T}'$. 
Then we show that ${\cal M}$ can be safely replaced by another one that is  in ${\cal T}'$.
Indeed,  if ${\cal M}$  is not a matching of ${\cal T}'$, then 
one of the triples in ${\cal M}$ is missing from ${\cal T}'$ and therefore is of the form $T=(x,y,v_{j})$ for some $j\in\{k+1,\ldots,\rho\}$.
Notice that  if one of the  triples in ${\cal M}\setminus \{T\}$ intersects a triple in $\{(x,y,v_{i}) \mid i=1\ldots k\}$, then 
this intersection will be a vertex in $\{v_{i}\mid i=1,\ldots,k\}$. As $|{\cal M}\setminus \{T\}|<k$, one, say $T'$  of the $k$ triples in $\{(x,y,v_{i}) \mid i=1\ldots k\}$
will not be intersected by the triples in ${\cal M}\setminus \{T\}$. Therefore,  ${\cal M}^*={\cal M}\setminus\{T\}\cup\{T'\}$ is a matching of ${\cal T'}$
and the claim follows as $|{\cal M}^*|\geq k$. Notice that the first rule simply ``truncates'' all but $k$ triples agreeing in the same two coordinates.

We now assume that ${\cal T}$ is a collection of triples 
where no more than $k$ of them agree in the same two coordinates.
Suppose now that ${\cal T}$ contains a sub-collection ${\cal B}$  with more than $2(k-1)k+1$ triples,  all agreeing to one coordinate.
W.l.o.g. we assume that the agreeing coordinate is the first one.
The second reduction rule  removes from ${\cal T}$ all but $2(k-1)k+1$ of the elements of ${\cal B}$. Again using a 
pigeonhole argument, it follows that in the  $2(k-1)k+1$ surviving triples of ${\cal B}$ there is a subset ${\cal C}$ of  $2k-1$ triples where each two of them disagree in both second and third coordinate. Again, if a discarded triple is used in a solution ${\cal M}$, then the $k-1$ other triples 
cannot intersect more than $2k-2$ triples of ${\cal C}$ and therefore a ``surviving'' one can substitute the discarded one in ${\cal M}$.
Therefore, the second truncation is also safe and leaves an equivalent  instance ${\cal T}'$ where no more than $2(k-1)k+1$ of them agree in the same coordinate. Recall now that the elements of the set $S$ intersect all triples in ${\cal T}$. As $|S|\le 3(k-1)$, we obtain
that, in the equivalent instance $({\cal T}',k)$, ${\cal T}'$ has at most $3(k-1)\cdot (2(k-1)k+1)=O(k^{3})$ triples. We conclude that $p${\sc -3-Dimensional-Matching} has a polynomial size kernel.

The above kernelization was proposed  by Fellows, Knauer, Nishimura, Ragde, Rosamond, Stege, Thilikos, and Whitesides in~\cite{FellowsKNRRSTW08fast} and, combined with the color coding technique,
gave algorithms of total complexity $2^{O(k)}+O(n)$ for a wide variety of packing and matching problems.
For other related results, see~\cite{ChenFKX07para,JiaZC04anef, PrietoS04look,MathiesonPS04pack,Abu-Khzam09aqu}.

\subsubsection{Superoptimality in Integer Programming}

Perhaps the first appearance of the  kernelization idea can be traced to the work of  Nemhauser and Trotter for {\sc Vertex Cover}~\cite{NemT75vert}.
It leads to a kernel of size $2k$ and it was proposed as an alternative way to obtain a 2-approximation for the vertex cover of a graph.
To explain the idea of this kernelization we need to express the problem as an integer programming problem:

\begin{eqnarray*}
\mbox{minimize} & \sum_{v\in V(G)}x_{v} &   \mbox{subject to} \\
\forall_{\{v,u\}\in E(G)} &   x_{v}+x_{u} \geq  1 &\\
\forall_{v\in V(G)} & 0\leq x_{v}  \leq  1 &
\end{eqnarray*}
Any integer solution to the above problem assigns, for each $v\in V(G)$, an integer value from $\{0,1\}$ to the corresponding 
variable $x_{v}$ and yields to a solution to the {\sc Vertex Cover} problem. 
When we relax the problem so that 
the variables lie in the real interval $[0,1]$ the resulting LP
can be solved in polynomial time. Moreover, the value of the optimal LP solution
is clearly a lower bound on the minimum size of a vertex cover.
However, these 
solutions do not represent any more an actual solution to the vertex cover problem. 
In~\cite{NemT75vert}, it was proved 
that every real solution can be transformed, in polynomial time, to a {\em half-integer} one where 
for each $v\in V(G)$, $x_{v}\in\{0,\frac{1}{2},1\}$. Such a, still non-integer, solution
defines a partition of $V(G)$ into three sets $V_{0}$, $V_{\frac{1}{2}}, $ and $V_{1}$.
As shown in~\cite{NemT75vert} (see also~\cite{Ch02thes}), these sets satisfy the following properties:
\begin{itemize}
\item[(1)] (approximability) $|V_{\frac{1}{2}}|\leq 2\cdot  \vc(G[V_{\frac{1}{2}}])$.
\item[(2)] ({\sl Super-optimality}) every minimum vertex cover $S$ of $G$ satisfies $V_{1}\subseteq S\subseteq V_{1}\cup V_{\frac{1}{2}}$
\end{itemize}

Observe that no edges in $G$ have both endpoints in $V_{0}$ or one endpoint
in $V_{0}$ and the other in $V_{\frac{1}{2}}$. This fact, together with property (2),
implies that if  $S$ is an optimal solution, then $S\setminus V_{1}\subseteq V_{\frac{1}{2}}$
and  $S\setminus V_{1}$ is a vertex cover of $G[V_{\frac{1}{2}}]$. Therefore $\vc(G)\geq \vc(G[V_{\frac{1}{2}}])+|V_{1}|$.
Moreover, if $S_{\frac{1}{2}}$ is a vertex cover of $G[V_{\frac{1}{2}}]$, then $S_{\frac{1}{2}}\cup V_{1}$
is a vertex cover of $G$, which implies that $\vc(G)\leq \vc(G[V_{\frac{1}{2}}])+|V_{1}|$. We conclude that:
\begin{eqnarray}
\vc(G) &  =  & \vc(G[V_{\frac{1}{2}}])+|V_{1}|. \nonumber
\end{eqnarray}
The above relation reduces the problem of computing $\vc(G)$ to the one of computing $\vc(G[V_{\frac{1}{2}}])$.  Recall that, from Property (1),  
$|V_{\frac{1}{2}}|/2\leq \vc(G[V_{\frac{1}{2}}])\leq |V_{\frac{1}{2}}|$. Therefore, the size of $V_{\frac{1}{2}}$
is a $2$-approximation of $\vc(G[V_{\frac{1}{2}}])$ and thus $|V_{\frac{1}{2}}|+|V_{1}|$
is a  $2$-approximation of $\vc(G)$. However, this type of approximation provides 
also a kernel of size $2k$ for {\sc Vertex Cover} as follows.
We first compute, in polynomial time,  $V_{\frac{1}{2}}$ and $V_{1}$ and we ask whether $|V_{\frac{1}{2}}|\leq 2(k-|V_{1}|)$.
If the answer is negative, then, again from Property (1), we have that $\vc(G[V_{\frac{1}{2}}])> k-|V_{1}|$. This implies that 
$\vc(G)>k$ and we can safely report that $G$ has no vertex cover of size $\leq k$ or, more formally, just return a {\tt NO}-instance 
such as $(K_{2},0)$.
 If $|V_{\frac{1}{2}}|\leq 2(k-|V_{1}|)$
then we output $(G[V_{\frac{1}{2}}],k-|V_{1}|)$ as an equivalent instance whose graph has size at most $2k$.\medskip

A natural question is whether it is possible to produce a linear kernel of  $c\cdot k$ vertices for some $c<2$. 
Notice that any such kernel would imply a ``better than two'' approximation for the {\sc Vertex Cover}.

As mentioned in Subsection~\ref{limker} the non-approximability results of~\cite{DinurS02thei}
imply that we cannot expect 
a kernel  with less than $1.36 k$ vertices for {\sc Vertex Cover}. 
Better lower bounds can  be derived for the case of {\sc Planar Vertex Cover}.
Using a technique developed in~\cite{ChenFKX07para}, it  follows easily that, for every $\epsilon>0$, there is no kernel with a planar graph
of size $(\frac{4}{3}-\epsilon)\cdot k$ for this problem, unless {\sf P}$=${\sc NP}.

An  intriguing question is to prove the existence of  a polynomial kernel for the following meta-problem.\\

\noindent{\sc {$k$-Distance from $H$-minor Free}}\\
\noindent {\sl Instance:}~~ A graph ${G}$ and an integer  ${k}\geq 1$.\\
{\sl {Parameter:} ${k}$}\\
{\sl Question:}  Does $G$ contain a subset $S$ of at most $k$ vertices such that $G-S$ does not 
contain $H$ as a minor?\\

Clearly, the above problem is {\sc Vertex Cover} when $H=K_{2}$ and  {\sc Feedback Vertex Set} when $H=K_{3}$ and these are the only two cases
where a polynomial kernel is known. Actually, the research on a kernel for the {\sc Feedback Vertex Set} was quite extended (see~\cite{BurrageEFLMR06theu, Bodlaender07acub})
and the best up to now result is the $O(k^{2})$ kernel provided by Thomassé in~\cite{Thomasse09aqua}. Recently, a polynomial kernel 
was devised for the case where $H$ is the $\Theta_{c}$ graph, that is 
two vertices connected with $c$ parallel edges~\cite{FominLMPS11hitt}.
For more general 
cases of $H$  the  problem remains open.

\subsubsection{Extremal combinatorics}

Many polynomial kernels make strong use of results in extremal graph theory. An easy, though characteristic, example is the following problem:\\

\noindent{\sc {$k$-Max Leaf}}\\
\noindent {\sl Instance:}~~ A connected graph ${G}$ and an integer  ${k}\geq 1$.\\
{\sl {Parameter:} ${k}$}\\
{\sl Question:}  Does $G$ contain a spanning tree with at least $k$ leaves?\\

The above problem is {\sf NP}-complete~\cite{Lemke98inst}. A  kernel of at most  $8k$ vertices  for this problem is the procedure that repetitively applies the following reduction rules until none of them applies more:
\begin{itemize}

\item[R1:] If $G$ contains a vertex $v$ of degree $1$ with a neighbor of  degree 2 then set $(G,k)\leftarrow (G\setminus \{v\},k)$.
\item[R2:] If $G$ contains two vertices $v$ and $v'$ of degree 1 with the same neighbor, then set $(G,k)\leftarrow (G\setminus \{v\},k-1)$.
\item[R3:] If $G$ contains a chain of length at least 4 then set $(G,k)\leftarrow (G',k)$ 
where $G'$ is the graph obtained if we contract some of the edges of this chain.
(A {\em chain} is a path of length at least 2, where all internal vertices have degree 2, and with endpoints of degree greater than 2.)
\end{itemize}

It is not hard to see that all rules above are producing (simpler) equivalent instances.
A {\em reduced instance} $(G',k')$, is an instance of the problem where none of 
the above rules can be applied any more.  
What remains  to prove is that if in a reduced instance $(G',k')$, $G'$ has  least $8k'$ vertices,  then $(G',k')$ is a 
 {\tt NO}-instance.

We denote by $V_{1}$, $V_{2}$, and $V_{\geq 3}$ the 
set of 
vertices of $G'$ with degrees 1, 2, or $\geq 3$ respectively.
Because of the first  rule,
 vertices in $V_{1}$ are  adjacent only with vertices in $V_{3}$
 and, from the second rule, we have that $|V_{1}|\leq |V_{3}|$.
Moreover, from the third rule, all the chains of $G'$  have length at most $3$.

We aim to  prove that if $|V(G')|\geq 8 k'$, then $G'$ contains a spanning 
tree of $k'$ leaves. 
For this, we construct, using $G'$, an auxiliary graph  $H$
by  removing all vertices of degree 1 and  by replacing all chains of $G'$ by chains of length 2. 
Notice that $H$ does not have vertices of degree 1 and all its vertices of degree 2 -- we call them {\em subdivision vertices}  -- 
are internal vertices of some chain of length 2. We denote by $Q_{2}$ and $Q_{\geq 3}$ the vertices 
of $H$ that have degree 2 or $\geq 3$ respectively.
Notice that $|V_{\geq 3}|=|Q_{\geq 3}|$ which, combined with the fact that  $|V_{1}|\leq |V_{3}|$, implies that 
$|V_{1}|+|V_{\geq 3}|\leq 2\cdot |Q_{\geq 3}|$.
Moreover, each vertex in $Q_{2}$ corresponds to at most two vertices of $V_{2}$. Therefore 
$|V_{2}|\leq 2\cdot |Q_{2}|$. We conclude that $|V(G')|=|V_{1}|+|V_{2}|+|V_{\geq 3}|\leq 2\cdot |Q_{2}|+2\cdot |Q_{\geq 3}|=2\cdot |V(H)|$. As $|V(G')|\geq 8 k'$, we conclude that $|V(H)|\geq 4k'$. We call two subdivided edges of $H$ {\em siblings} if they have the same neighborhood.
We  add edges between siblings so that $H$ is transformed to a graph $H'$ with minimum degree 3.
It is easy to see that if $H'$ contains a spanning tree with at least $k'$ leaves then also $H$
contains a spanning tree with at least $k'$ leaves.  As $|V(H')|\geq 4 k'$, then from the main result of Kleitman and West in~\cite{KleitmanW91span}
we have  that $H'$ (and therefore $H$) contains a spanning tree of with  least $k'$ leaves. By the construction of $H$,
it easily follows that this spanning tree of $H$ can be extended  to a spanning tree in $G'$ with at least the same number of leaves and we are done.\medskip

Clearly, the correctness of the above kernel is based on the combinatorial result of \cite{KleitmanW91span}. By using more reduction
rules and more refined arguments  it is possible to construct even better kernels for {\sc {$k$-Max Leaf}}.
For instance, in~\cite{Estivill-CastroFLR05fpti} Estivill-Castro, Fellows, Langston, and Rosamond give a $3.75\cdot k$ kernel for this problem. Notice that the parameterized complexity of the problem
changes drastically if we ask  for a spanning tree with {\sl at most} $k$ leaves. Even for $k=2$, this is 
an {\sf NP}-complete problem. Therefore there is no $n^{O(k)}$ algorithm for this parametrization, unless ${\sf P}={\sf NP}$.

\subsubsection{Kernels for sparse graphs}

Clearly, one cannot expect that a kernel exists for a problem that is hard for some level of the W-hierarchy. However, such problems may become fixed parameter tractable when their instances
 are restricted to sparse structures. In case of graph-theoretic
problems such a sparsity criterion is usually planarity  or the exclusion of certain types of graphs as a minor.
The prototype of such a problem was {\sc Planar Dominating Set} that, as we mentioned in Subsection~\ref{subexpalg}, belongs to {\sf SUBFPT}. The first kernel for this problem was given by the celebrated result 
of Alber, Fellows, and Niedermeier in~\cite{AlberFN04poly} and was of size $335k$. The kernel of~\cite{AlberFN04poly} is using the fact that 
the sphere in which 
 the input graph is embedded, can  be decomposed into a set of  $O(k)$ regions such that each vertex inside a region
is in the neighborhood of its boundary vertices. This last property, makes it possible to apply a suitable reduction in each region
and reduce its vertices to a constant size. Because of planarity,  there are  $O(k)$ regions in total and this implies the existence of a linear kernel.
Based on the same idea and a more refined analysis, this kernel was improved to one of size $67k$ in~\cite{ChenFKX07para}.
The idea of a region decomposition was further extended in~\cite{GuoNiedermeier2007line} for several problems such as  {\sc Connected Vertex Cover}, {\sc Minimum Edge Dominating Set}, {\sc Maximum Triangle Packing}, and {\sc Efficient Dominating Set} on planar graphs.
While in all the above results the reduction rules were particular for each problem, it appeared that the whole idea can be vastly extended 
to a general meta-algorithmic framework that we will describe in the rest of this section.

\paragraph{The compactness condition.}
Let ${\cal G}_g$ be the family of all graphs that can be embedded in a surface  $\Sigma$ of Euler-genus at most $g$. Given a graph $G$ embedded in a surface  $\Sigma$ of Euler-genus $g$, and a set $S$, we define ${\bf R}_{G}^{r}(S)$ to be the set of all vertices of $G$ whose radial distance from some vertex of $S$ is at most $r$. The {\em radial distance} between two vertices $x,y$ is the minimum length of an alternating sequence of vertices and faces starting from $x$ and ending in $y$, such that every two consecutive elements of this sequence are incident 
to each other.

We consider parameterized problems where $\Pi \subseteq {\cal G}_{g} \times \mathbb{N}$, i.e. we impose bounded genus as a promise condition.
We say that a parameterized problem $\Pi \subseteq {\cal G}_{g} \times \mathbb{N}$ is \emph{\sl compact} if there exists an integer $r$ such that for all 
$(G=(V,E),k)\in\Pi$, there is an embedding of $G$ in a surface  $\Sigma$ of Euler-genus at most $g$ and a set $S \subseteq V$ such that $|S|\leq r \cdot k$ and  ${\bf R}_{G}^{r}(S)=V$.  The compactness condition is what makes it possible to  construct  a region decomposition for the instances of a parameterized problem. The region decomposition is constructed on the surface where the input graph is embedded.
For the automation of the reduction rules that are applied in the regions we demand the expressibility of the problem 
by a certain logic.


%

\paragraph*{Counting Monadic Second Order Logic.} {Counting Monadic Second-Order Logic} (CMSOL) is an extension of
MSOL defined in~Subsection~\ref{ftpfol}  where, in addition to the usual features of monadic second-order logic, we have atomic formulas testing whether the cardinality of a set is equal 
to $n$ modulo $p$, where $n$ and $p$ are integers such that $ 0\leq n<p $ and $p\geq 2$ (see~\cite{ArnborgLS91easy,Courcelle90them,Courcelle97thee}). So essentially CMSOL 
is MSOL with  atomic formulas of the following type:
\begin{quote}
{\sl If $U$ denotes a set $X$, then 
$\mathbf{card}_{n,p}(U) = \mathbf{true}$ if and only if $|X |$ is $n~\mathbf{mod}~p$. }
\end{quote}
%

%

In a {\sc $k$-min-CMSOL} parameterized graph problem $\Pi \subseteq {\cal G}_{g} \times \mathbb{N}$, we are given a graph $G=(V,E)$ and an integer $k$ as input. The objective is to decide whether there is a vertex/edge set $S$ of size at most $k$ such that the CMSOL-expressible predicate $P_\Pi(G,S)$ is satisfied.
The \emph{\sl annotated} version $\Pi^{\alpha}$ of a {\sc $k$-min/eq/max-CMSOL} problem $\Pi$ is defined as follows. The input is a triple $(G=(V,E),Y,k)$ where $G$ is a graph, $Y\subseteq V$ is a set of black vertices, and $k$ is a non-negative integer.  In the {\em annotated version} of a {\sc $k$-min-CMSOL} graph problem, $S$ is additionally required to be a subset of $Y$. 
Finally, for a parameterized problem $\Pi \subseteq {\cal G}_g\times \mathbb{N}$, let 
$\overline{\Pi}\subseteq {\cal G}_g \times \mathbb{N}$ denote the set of all no-instances of $\Pi$.

\paragraph{Kernelization meta-theorems.}
The first meta-algorithmic  results on kernels appeared by Bodlaender, Fomin, Lokshtanov, Penninkx, Saurabh, and Thilikos in~\cite{BodlaenderFLPS09meta} and by Fomin, Lokshtanov, Saurabh, and Thilikos in~\cite{FominLST10bidi}. As a sample, we present the following one:

\begin{theorem}[\!\!\cite{BodlaenderFLPS09meta}]
\label{cor:antonormal}
Let $\Pi\subseteq {\cal G}_g\times \mathbb{N}$ be an  {\sf NP}-complete 
{\sc $k$-min-CMSOL} parameterized problem such that  either $\Pi$ or $\overline{\Pi}$ is compact and $\Pi^{\alpha}$ is in $\classNP{}$.
Then $\Pi$ admits a polynomial kernel.
\end{theorem}

The above theorem is essentially a corollary of the fact, proven in~\cite{BodlaenderFLPS09meta}, that if  $\Pi\subseteq {\cal G}_g\times \mathbb{N}$ is  a  {\sc $k$-min-CMSOL} parameterized  problem and either  $\Pi$ or $\overline{\Pi}$ is compact, then its annotated version $\Pi^{\alpha}$  admits a  quadratic kernel.
Another condition that, together with the compactness condidtion, allows the derivation of a linear kernel for a  {\sc $k$-min-CMSOL} parameterized  problem
is the one having {\em finite integer index}, introduced in~\cite{BodlaenderV01redu}. This condition, combined with bidimensionality theory (presented in Subsection~\ref{subexpalg}) was used in~\cite{FominLST10bidi}
in order to derive kernelization meta-theorems for more general graph classes, namely classes excluding some graph (apex or general) as a minor. \smallskip

Theorem~\ref{cor:antonormal}  immediately implies the existence of polynomial kernels for a wide family of problems. 
However, the potential of  this type of meta-algorithmic results has not been fully investigated yet. \medskip

We conclude our kernelization subsection by  mentioning that it cannot be more than incomplete. For a small sample of  
recent results on the existence of polynomial kernels for several parameterized problems, see~\cite{GutinIMY11ever,CrowstonGJKR10syst,AG08TechReport,PhilipRS09,AlonGKSY10solv,GutinKSY11apro,DellM10sati,KratschMW10para,CyganPPW10kern,BodlaenderJK11cros,BodlaenderJK11prep,HermelinHKW11para,FominLMPS11hitt,JansenB11vert}.

\subsection{(Much) more on parameterized algorithms}

Certainly, there are many issues on parameterized algorithm design that
are not (and cannot be) covered here. Among them, we should mention parameterized algorithms for counting problems~\cite{McCartin06para,DiazST08effi,NishimuraRT05para, FlumG02thep,Thurley07kern,DemaineHM10para}, parameterized approximation~\cite{DowneyFM08para,Jansen09para,DowneyFM06para,CaiH06fixe,ChenGG06onpa,MarxR09cons}, 
 and parameterized  parallelization~\cite{Cesati98para}.   
 Also we should mention that new powerful  techniques emerge quite rapidly in parameterized algorithm design, such as 
 the use of {\sl important seprators} for the {\sf FPT}-algorithm proposed for the {\sc $k$-Edge Multicut} and 
 the {\sc $k$-Vertex  Multicut} problems by Dániel Marx and Igor Razgon in~\cite{MarxR11fixed} (see also~\cite{BousquetDT11mult}). 
 
\bigskip


\paragraph{Acknowledgements} Many thanks to Mike Fellows, Fedor V. Fomin, Bart M. P. Jansen, Stavros G. Kolliopoulos, Moritz M{\"u}ller, Saket Saurabh, and Geoff Whittle
for corrections and  helpful comments.



\addcontentsline{toc}{section}{\bf \ \ \ \ References}


\begin{thebibliography}{100}

\bibitem{AbrahamsonFEM89}
K.~Abrahamson, M.~Fellows, J.~Ellis, and M.~Mata.
\newblock On the complexity of fixed parameter problems.
\newblock In {\em 30th Annual IEEE Symposium on Foundations of Computer Science
  (FOCS '89)}, pages 210--215, Los Alamitos, CA, USA, 1989. IEEE Computer
  Society.

\bibitem{AbrahamsonDF95fixed}
K.~A. Abrahamson, R.~G. Downey, and M.~R. Fellows.
\newblock Fixed-parameter tractability and completeness. {I}{V}. {O}n
  completeness for \mbox{{W}[{P}]} and {P}{S}{P}{A}{C}{E} analogues.
\newblock {\em Annals of Pure and Applied Logic}, 73(3):235--276, 1995.

\bibitem{Abu-Khzam09aqu}
F.~Abu-Khzam.
\newblock A quadratic kernel for 3-set packing.
\newblock In J.~Chen and S.~Cooper, editors, {\em Theory and Applications of
  Models of Computation}, volume 5532 of {\em Lecture Notes in Computer
  Science}, pages 81--87. Springer Berlin / Heidelberg, 2009.

\bibitem{Abu-KhzamCFLSS04kern}
F.~Abu-Khzam, R.~L. Collins, M.~R. Fellows, M.~A. Langston, W.~H. Suters, and
  C.~T. Symons.
\newblock Kernelization algorithms for the vertex cover problem: Theory and
  experiments.
\newblock In {\em ALENEX/ANALC}, pages 62--69, 2004.

\bibitem{Adler08open}
I.~Adler.
\newblock Open problems related to computing obstruction sets.
\newblock Manuscript, September 2008.

\bibitem{AdlerGK08comp}
I.~Adler, M.~Grohe, and S.~Kreutzer.
\newblock Computing excluded minors.
\newblock In {\em Proceedings of the nineteenth annual ACM-SIAM symposium on
  Discrete algorithms}, SODA '08, pages 641--650, Philadelphia, PA, USA, 2008.
  Society for Industrial and Applied Mathematics.

\bibitem{AdlerKKLST12tigh}
I.~Adler, S.~G. Kolliopoulos, P.~K. Krause, D.~Lokshtanov, S.~Saurabh, and
  D.~M. Thilikos.
\newblock Tight bounds for linkages in planar graphs.
\newblock In {\em Proceedings of the 38th International Colloquium on Automata,
  Languages and Programming (ICALP 2011)}, 2011.

\bibitem{AlberBFKN02}
J.~Alber, H.~L. Bodlaender, H.~Fernau, T.~Kloks, and R.~Niedermeier.
\newblock Fixed parameter algorithms for dominating set and related problems on
  planar graphs.
\newblock {\em Algorithmica}, 33(4):461--493, 2002.

\bibitem{AlberBFN00}
J.~Alber, H.~L. Bodlaender, H.~Fernau, and R.~Niedermeier.
\newblock Fixed parameter algorithms for planar dominating set and related
  problems.
\newblock In {\em 6th Scandinavian Workshop on Algorithm Theory---SWAT 2000
  (Bergen)}, pages 97--110. Springer, Berlin, 2000.

\bibitem{AlberFN05expe}
J.~Alber, F.~Dorn, and R.~Niedermeier.
\newblock Experimental evaluation of a tree decomposition-based algorithm for
  vertex cover on planar graphs.
\newblock {\em Discrete Applied Mathematics}, 145(2):219 -- 231, 2005.
\newblock Structural Decompositions, Width Parameters, and Graph Labelings.

\bibitem{AlberFN04poly}
J.~Alber, M.~R. Fellows, and R.~Niedermeier.
\newblock Polynomial-time data reduction for dominating set.
\newblock {\em J. Assoc. Comput. Mach.}, 51(3):363--384, 2004.

\bibitem{AlberN02impr}
J.~Alber and R.~Niedermeier.
\newblock Improved tree decomposition based algorithms for domination-like
  problems.
\newblock In S.~Rajsbaum, editor, {\em LATIN 2002: Theoretical Informatics},
  volume 2286 of {\em Lecture Notes in Computer Science}, pages 221--233.
  Springer Berlin / Heidelberg, 2002.

\bibitem{AlekhnovichR01reso}
M.~Alekhnovich and A.~Razborov.
\newblock Resolution is not automatizable unless {W}[{P}] is tractable.
\newblock {\em Foundations of Computer Science, Annual IEEE Symposium on},
  0:210, 2001.

\bibitem{AlonGKSY10solv}
N.~Alon, G.~Gutin, E.~J. Kim, S.~Szeider, and A.~Yeo.
\newblock Solving {MAX}-r-{SAT} above a tight lower bound.
\newblock In {\em Proceedings of the Twenty-First Annual ACM-SIAM Symposium on
  Discrete Algorithms}, SODA '10, pages 511--517, Philadelphia, PA, USA, 2010.
  Society for Industrial and Applied Mathematics.

\bibitem{AG08TechReport}
N.~Alon and S.~Gutner.
\newblock Kernels for the dominating set problem on graphs with an excluded
  minor.
\newblock Technical report, ECCC, 2008.

\bibitem{AlonG09line}
N.~Alon and S.~Gutner.
\newblock Linear time algorithms for finding a dominating set of fixed size in
  degenerated graphs.
\newblock {\em Algorithmica}, 54(4):544--556, 2009.

\bibitem{AlonLS09fast}
N.~Alon, D.~Lokshtanov, and S.~Saurabh.
\newblock Fast {F}{A}{S}{T}.
\newblock In S.~Albers, A.~Marchetti-Spaccamela, Y.~Matias, S.~Nikoletseas, and
  W.~Thomas, editors, {\em Automata, Languages and Programming}, volume 5555 of
  {\em Lecture Notes in Computer Science}, pages 49--58. Springer Berlin /
  Heidelberg, 2009.

\bibitem{AlonYZ95colo}
N.~Alon, R.~Yuster, and U.~Zwick.
\newblock Color-coding.
\newblock {\em J. Assoc. Comput. Mach.}, 42(4):844--856, 1995.

\bibitem{AlonYZ08colo}
N.~Alon, R.~Yuster, and U.~Zwick.
\newblock Color coding.
\newblock In {\em Encyclopedia of Algorithms}. Springer, 2008.

\bibitem{AminiFS08impl}
O.~Amini, F.~Fomin, and S.~Saurabh.
\newblock Implicit branching and parameterized partial cover problems (extended
  abstract).
\newblock In R.~Hariharan, M.~Mukund, and V.~Vinay, editors, {\em IARCS Annual
  Conference on Foundations of Software Technology and Theoretical Computer
  Science (FSTTCS 2008)}, Dagstuhl, Germany, 2008. Schloss Dagstuhl -
  Leibniz-Zentrum fuer Informatik, Germany.

\bibitem{ArnborgLS91easy}
S.~Arnborg, J.~Lagergren, and D.~Seese.
\newblock Easy problems for tree-decomposable graphs.
\newblock {\em Journal of Algorithms}, 12:308--340, 1991.

\bibitem{Arora96poly}
S.~Arora.
\newblock Polynomial time approximation schemes for {E}uclidean {TSP} and other
  geometric problems.
\newblock In {\em 37th {A}nnual {S}ymposium on {F}oundations of {C}omputer
  {S}cience ({B}urlington, {VT}, 1996)}, pages 2--11. IEEE Comput. Soc. Press,
  Los Alamitos, CA, 1996.

\bibitem{Arora97near}
S.~Arora.
\newblock Nearly linear time approximation schemes for euclidean tsp and other
  geometric problems.
\newblock In {\em Proceedings of the 38th Annual Symposium on Foundations of
  Computer Science}, pages 554--, Washington, DC, USA, 1997. IEEE Computer
  Society.

\bibitem{BalasubramanianFR98anim}
R.~Balasubramanian, M.~Fellows, and V.~Raman.
\newblock An improved fixed-parameter algorithm for vertex cover.
\newblock {\em Information Proccessing Letters}, 65:163--168, 1998.

\bibitem{Bazgan95sche}
C.~Bazgan.
\newblock Sch{\'e}mas d'approximation et complexit{\'e} param{\'e}tr{\'e}e.
\newblock Technical report, Rapport de stage de DEA d'Informatique, Orsay,
  1995.

\bibitem{BerthomeN08auni}
P.~Berthom\'e and N.~Nisse.
\newblock A unified {FPT} algorithm for width of partition functions.
\newblock Technical Report INRIA-00321766, INRIA, September 2008.

\bibitem{BetzlerDKN08para}
N.~Betzler, M.~Fellows, C.~Komusiewicz, and R.~Niedermeier.
\newblock Parameterized algorithms and hardness results for some graph motif
  problems.
\newblock In P.~Ferragina and G.~Landau, editors, {\em Combinatorial Pattern
  Matching}, volume 5029 of {\em Lecture Notes in Computer Science}, pages
  31--43. Springer Berlin / Heidelberg, 2008.

\bibitem{BetzlerNU06tree}
N.~Betzler, R.~Niedermeier, and J.~Uhlmann.
\newblock Tree decompositions of graphs: Saving memory in dynamic programming.
\newblock {\em Discrete Optimization}, 3(3):220 -- 229, 2006.
\newblock Graphs and Combinatorial Optimization.

\bibitem{BjorklundHKK07four}
A.~Bj{\"o}rklund, T.~Husfeldt, P.~Kaski, and M.~Koivisto.
\newblock Fourier meets m{\"o}bius: fast subset convolution.
\newblock In {\em STOC}, pages 67--74, 2007.

\bibitem{BockerBBT08goin}
S.~B{\"o}cker, S.~Briesemeister, Q.~Bui, and A.~Truss.
\newblock Going weighted: Parameterized algorithms for cluster editing.
\newblock In B.~Yang, D.-Z. Du, and C.~Wang, editors, {\em Combinatorial
  Optimization and Applications}, volume 5165 of {\em Lecture Notes in Computer
  Science}, pages 1--12. Springer Berlin / Heidelberg, 2008.

\bibitem{Bodlaender07acub}
H.~Bodlaender.
\newblock A cubic kernel for feedback vertex set.
\newblock In W.~Thomas and P.~Weil, editors, {\em STACS 2007}, volume 4393 of
  {\em Lecture Notes in Computer Science}, pages 320--331. Springer Berlin /
  Heidelberg, 2007.

\bibitem{BodlaenderFLPS09meta}
H.~Bodlaender, F.~Fomin, D.~Lokshtanov, E.~Penninkx, S.~Saurabh, and
  D.~Thilikos.
\newblock ({M}eta) kernelization.
\newblock In {\em Proc. of the 50th Annual IEEE Symposium on Foundations of
  Computer Science, (FOCS 2009)}, 2009.

\bibitem{Bodlaender90poly}
H.~L. Bodlaender.
\newblock Polynomial algorithms for graph isomorphism and chromatic index on
  partial $k$-trees.
\newblock {\em Journal of Algorithms}, 11:631--643, 1990.

\bibitem{Bodlaender96alin}
H.~L. Bodlaender.
\newblock A linear-time algorithm for finding tree-decompositions of small
  treewidth.
\newblock {\em SIAM J. Comput.}, 25(6):1305--1317, 1996.

\bibitem{BodlaenderDFH09onpr}
H.~L. Bodlaender, R.~G. Downey, M.~R. Fellows, and D.~Hermelin.
\newblock On problems without polynomial kernels.
\newblock {\em J. Comput. Syst. Sci.}, 75:423--434, December 2009.

\bibitem{BodlaenderFH04beyo}
H.~L. Bodlaender, M.~R. Fellows, and M.~T. Hallett.
\newblock Beyond np-completeness for problems of bounded width (extended
  abstract): hardness for the {W} hierarchy.
\newblock In {\em Proceedings of the twenty-sixth annual ACM symposium on
  Theory of computing}, STOC '94, pages 449--458, New York, NY, USA, 1994. ACM.

\bibitem{BodlaenderFT09deri}
H.~L. Bodlaender, M.~R. Fellows, and D.~M. Thilikos.
\newblock Derivation of algorithms for cutwidth and related graph layout
  parameters.
\newblock {\em Journal of Computer and System Sciences}, 75(4):231--244, 2009.

\bibitem{BodlaenderJK11cros}
H.~L. Bodlaender, B.~M.~P. Jansen, and S.~Kratsch.
\newblock {Cross-Composition: A New Technique for Kernelization Lower Bounds}.
\newblock In T.~Schwentick and C.~D{\"u}rr, editors, {\em 28th International
  Symposium on Theoretical Aspects of Computer Science (STACS 2011)}, volume~9
  of {\em Leibniz International Proceedings in Informatics (LIPIcs)}, pages
  165--176, Dagstuhl, Germany, 2011. Schloss Dagstuhl--Leibniz-Zentrum fuer
  Informatik.

\bibitem{BodlaenderJK11prep}
H.~L. Bodlaender, B.~M.~P. Jansen, and S.~Kratsch.
\newblock Preprocessing for treewidth: A combinatorial analysis through
  kernelization.
\newblock In {\em Proceedings of the 38th International Colloquium on Automata,
  Languages and Programming (ICALP 2009)}, 2011.

\bibitem{BodlaenderK96effi}
H.~L. Bodlaender and T.~Kloks.
\newblock Efficient and constructive algorithms for the pathwidth and treewidth
  of graphs.
\newblock {\em J. Algorithms}, 21(2):358--402, 1996.

\bibitem{BodlaenderT97cons}
H.~L. Bodlaender and D.~M. Thilikos.
\newblock Constructive linear time algorithms for branchwidth.
\newblock In {\em Automata, languages and programming (Bologna, 1997)}, volume
  1256 of {\em Lecture Notes in Computer Science}, pages 627--637. Springer,
  Berlin, 1997.

\bibitem{BodlaenderT04comp}
H.~L. Bodlaender and D.~M. Thilikos.
\newblock Computing small search numbers in linear time.
\newblock In {\em IWPEC}, pages 37--48, 2004.

\bibitem{BodlaenderTY08anal}
H.~L. Bodlaender, S.~Thomass{\'e}, and A.~Yeo.
\newblock Analysis of data reduction: Transformations give evidence for
  non-existence of polynomial kernels.
\newblock Technical Report UU-CS-2008-030, Department of Information and
  Computing Sciences, Utrecht University, 2008.

\bibitem{BodlaenderV01redu}
H.~L. Bodlaender and B.~van Antwerpen-de Fluiter.
\newblock Reduction algorithms for graphs of small treewidth.
\newblock {\em Inf. Comput.}, 167:86--119, June 2001.

\bibitem{BoriePT92auto}
R.~B. Borie, R.~G. Parker, and C.~A. Tovey.
\newblock Automatic generation of linear-time algorithms from predicate
  calculus descriptions of problems on recursively constructed graph families.
\newblock {\em Algorithmica}, 7:555--581, 1992.

\bibitem{BousquetDT11mult}
N.~Bousquet, J.~Daligault, and S.~Thomass{\'e}.
\newblock Multicut is fpt.
\newblock In {\em STOC}, pages 459--468, 2011.

\bibitem{BuhrmanH08npha}
H.~Buhrman and J.~M. Hitchcock.
\newblock {NP}-hard sets are exponentially dense unless co{NP}$\subseteq$
  {NP}/poly.
\newblock In {\em IEEE Conference on Computational Complexity}, pages 1--7,
  2008.

\bibitem{BurrageEFLMR06theu}
K.~Burrage, V.~Estivill-Castro, M.~Fellows, M.~Langston, S.~Mac, and
  F.~Rosamond.
\newblock The undirected feedback vertex set problem has a poly($k$) kernel.
\newblock In H.~Bodlaender and M.~Langston, editors, {\em Parameterized and
  Exact Computation}, volume 4169 of {\em Lecture Notes in Computer Science},
  pages 192--202. Springer Berlin / Heidelberg, 2006.

\bibitem{BussG93nond}
J.~F. Buss and J.~Goldsmith.
\newblock Nondeterminism within {${\rm P}$}.
\newblock {\em SIAM J. Comput.}, 22(3):560--572, 1993.

\bibitem{CaiCDF97onth}
L.~Cai, J.~Chen, R.~G. Downey, and M.~R. Fellows.
\newblock On the parameterized complexity of short computation and
  factorization.
\newblock {\em Archive for Mathematical Logic}, 36:321--337, 1997.

\bibitem{CaiFJR07thec}
L.~Cai, M.~Fellows, D.~Juedes, and F.~Rosamond.
\newblock The complexity of polynomial-time approximation.
\newblock {\em Theor. Comp. Sys.}, 41:459--477, October 2007.

\bibitem{CaiH06fixe}
L.~Cai and X.~Huang.
\newblock Fixed-parameter approximation: Conceptual framework and
  approximability results.
\newblock In H.~Bodlaender and M.~Langston, editors, {\em Parameterized and
  Exact Computation}, volume 4169 of {\em Lecture Notes in Computer Science},
  pages 96--108. Springer Berlin / Heidelberg, 2006.

\bibitem{CaiJ01sube}
L.~Cai and D.~Juedes.
\newblock Subexponential parameterized algorithms collapse the {$W$}-hierarchy.
\newblock In {\em Automata, languages and programming}, volume 2076 of {\em
  Lecture Notes in Comput. Sci.}, pages 273--284. Springer, Berlin, 2001.

\bibitem{CaiJ03onth}
L.~Cai and D.~Juedes.
\newblock On the existence of subexponential parameterized algorithms.
\newblock {\em Journal of Computer and System Sciences}, 67(4):789 -- 807,
  2003.
\newblock Parameterized Computation and Complexity 2003.

\bibitem{CattellDDFL00onco}
K.~Cattell, M.~J. Dinneen, R.~G. Downey, M.~R. Fellows, and M.~A. Langston.
\newblock On computing graph minor obstruction sets.
\newblock {\em Theor. Comput. Sci.}, 233:107--127, February 2000.

\bibitem{Cesati02perf}
M.~Cesati.
\newblock Perfect code is {${\rm W}[1]$}-complete.
\newblock {\em Inform. Process. Lett.}, 81(3):163--168, 2002.

\bibitem{Cesati03thet}
M.~Cesati.
\newblock The {T}uring way to parameterized complexity.
\newblock {\em Journal of Computer and System Sciences}, 67(4):654 -- 685,
  2003.
\newblock Parameterized Computation and Complexity 2003.

\bibitem{CesatiF96spar}
M.~Cesati and M.~R. Fellows.
\newblock Sparse parameterized problems.
\newblock {\em Annals of Pure and Applied Logic}, 82(1):1 -- 15, 1996.

\bibitem{Cesati98para}
M.~Cesati and M.~D. Ianni.
\newblock Parameterized parallel complexity.
\newblock In {\em Proceedings of the 4th International Euro-Par Conference on
  Parallel Processing}, Euro-Par '98, pages 892--896, London, UK, 1998.
  Springer-Verlag.

\bibitem{CDRST03}
J.~Cheetham, F.~Dehne, A.~Rau{-}Chaplin, U.~Stege, and P.~J. Taillon.
\newblock Solving large {FPT} problems on coarse-grained parallel machines.
\newblock {\em Journal of Computer and System Sciences}, 67(4):691--706, 2003.

\bibitem{ChekuriK00apta}
C.~Chekuri and S.~Khanna.
\newblock A {PTAS} for the multiple knapsack problem.
\newblock In {\em Proceedings of the {E}leventh {A}nnual {ACM}-{SIAM}
  {S}ymposium on {D}iscrete {A}lgorithms ({S}an {F}rancisco, {CA}, 2000)},
  pages 213--222, New York, 2000. ACM.

\bibitem{ChenCFHJKX05tigh}
J.~Chen, B.~Chor, M.~Fellows, X.~Huang, D.~Juedes, I.~A. Kanj, and G.~Xia.
\newblock Tight lower bounds for certain parameterized {NP}-hard problems.
\newblock {\em Inf. Comput.}, 201:216--231, September 2005.

\bibitem{ChenFKX07para}
J.~Chen, H.~Fernau, I.~A. Kanj, and G.~Xia.
\newblock Parametric duality and kernelization: Lower bounds and upper bounds
  on kernel size.
\newblock {\em SIAM J. Comput.}, 37(4):1077--1106, 2007.

\bibitem{ChenFLSV08impr}
J.~Chen, F.~V. Fomin, Y.~Liu, S.~Lu, and Y.~Villanger.
\newblock Improved algorithms for feedback vertex set problems.
\newblock {\em J. Comput. Syst. Sci.}, 74:1188--1198, November 2008.

\bibitem{ChenFJK04usin}
J.~Chen, D.~K. Friesen, W.~Jia, and I.~A. Kanj.
\newblock Using nondeterminism to design efficient deterministic algorithms.
\newblock {\em Algorithmica}, 40(2):83--97, 2004.

\bibitem{ChenKPSX03genu}
J.~Chen, I.~Kanj, L.~Perkovic, E.~Sedgwick, and G.~Xia.
\newblock Genus charactterizes the complexity of graph problems: some tight
  results.
\newblock In {\em Proc. of the 30th International Colloquium on Automata,
  Languages, and Programming ICALP, Eindhoven, 2003}. Springer Verlag, Lecture
  Notes in Computer Science, vol. 2719, 2003.

\bibitem{ChenKJ01vert}
J.~Chen, I.~A. Kanj, and W.~Jia.
\newblock Vertex cover: further observations and further improvements.
\newblock {\em J. Algorithms}, 41(2):280--301, 2001.

\bibitem{ChenKX10impr}
J.~Chen, I.~A. Kanj, and G.~Xia.
\newblock Improved upper bounds for vertex cover.
\newblock {\em Theor. Comput. Sci.}, 411:3736--3756, September 2010.

\bibitem{ChenLJ00impr}
J.~Chen, L.~Liu, and W.~Jia.
\newblock Improvement on vertex cover for low-degree graphs.
\newblock {\em Networks}, 35(4):253--259, 2000.

\bibitem{ChenLOR08afix}
J.~Chen, Y.~Liu, S.~Lu, B.~O'sullivan, and I.~Razgon.
\newblock A fixed-parameter algorithm for the directed feedback vertex set
  problem.
\newblock {\em J. ACM}, 55:21:1--21:19, November 2008.

\bibitem{ChenLSZ07impr}
J.~Chen, S.~Lu, S.-H. Sze, and F.~Zhang.
\newblock Improved algorithms for path, matching, and packing problems.
\newblock In {\em Proceedings of the eighteenth annual ACM-SIAM symposium on
  Discrete algorithms}, SODA '07, pages 298--307, Philadelphia, PA, USA, 2007.
  Society for Industrial and Applied Mathematics.

\bibitem{ChenM08onpa}
J.~Chen and J.~Meng.
\newblock On parameterized intractability: Hardness and completeness.
\newblock {\em Computer J.}, 51(1):39--59, 2008.

\bibitem{ChenM01apol}
J.~Chen and A.~Miranda.
\newblock A polynomial time approximation scheme for general multiprocessor job
  scheduling.
\newblock {\em SIAM J. Comput.}, 31(1):1--17 (electronic), 2001.

\bibitem{ChenF06onmi}
Y.~Chen and J.~Flum.
\newblock On miniaturized problems in parameterized complexity theory.
\newblock {\em Theor. Comput. Sci.}, 351:314--336, February 2006.

\bibitem{ChenFM09lowe}
Y.~Chen, J.~Flum, and M.~M{\"u}ller.
\newblock Lower bounds for kernelizations and other preprocessing procedures.
\newblock In K.~Ambos-Spies, B.~L{\"o}we, and W.~Merkle, editors, {\em
  Mathematical Theory and Computational Practice}, volume 5635 of {\em Lecture
  Notes in Computer Science}, pages 118--128. Springer Berlin / Heidelberg,
  2009.

\bibitem{ChenG07anis}
Y.~Chen and M.~Grohe.
\newblock An isomorphism between subexponential and parameterized complexity
  theory.
\newblock {\em SIAM J. Comput.}, 37:1228--1258, November 2007.

\bibitem{ChenGG06onpa}
Y.~Chen, M.~Grohe, and M.~Gr{\"u}ber.
\newblock On parameterized approximability.
\newblock In H.~Bodlaender and M.~Langston, editors, {\em Parameterized and
  Exact Computation}, volume 4169 of {\em Lecture Notes in Computer Science},
  pages 109--120. Springer Berlin / Heidelberg, 2006.

\bibitem{Ch02thes}
J.~Chleb{\'\i}kov{\'a}.
\newblock The structure of obstructions to treewidth and pathwidth.
\newblock {\em Discrete Applied Mathematics}, 120(1-3):61--71, 2002.

\bibitem{Courcelle90them}
B.~Courcelle.
\newblock The monadic second-order logic of graphs. {I}. {R}ecognizable sets of
  finite graphs.
\newblock {\em Information and Computation}, 85(1):12--75, 1990.

\bibitem{Courcelle92them}
B.~Courcelle.
\newblock The monadic second-order logic of graphs. {III}.
  {T}ree-decompositions, minors and complexity issues.
\newblock {\em RAIRO Inform. Th\'eor. Appl.}, 26(3):257--286, 1992.

\bibitem{Courcelle97thee}
B.~Courcelle.
\newblock The expression of graph properties and graph transformations in
  monadic second-order logic.
\newblock {\em Handbook of Graph Grammars}, pages 313--400, 1997.

\bibitem{CourcelleMR00line}
B.~Courcelle, J.~A. Makowsky, and U.~Rotics.
\newblock Linear time solvable optimization problems on graphs of bounded
  clique-width.
\newblock {\em Theory Comput. Syst.}, 33(2):125--150, 2000.

\bibitem{CourcelleM93mona}
B.~Courcelle and M.~Mosbah.
\newblock Monadic second-order evaluations on tree-decomposable graphs.
\newblock {\em Theoretical Comput. Sci.}, 109:49--82, 1993.

\bibitem{CrowstonGJKR10syst}
R.~Crowston, G.~Gutin, M.~Jones, E.~Kim, and I.~Ruzsa.
\newblock Systems of linear equations over {$\Bbb{F}$} and problems
  parameterized above average.
\newblock In H.~Kaplan, editor, {\em Algorithm Theory - SWAT 2010}, volume 6139
  of {\em Lecture Notes in Computer Science}, pages 164--175. Springer Berlin /
  Heidelberg, 2010.

\bibitem{CyganNPPRW11solv}
M.~Cygan, J.~Nederlof, M.~Pilipczuk, M.~Pilipczuk, J.~van Rooij, and J.~O.
  Wojtaszczyk.
\newblock Solving connectivity problems parameterized by treewidth in single
  exponential time.
\newblock {\em arxiv.org/abs/1103.0534}, 2011.

\bibitem{CyganPPW10kern}
M.~Cygan, M.~Pilipczuk, M.~Pilipczuk, and J.~O. Wojtaszczyk.
\newblock Kernelization hardness of connectivity problems in d-degenerate
  graphs.
\newblock In {\em Proceedings of the 36th international conference on
  Graph-theoretic concepts in computer science}, WG'10, pages 147--158, Berlin,
  Heidelberg, 2010. Springer-Verlag.

\bibitem{DawarGK07loca}
A.~Dawar, M.~Grohe, and S.~Kreutzer.
\newblock Locally excluding a minor.
\newblock In {\em Proc. of the 21st IEEE Symposium on Logic in Computer Science
  (LICS'07)}, pages 270--279, New York, 2007. IEEE.

\bibitem{DawarK09domi}
A.~Dawar and S.~Kreutzer.
\newblock Domination problems in nowhere-dense classes.
\newblock In {\em IARCS Annual Conference on Foundations of Software Technology
  and Theoretical Computer Science (FSTTCS 2009)}, pages 157--168, 2009.

\bibitem{DehneFFPR06nonb}
F.~K. H.~A. Dehne, M.~R. Fellows, H.~Fernau, E.~Prieto, and F.~A. Rosamond.
\newblock Nonblocker: Parameterized algorithmics for minimum dominating set.
\newblock In {\em SOFSEM}, pages 237--245, 2006.

\bibitem{DehneFLRS05anfp}
F.~K. H.~A. Dehne, M.~R. Fellows, M.~A. Langston, F.~A. Rosamond, and
  K.~Stevens.
\newblock An ${O}(2^{{{O}(k)}}n^{{3}})$ {FPT} algorithm for the undirected
  feedback vertex set problem.
\newblock In {\em Proc. of the 11th Annual International Conference on
  Computing and Combinatorics (COCOON 2005)}, volume 3595 of {\em Lecture Notes
  in Computer Science}, pages 859--869, Berlin, 2005. Springer.

\bibitem{DehneFLRS07anft}
F.~K. H.~A. Dehne, M.~R. Fellows, M.~A. Langston, F.~A. Rosamond, and
  K.~Stevens.
\newblock An $o(2^{O(k)}n^{3})$ fpt algorithm for the undirected feedback
  vertex set problem.
\newblock {\em Theory Comput. Syst.}, 41(3):479--492, 2007.

\bibitem{DellM10sati}
H.~Dell and D.~van Melkebeek.
\newblock Satisfiability allows no nontrivial sparsification unless the
  polynomial-time hierarchy collapses.
\newblock In {\em Proceedings of the 42nd ACM symposium on Theory of
  computing}, STOC '10, pages 251--260, New York, NY, USA, 2010. ACM.

\bibitem{DemaineH07theb}
E.~Demaine and M.~Hajiaghayi.
\newblock The bidimensionality theory and its algorithmic applications.
\newblock {\em The Computer Journal}, 51(3):292--302, 2007.

\bibitem{DemaineFHT05bidi}
E.~D. Demaine, F.~V. Fomin, M.~Hajiaghayi, and D.~M. Thilikos.
\newblock Bidimensional parameters and local treewidth.
\newblock {\em SIAM J. Discrete Math.}, 18(3):501--511, 2005.

\bibitem{DemaineFHT05talg}
E.~D. Demaine, F.~V. Fomin, M.~Hajiaghayi, and D.~M. Thilikos.
\newblock Fixed-parameter algorithms for $(k,r)$-center in planar graphs and
  map graphs.
\newblock {\em ACM Trans. Algorithms}, 1(1):33--47, 2005.

\bibitem{DemaineFHT05sube}
E.~D. Demaine, F.~V. Fomin, M.~Hajiaghayi, and D.~M. Thilikos.
\newblock Subexponential parameterized algorithms on bounded-genus graphs and
  {$H$}-minor-free graphs.
\newblock {\em J. Assoc. Comput. Mach.}, 52(6):866--893, 2005.

\bibitem{DemaineHaj05bidi}
E.~D. Demaine and M.~Hajiaghayi.
\newblock Bidimensionality: new connections between {FPT} algorithms and
  {PTAS}s.
\newblock In {\em Proc. of the 16th Annual ACM-SIAM Symposium on Discrete
  Algorithms (SODA 2005)}, pages 590--601, New York, 2005. ACM-SIAM.

\bibitem{DemaineH08bidi}
E.~D. Demaine and M.~Hajiaghayi.
\newblock Bidimensionality.
\newblock In {\em Encyclopedia of Algorithms}. Springer, 2008.

\bibitem{DemaineH08line}
E.~D. Demaine and M.~Hajiaghayi.
\newblock Linearity of grid minors in treewidth with applications through
  bidimensionality.
\newblock {\em Combinatorica}, 28(1):19--36, 2008.

\bibitem{DemaineHM09mini}
E.~D. Demaine, M.~Hajiaghayi, and D.~Marx.
\newblock Minimizing movement: Fixed-parameter tractability.
\newblock In {\em ESA}, pages 718--729, 2009.

\bibitem{DemaineHM10para}
E.~D. Demaine, M.~Hajiaghayi, and D.~Marx.
\newblock 09511 abstracts collection -- parameterized complexity and
  approximation algorithms.
\newblock In E.~D. Demaine, M.~Hajiaghayi, and D.~Marx, editors, {\em
  Parameterized complexity and approximation algorithms}, number 09511 in
  Dagstuhl Seminar Proceedings, Dagstuhl, Germany, 2010. Schloss Dagstuhl -
  Leibniz-Zentrum fuer Informatik, Germany.

\bibitem{DemaineHT05expo}
E.~D. Demaine, M.~Hajiaghayi, and D.~M. Thilikos.
\newblock Exponential speedup of fixed-parameter algorithms for classes of
  graphs excluding single-crossing graphs as minors.
\newblock {\em Algorithmica}, 41:245--267, 2005.

\bibitem{DemaineHT06theb}
E.~D. Demaine, M.~Hajiaghayi, and D.~M. Thilikos.
\newblock The bidimensional theory of bounded-genus graphs.
\newblock {\em SIAM J. Discrete Math.}, 20(2):357--371, 2006.

\bibitem{DiazST08effi}
J.~D\'{\i}az, M.~J. Serna, and D.~M. Thilikos.
\newblock Efficient algorithms for counting parameterized list h-colorings.
\newblock {\em J. Comput. Syst. Sci.}, 74(5):919--937, 2008.

\bibitem{DinneenCF01forb}
M.~J. Dinneen, K.~Cattell, and M.~R. Fellows.
\newblock Forbidden minors to graphs with small feedback sets.
\newblock {\em Discrete Math.}, 230(1-3):215--252, 2001.
\newblock Paul Catlin memorial collection (Kalamazoo, MI, 1996).

\bibitem{DinneenL07prop}
M.~J. Dinneen and R.~Lai.
\newblock Properties of vertex cover obstructions.
\newblock {\em Discrete Mathematics}, 307(21):2484--2500, 2007.

\bibitem{DinurS02thei}
I.~Dinur and S.~Safra.
\newblock The importance of being biased.
\newblock In {\em Proceedings of the thiry-fourth annual ACM symposium on
  Theory of computing}, STOC '02, pages 33--42, New York, NY, USA, 2002. ACM.

\bibitem{DomLS09inco}
M.~Dom, D.~Lokshtanov, and S.~Saurabh.
\newblock Incompressibility through colors and ids.
\newblock In S.~Albers, A.~Marchetti-Spaccamela, Y.~Matias, S.~Nikoletseas, and
  W.~Thomas, editors, {\em Automata, Languages and Programming}, volume 5555 of
  {\em Lecture Notes in Computer Science}, pages 378--389. Springer Berlin /
  Heidelberg, 2009.

\bibitem{Dorn06fast}
F.~Dorn.
\newblock Dynamic programming and fast matrix multiplication.
\newblock In {\em Proceedings of the 14th conference on Annual European
  Symposium - Volume 14}, pages 280--291, London, UK, 2006. Springer-Verlag.

\bibitem{Dorn07desi}
F.~Dorn.
\newblock {\em Designing Subexponential Algorithms: Problems, Techniques and
  Structures}.
\newblock PhD thesis, University of Bergen, July 2007.

\bibitem{DornFLRS10beyo}
F.~Dorn, F.~V. Fomin, D.~Lokshtanov, V.~Raman, and S.~Saurabh.
\newblock Beyond bidimensionality: Parameterized subexponential algorithms on
  directed graphs.
\newblock In {\em Proceedings of the 27th International Symposium on
  Theoretical Aspects of Computer Science (STACS 2010)}, volume~5 of {\em
  LIPIcs}, pages 251--262. Schloss Dagstuhl - Leibniz-Zentrum fuer Informatik,
  2010.

\bibitem{DornFT06fast}
F.~Dorn, F.~V. Fomin, and D.~M. Thilikos.
\newblock Fast subexponential algorithm for non-local problems on graphs of
  bounded genus.
\newblock In {\em Proc. of the 10th Scandinavian Workshop on Algorithm Theory
  (SWAT 2006)}, Lecture Notes in Computer Science, pages 172--183, Berlin,
  2006. Springer.

\bibitem{DornFT08cata}
F.~Dorn, F.~V. Fomin, and D.~M. Thilikos.
\newblock Catalan structures and dynamic programming in {$H$}-minor-free
  graphs.
\newblock In {\em Proc. of the ACM-SIAM Symposium on Discrete Algorithms (SODA
  2008)}, pages 631--640, 2008.

\bibitem{DornFT08sube}
F.~Dorn, F.~V. Fomin, and D.~M. Thilikos.
\newblock Subexponential parameterized algorithms.
\newblock {\em Comp. Sci. Rev.}, 2(1):29--39, 2008.

\bibitem{DornPBF10effi}
F.~Dorn, E.~Penninkx, H.~L. Bodlaender, and F.~V. Fomin.
\newblock Efficient exact algorithms on planar graphs: Exploiting sphere cut
  decompositions.
\newblock {\em Algorithmica}, 58(3):790--810, 2010.

\bibitem{Downey03para}
R.~Downey.
\newblock Parameterized complexity for the skeptic.
\newblock In {\em Proceedings. 18th IEEE Annual Conference on Proceedings 18th
  IEEE Annual Conference on Computational Complexity}, pages 147 -- 168, 2003.

\bibitem{DowneyF93fixe}
R.~Downey and M.~Fellows.
\newblock Fixed-parameter tractability and completeness. {III}. {S}ome
  structural aspects of the {$W$} hierarchy.
\newblock In {\em Complexity theory}, pages 191--225. Cambridge Univ. Press,
  Cambridge, 1993.

\bibitem{DowneyFM06para}
R.~Downey, M.~Fellows, and C.~McCartin.
\newblock Parameterized approximation problems.
\newblock In H.~Bodlaender and M.~Langston, editors, {\em Parameterized and
  Exact Computation}, volume 4169 of {\em Lecture Notes in Computer Science},
  pages 121--129. Springer Berlin / Heidelberg, 2006.

\bibitem{DowneyCFPR03cutt}
R.~G. Downey, V.~Estivill-Castro, M.~Fellows, E.~Prieto, and F.~A. Rosamund.
\newblock Cutting up is hard to do: The parameterised complexity of $k$-cut and
  related problems.
\newblock {\em Electronic Notes in Theoretical Computer Science}, 78:209 --
  222, 2003.
\newblock CATS'03, Computing: the Australasian Theory Symposium.

\bibitem{DowneyEF93para}
R.~G. Downey, P.~A. Evans, and M.~R. Fellows.
\newblock Parameterized learning complexity.
\newblock In {\em Proceedings of the sixth annual conference on Computational
  learning theory}, COLT '93, pages 51--57, New York, NY, USA, 1993. ACM.

\bibitem{DowneyF92}
R.~G. Downey and M.~R. Fellows.
\newblock Fixed-parameter tractability and completeness.
\newblock In {\em Proc. of the Twenty-first Manitoba Conference on Numerical
  Mathematics and Computing (Winnipeg, MB, 1991)}, volume~87, pages 161--178,
  1992.

\bibitem{DowneyF95fixe-I}
R.~G. Downey and M.~R. Fellows.
\newblock Fixed-parameter tractability and completeness. {I}. {B}asic results.
\newblock {\em SIAM J. Comput.}, 24(4):873--921, 1995.

\bibitem{DowneyF95fixe-II}
R.~G. Downey and M.~R. Fellows.
\newblock Fixed-parameter tractability and completeness {II}: {O}n completeness
  for ${W}[1]$.
\newblock {\em Theoretical Computer Science}, 141(1-2):109--131, 1995.

\bibitem{DowneyF95survey}
R.~G. Downey and M.~R. Fellows.
\newblock Parameterized computational feasibility.
\newblock In {\em Feasible mathematics, II (Ithaca, NY, 1992)}, volume~13 of
  {\em Progr. Comput. Sci. Appl. Logic}, pages 219--244. Birkh\"auser Boston,
  Boston, MA, 1995.

\bibitem{DowneyF99para}
R.~G. Downey and M.~R. Fellows.
\newblock {\em Parameterized complexity}.
\newblock Monographs in Computer Science. Springer-Verlag, New York, 1999.

\bibitem{DowneyFS99para}
R.~G. Downey and M.~R. Fellows.
\newblock Parameterized complexity after (almost) ten years: review and open
  questions.
\newblock In {\em Combinatorics, computation \& logic '99 (Auckland)},
  volume~21 of {\em Aust. Comput. Sci. Commun.}, pages 1--33. Springer,
  Singapore, 1999.

\bibitem{DowneyF12fund}
R.~G. Downey and M.~R. Fellows.
\newblock {\em Fundamentals of Parameterized complexity}.
\newblock Springer-Verlag, undegraduate texts in computer science, 2012.

\bibitem{Downey08thec}
R.~G. Downey, M.~R. Fellows, and M.~A. Langston.
\newblock The {C}omputer {J}ournal special issue on parameterized complexity:
  Foreword by the guest editors.
\newblock {\em The Computer Journal}, 51(1):1--6, 2008.

\bibitem{DowneyFM08para}
R.~G. Downey, M.~R. Fellows, C.~McCartin, and F.~Rosamond.
\newblock Parameterized approximation of dominating set problems.
\newblock {\em Information Processing Letters}, 109(1):68 -- 70, 2008.

\bibitem{DowneyFR98para}
R.~G. Downey, M.~R. Fellows, and K.~W. Regan.
\newblock Parameterized circuit complexity and the {W} hierarchy.
\newblock {\em Theor. Comput. Sci.}, 191:97--115, January 1998.

\bibitem{DowneyFS99thev}
R.~G. Downey, M.~R. Fellows, and U.~Stege.
\newblock Computational tractability: the view from {M}ars.
\newblock {\em Bull. Eur. Assoc. Theor. Comput. Sci. EATCS}, (69):73--97, 1999.

\bibitem{DowneyFS97survey}
R.~G. Downey, M.~R. Fellows, and U.~Stege.
\newblock Parameterized complexity: a framework for systematically confronting
  computational intractability.
\newblock In {\em Contemporary trends in discrete mathematics (\v Sti\v r\'\i n
  Castle, 1997)}, volume~49 of {\em DIMACS Ser. Discrete Math. Theoret. Comput.
  Sci.}, pages 49--99. Amer. Math. Soc., Providence, RI, 1999.

\bibitem{DowneyFVW99thep}
R.~G. Downey, M.~R. Fellows, A.~Vardy, and G.~Whittle.
\newblock The parametrized complexity of some fundamental problems in coding
  theory.
\newblock {\em SIAM J. Comput.}, 29:545--570, October 1999.

\bibitem{RodneyM04some}
R.~G. Downey and C.~McCartin.
\newblock Some new directions and questions in parameterized complexity.
\newblock In {\em Developments in Language Theory}, pages 12--26, 2004.

\bibitem{DvorakK09algo}
Z.~Dvorak and D.~Kr{\'a}l.
\newblock Algorithms for classes of graphs with bounded expansion.
\newblock In {\em WG}, pages 17--32, 2009.

\bibitem{DvorakKR10deci}
Z.~Dvorak, D.~Kral, and R.~Thomas.
\newblock Deciding first-order properties for sparse graphs.
\newblock In {\em Proceedings of the 2010 IEEE 51st Annual Symposium on
  Foundations of Computer Science}, FOCS '10, pages 133--142, Washington, DC,
  USA, 2010. IEEE Computer Society.

\bibitem{EickmeyerGG08appr}
K.~Eickmeyer, M.~Grohe, and M.~Gr{\"u}ber.
\newblock Approximation of natural w[p]-complete minimisation problems is hard.
\newblock In {\em IEEE Conference on Computational Complexity}, pages 8--18,
  2008.

\bibitem{ErlebachJS05poly}
T.~Erlebach, K.~Jansen, and E.~Seidel.
\newblock Polynomial-time approximation schemes for geometric intersection
  graphs.
\newblock {\em SIAM J. Comput.}, 34(6):1302--1323 (electronic), 2005.

\bibitem{Estivill-CastroFLR05fpti}
V.~Estivill-Castro, M.~R. Fellows, M.~A. Langston, and F.~A. Rosamond.
\newblock Fpt is p-time extremal structure i.
\newblock In {\em ACiD}, pages 1--41, 2005.

\bibitem{FellowsHRST04find}
M.~Fellows, P.~Heggernes, F.~Rosamond, C.~Sloper, and J.~A. Telle.
\newblock Finding {$k$} disjoint triangles in an arbitrary graph.
\newblock In {\em Graph-theoretic concepts in computer science}, volume 3353 of
  {\em Lecture Notes in Comput. Sci.}, pages 235--244. Springer, Berlin, 2004.

\bibitem{Fellows01para}
M.~R. Fellows.
\newblock Parameterized complexity: the main ideas and some research frontiers.
\newblock In {\em Algorithms and computation (Christchurch, 2001)}, volume 2223
  of {\em Lecture Notes in Comput. Sci.}, pages 291--307. Springer, Berlin,
  2001.

\bibitem{Fell0ows02survey}
M.~R. Fellows.
\newblock Parameterized complexity: The main ideas and connections to practical
  computing.
\newblock In {\em Experimental Algorithmics}, volume 2547 of {\em Lecture Notes
  in Comput. Sci.}, pages 51--77. Springer, Berlin, 2002.

\bibitem{Fellows03survey}
M.~R. Fellows.
\newblock New directions and new challenges in algorithm design and complexity,
  parameterized.
\newblock In {\em Algorithms and data structures}, volume 2748 of {\em Lecture
  Notes in Comput. Sci.}, pages 505--519. Springer, Berlin, 2003.

\bibitem{FellowsFLRSST11othe}
M.~R. Fellows, F.~V. Fomin, D.~Lokshtanov, F.~A. Rosamond, S.~Saurabh,
  S.~Szeider, and C.~Thomassen.
\newblock On the complexity of some colorful problems parameterized by
  treewidth.
\newblock {\em Inf. Comput.}, 209(2):143--153, 2011.

\bibitem{FellowsKNRRSTW08fast}
M.~R. Fellows, C.~Knauer, N.~Nishimura, P.~Ragde, F.~Rosamond, U.~Stege, D.~M.
  Thilikos, and S.~Whitesides.
\newblock Faster fixed-parameter tractable algorithms for matching and packing
  problems.
\newblock {\em Algorithmica}, 52(2):167--176, 2008.

\bibitem{FellowsL88nonc}
M.~R. Fellows and M.~A. Langston.
\newblock Nonconstructive tools for proving polynomial-time decidability.
\newblock {\em J. Assoc. Comput. Mach.}, 35(3):727--739, 1988.

\bibitem{FellowsL94}
M.~R. Fellows and M.~A. Langston.
\newblock On search, decision, and the efficiency of polynomial-time
  algorithms.
\newblock {\em J. Comput. System Sci.}, 49(3):769--779, 1994.

\bibitem{FellowsL89}
M.~R. Fellows and M.~A. Langston.
\newblock An analogue of the {M}yhill-{N}erode theorem and its use in computing
  finite-basis characterisations (extended abstract).
\newblock In {\em Proc. of the 30th Annual IEEE Symposium on Foundations of
  Computer Science, FOCS 1989}, pages 520--525, 1989.

\bibitem{Fernau02grap}
H.~Fernau.
\newblock Graph separator algorithms: A refined analysis.
\newblock In {\em The 28th International Workshop on Graph-Theoretic Concepts
  in Computer Science(WG 2002)}, page to appear. Springer, Lecture Notes in
  Computer Science, Berlin, 2002.

\bibitem{Fernau05para}
H.~Fernau.
\newblock {\em Parameterized Algorithmics: A Graph-Theoretic Approach}.
\newblock Habilitationsschrift, Universitat Tubingen, Tubingen, Germany., 2005.

\bibitem{FernauFLRSV09kern}
H.~Fernau, F.~V. Fomin, D.~Lokshtanov, D.~Raible, S.~Saurabh, and Y.~Villanger.
\newblock Kernel(s) for problems with no kernel: On out-trees with many leaves.
\newblock In S.~Albers and J.-Y. Marion, editors, {\em 26th International
  Symposium on Theoretical Aspects of Computer Science (STACS 2009)}, volume~3
  of {\em Leibniz International Proceedings in Informatics (LIPIcs)}, pages
  421--432, Dagstuhl, Germany, 2009. Schloss Dagstuhl--Leibniz-Zentrum fuer
  Informatik.

\bibitem{FernauJ04age}
H.~Fernau and D.~Juedes.
\newblock A geometric approach to parameterized algorithms for domination
  problems on planar graphs.
\newblock In {\em Proc. of the 29th International Symposium on Mathematical
  Foundations of Computer (MFCS 2004)}, volume 3153 of {\em Lecture Notes in
  Comput. Sci.}, pages 488--499. Springer, Berlin, 2004.

\bibitem{FlumG02thep}
J.~Flum and M.~Grohe.
\newblock The parameterized complexity of counting problems.
\newblock In {\em Proceedings of the 43rd Symposium on Foundations of Computer
  Science}, FOCS '02, pages 538--, Washington, DC, USA, 2002. IEEE Computer
  Society.

\bibitem{FlumG04para}
J.~Flum and M.~Grohe.
\newblock Parameterized complexity and subexponential time.
\newblock {\em Bull. Eur. Assoc. Theor. Comput. Sci. EATCS}, (84):71--100,
  2004.

\bibitem{FlumG06para}
J.~Flum and M.~Grohe.
\newblock {\em Parameterized Complexity theory}.
\newblock Texts in Theoretical Computer Science. An EATCS Series.
  Springer-Verlag, Berlin, 2006.

\bibitem{FlumGW06boun}
J.~Flum, M.~Grohe, and M.~Weyer.
\newblock Bounded fixed-parameter tractability and $\log^{2}n$ nondeterministic
  bits.
\newblock {\em J. Comput. Syst. Sci.}, 72:34--71, February 2006.

\bibitem{FominGLS10algo}
F.~Fomin, P.~Golovach, D.~Lokshtanov, and S.~Saurabh.
\newblock Algorithmic lower bounds for problems parameterized with
  clique-width, proceedings of acm-siam symposium on discrete algorithms.
\newblock In {\em SODA 2010}, pages 493--502, 2010.

\bibitem{FominGT09cont}
F.~V. Fomin, P.~Golovach, and D.~M. Thilikos.
\newblock Contraction bidimensionality: the accurate picture, 2009.

\bibitem{FedorGLS09cliq}
F.~V. Fomin, P.~A. Golovach, D.~Lokshtanov, and S.~Saurabh.
\newblock Clique-width: on the price of generality.
\newblock In {\em SODA '09: Proc. of the twentieth Annual ACM-SIAM Symposium on
  Discrete Algorithms}, pages 825--834, Philadelphia, PA, USA, 2009. Society
  for Industrial and Applied Mathematics.

\bibitem{FominGK06meas}
F.~V. Fomin, F.~Grandoni, and D.~Kratsch.
\newblock Measure and conquer: a simple ${O}(2^{0.288 n})$) independent set
  algorithm.
\newblock In {\em SODA}, pages 18--25, 2006.

\bibitem{FominLMPS11hitt}
F.~V. Fomin, D.~Lokshtanov, N.~Misra, G.~Philip, and S.~Saurabh.
\newblock {Hitting forbidden minors: Approximation and Kernelization}.
\newblock In T.~Schwentick and C.~D{\"u}rr, editors, {\em 28th International
  Symposium on Theoretical Aspects of Computer Science (STACS 2011)}, volume~9
  of {\em Leibniz International Proceedings in Informatics (LIPIcs)}, pages
  189--200, Dagstuhl, Germany, 2011. Schloss Dagstuhl--Leibniz-Zentrum fuer
  Informatik.

\bibitem{FominLRS10fast}
F.~V. Fomin, D.~Lokshtanov, V.~Raman, and S.~Saurabh.
\newblock Fast local search algorithm for weighted feedback arc set in
  tournaments.
\newblock In {\em AAAI}, 2010.

\bibitem{FominLRS11bidi}
F.~V. Fomin, D.~Lokshtanov, V.~Raman, and S.~Saurabh.
\newblock Bidimensionality and {EPTAS}.
\newblock In {\em 22st ACM--SIAM Symposium on Discrete Algorithms (SODA 2011)},
  2011.

\bibitem{FominLST10bidi}
F.~V. Fomin, D.~Lokshtanov, S.~Saurabh, and D.~M. Thilikos.
\newblock Bidimensionality and kernels.
\newblock In M.~Charikar, editor, {\em Twenty-First Annual ACM-SIAM Symposium
  on Discrete Algorithms (SODA 2010), Austin, Texas}, pages 503--510. SIAM,
  2010.

\bibitem{FominT06domi}
F.~V. Fomin and D.~M. Thilikos.
\newblock Dominating sets in planar graphs: branch-width and exponential
  speed-up.
\newblock {\em SIAM J. Comput.}, 36(2):281--309 (electronic), 2006.

\bibitem{FominV11sube}
F.~V. Fomin and Y.~Villanger.
\newblock Subexponential parameterized algorithm for minimum fill-in.
\newblock {\em CoRR}, abs/1104.2230, 2011.

\bibitem{FortnowS11infe}
L.~Fortnow and R.~Santhanam.
\newblock Infeasibility of instance compression and succinct {PCP}s for {NP}.
\newblock {\em Journal of Computer and System Sciences}, 77(1):91 -- 106, 2011.
\newblock Celebrating Karp's Kyoto Prize.

\bibitem{FredmanKS82stor}
M.~L. Fredman, J.~Koml\'os, and E.~Szemer\'edi.
\newblock Storing a sparse table with ${O}(1)$ worst-case access time.
\newblock In {\em Proceedings of the 23rd Annual IEEE Symposium on Foundations
  of Computer Science}, pages 165--169, 1982.

\bibitem{FrickG99deci}
M.~Frick and M.~Grohe.
\newblock Deciding first-order properties of locally tree-decomposable graphs.
\newblock In J.~Wiedermann, P.~van Emde~Boas, and M.~Nielsen, editors, {\em
  Automata, Languages and Programming}, volume 1644 of {\em Lecture Notes in
  Computer Science}, pages 72--72. Springer Berlin / Heidelberg, 1999.

\bibitem{FrickG04thec}
M.~Frick and M.~Grohe.
\newblock The complexity of first-order and monadic second-order logic
  revisited.
\newblock {\em Ann. Pure Appl. Logic}, 130(1-3):3--31, 2004.

\bibitem{FriedmanRS87them}
H.~Friedman, N.~Robertson, and P.~Seymour.
\newblock The metamathematics of the graph minor theorem.
\newblock In {\em Logic and combinatorics (Arcata, Calif., 1985)}, volume~65 of
  {\em Contemp. Math.}, pages 229--261. Amer. Math. Soc., Providence, RI, 1987.

\bibitem{GareyJ79comp}
M.~R. Garey and D.~S. Johnson.
\newblock {\em Computers and intractability. A guide to the theory of
  NP-completeness}.
\newblock W. H. Freeman and Co., San Francisco, Calif., 1979.

\bibitem{GeelenGW06}
J.~Geelen, B.~Gerards, and G.~Whittle.
\newblock Towards a structure theory for matrices and matroids.
\newblock In {\em International {C}ongress of {M}athematicians. {V}ol. {III}},
  pages 827--842. Eur. Math. Soc., Z{\"u}rich, 2006.

\bibitem{GolovachT09para}
P.~Golovach and D.~Thilikos.
\newblock Paths of bounded length and their cuts: Parameterized complexity and
  algorithms.
\newblock In {\em Parameterized and Exact Computation}, volume 5917 of {\em
  Lecture Notes in Comput. Sci.}, pages 210--221. Springer, Berlin, 2009.

\bibitem{GolovachKPT09indu}
P.~A. Golovach, M.~Kami{\'n}ski, D.~Paulusma, and D.~M. Thilikos.
\newblock Induced packing of odd cycles in a planar graph.
\newblock In {\em Proc. of the 20th International Symposium on Algorithms and
  Computation (ISAAC 2009)}, volume 5878 of {\em Lecture Notes in Comput.
  Sci.}, pages 514--523. Springer, Berlin, 2009.

\bibitem{GolovachT11path}
P.~A. Golovach and D.~M. Thilikos.
\newblock Paths of bounded length and their cuts: Parameterized complexity and
  algorithms.
\newblock {\em Discrete Optimization}, 8(1):72 -- 86, 2011.
\newblock Parameterized Complexity of Discrete Optimization.

\bibitem{GrammGHN04auto}
J.~Gramm, J.~Guo, F.~H{\"u}ffner, and R.~Niedermeier.
\newblock Automated generation of search tree algorithms for hard graph
  modification problems.
\newblock {\em Algorithmica}, 39(4):321--347, 2004.

\bibitem{Grohe08algo}
M.~Grohe.
\newblock Algorithmic meta theorems.
\newblock In {\em WG}, page~30, 2008.

\bibitem{Grohe08logi}
M.~Grohe.
\newblock Logic, graphs, and algorithms.
\newblock In {\em Logic and Automata}, pages 357--422, 2008.

\bibitem{GroheKMW10find}
M.~Grohe, K.~Kawarabayashi, D.~Marx, and P.~Wollan.
\newblock Finding topological subgraphs is fixed-parameter tractable.
\newblock In {\em STOC 2011}, pages 479--488, 2011.

\bibitem{GuT08opti}
Q.-P. Gu and H.~Tamaki.
\newblock Optimal branch decomposition of planar graphs in ${O}(n^3)$ time.
\newblock {\em ACM Trans. Algorithms}, 4(3):1 -- 13, 2008.
\newblock Article No. 30.

\bibitem{GuoT10impr}
Q.-P. Gu and H.~Tamaki.
\newblock Improved bounds on the planar branchwidth with respect to the largest
  grid minor size.
\newblock In O.~Cheong, K.-Y. Chwa, and K.~Park, editors, {\em Algorithms and
  Computation}, volume 6507 of {\em Lecture Notes in Computer Science}, pages
  85--96. Springer Berlin / Heidelberg, 2010.

\bibitem{GuoGHNW06comp}
J.~Guo, J.~Gramm, F.~H{\"u}ffner, R.~Niedermeier, and S.~Wernicke.
\newblock Compression-based fixed-parameter algorithms for feedback vertex set
  and edge bipartization.
\newblock {\em J. Comput. System Sci.}, 72(8):1386--1396, 2006.

\bibitem{GuoMN09iter}
J.~Guo, H.~Moser, and R.~Niedermeier.
\newblock Iterative compression for exactly solving {NP}-hard minimization
  problems.
\newblock In J.~Lerner, D.~Wagner, and K.~Zweig, editors, {\em Algorithmics of
  Large and Complex Networks}, volume 5515 of {\em Lecture Notes in Computer
  Science}, pages 65--80. Springer Berlin / Heidelberg, 2009.

\bibitem{GuoN07invi}
J.~Guo and R.~Niedermeier.
\newblock Invitation to data reduction and problem kernelization.
\newblock {\em SIGACT News}, 38:31--45, March 2007.

\bibitem{GuoNiedermeier2007line}
J.~Guo and R.~Niedermeier.
\newblock Linear problem kernels for {NP}-hard problems on planar graphs.
\newblock In {\em Automata, languages and programming}, volume 4596 of {\em
  Lecture Notes in Comput. Sci.}, pages 375--386. Springer, Berlin, 2007.

\bibitem{GutinKSY11apro}
G.~Gutin, E.~Kim, S.~Szeider, and A.~Yeo.
\newblock A probabilistic approach to problems parameterized above or below
  tight bounds.
\newblock {\em {Journal of Computer and System Sciences}}, 77:422--429, 2011.

\bibitem{GutinIMY11ever}
G.~Gutin, L.~van Iersel, M.~Mnich, and A.~Yeo.
\newblock Every ternary permutation constraint satisfaction problem
  parameterized above average has a kernel with a quadratic number of
  variables.
\newblock {\em Journal of Computer and System Sciences}, In Press, Corrected
  Proof:--, 2011.

\bibitem{HarnikN06onth}
D.~Harnik and M.~Naor.
\newblock On the compressibility of np instances and cryptographic
  applications.
\newblock {\em Foundations of Computer Science, Annual IEEE Symposium on},
  0:719--728, 2006.

\bibitem{HermelinHKW11para}
D.~Hermelin, C.-C. Huang, S.~Kratsch, and M.~Wahlstr{\"o}m.
\newblock Parameterized two-player nash equilibrium.
\newblock In {\em Proceedings of the 37th International Workshop on
  Graph-Theoretic Concepts in Computer Science (WG 2011)}, Lecture Notes in
  Computer Science, page to appear. Springer, 2011.

\bibitem{Hlineny06}
P.~Hlin{\v{e}}n{\'y}.
\newblock Branch-width, parse trees, and monadic second-order logic for
  matroids.
\newblock {\em J. Combin. Theory Ser. B}, 96(3):325--351, 2006.

\bibitem{Huffner05inst}
F.~H{\"u}ffner.
\newblock Algorithm engineering for optimal graph bipartization.
\newblock In S.~E. Nikoletseas, editor, {\em Experimental and Efficient
  Algorithms}, volume 3503 of {\em Lecture Notes in Computer Science}, pages
  11--18. Springer Berlin / Heidelberg, 2005.

\bibitem{HKMN10fixe}
F.~H{\"u}ffner, C.~Komusiewicz, H.~Moser, and R.~Niedermeier.
\newblock Fixed-parameter algorithms for cluster vertex deletion.
\newblock {\em Theory of Computing Systems}, 47:196--217, 2010.

\bibitem{HuffnerWZ08algo}
F.~H{\"u}ffner, S.~Wernicke, and T.~Zichner.
\newblock Algorithm engineering for color-coding with applications to signaling
  pathway detection.
\newblock {\em Algorithmica}, 52:114--132, 2008.

\bibitem{ImpagliazzoP99thec}
R.~Impagliazzo and R.~Paturi.
\newblock The complexity of k-{SAT}.
\newblock In {\em COCO '99: Proc. of the Fourteenth Annual IEEE Conference on
  Computational Complexity}, page 237, Washington, DC, USA, 1999. IEEE Computer
  Society.

\bibitem{ImpagliazzoPZ01whic}
R.~Impagliazzo, R.~Paturi, and F.~Zane.
\newblock Which problems have strongly exponential complexity.
\newblock {\em J. Comput. System Sci.}, 63(4):512--530, 2001.
\newblock Special issue on FOCS (Palo Alto, CA).

\bibitem{JansenB11vert}
B.~M.~P. Jansen and H.~L. Bodlaender.
\newblock Vertex cover kernelization revisited: Upper and lower bounds for a
  refined parameter.
\newblock In {\em STACS}, pages 177--188, 2011.

\bibitem{Jansen09para}
K.~Jansen.
\newblock Parameterized approximation scheme for the multiple knapsack problem.
\newblock In {\em Proceedings of the twentieth Annual ACM-SIAM Symposium on
  Discrete Algorithms}, SODA '09, pages 665--674, Philadelphia, PA, USA, 2009.
  Society for Industrial and Applied Mathematics.

\bibitem{JiaZC04anef}
W.~Jia, C.~Zhang, and J.~Chen.
\newblock An efficient parameterized algorithm for {$m$}-set packing.
\newblock {\em J. Algorithms}, 50(1):106--117, 2004.

\bibitem{KanjPer02impr}
I.~Kanj and L.~Perkovi\'c.
\newblock Improved parameterized algorithms for planar dominating set.
\newblock In {\em Mathematical Foundations of Computer Science---MFCS 2002
  (Warsaw, Poland)}, volume 2420, pages 399--410. Springer, Lecture Notes in
  Computer Science, Berlin, 2002.

\bibitem{Karp74}
R.~M. Karp.
\newblock On the computational complexity of combinatorial problems.
\newblock {\em Networks}, 5(1):45--68, 1975.
\newblock (Proc. of the Symposium on Large-Scale Networks, Evanston, IL, USA,
  18-19 April 1974).

\bibitem{Kawarabayashi11thed}
K.~Kawarabayashi.
\newblock The disjoint paths problem: Algorithm and structure.
\newblock In N.~Katoh and A.~Kumar, editors, {\em WALCOM: Algorithms and
  Computation}, volume 6552 of {\em Lecture Notes in Computer Science}, pages
  2--7. Springer Berlin / Heidelberg, 2011.

\bibitem{KawarabayashiK08integ}
K.~Kawarabayashi and Y.~Kobayashi.
\newblock The induced disjoint path problem.
\newblock In {\em Integer programming and combinatorial optimization}, volume
  5035 of {\em Lecture Notes in Comput. Sci.}, pages 47--61. Springer, Berlin,
  2008.

\bibitem{Kawarabayashi2010oddc}
K.~Kawarabayashi and B.~Reed.
\newblock Odd cycle packing.
\newblock In {\em Proceedings of the 42nd ACM Symposium on Theory of Computing
  (STOC 2010)}, pages 695--704, New York, NY, USA, 2010. ACM.

\bibitem{KawarabayashiW10asho}
K.~Kawarabayashi and P.~Wollan.
\newblock A shorter proof of the graph minor algorithm: the unique linkage
  theorem.
\newblock In {\em STOC}, pages 687--694, 2010.

\bibitem{KleitmanW91span}
D.~J. Kleitman and D.~B. West.
\newblock Spanning trees with many leaves.
\newblock {\em SIAM J. Discret. Math.}, 4:99--106, January 1991.

\bibitem{KloksLL02newa}
T.~Kloks, C.~M. Lee, and J.~Liu.
\newblock New algorithms for {$k$}-face cover, {$k$}-feedback vertex set, and
  {$k$}-disjoint cycles on plane and planar graphs.
\newblock In {\em Proc. of the 28th International Workshop on Graph Theoretic
  Concepts in Computer Science (WG 2002)}, volume 2573 of {\em Lecture Notes in
  Comput. Sci.}, pages 282--295. Springer, Berlin, 2002.

\bibitem{KneisL09apra}
J.~Kneis and A.~Langer.
\newblock A practical approach to {C}ourcelle's theorem.
\newblock {\em Electron. Notes Theor. Comput. Sci.}, 251:65--81, September
  2009.

\bibitem{KobayashiK09algo}
Y.~Kobayashi and K.~Kawarabayashi.
\newblock Algorithms for finding an induced cycle in planar graphs and bounded
  genus graphs.
\newblock In {\em Proceedings of the twentieth Annual ACM-SIAM Symposium on
  Discrete Algorithms (SODA 2009)}, pages 1146--1155. ACM-SIAM, 2009.

\bibitem{KoutsonasT10plan}
A.~Koutsonas and D.~Thilikos.
\newblock Planar feedback vertex set and face cover: Combinatorial bounds and
  subexponential algorithms.
\newblock {\em Algorithmica}, pages 1--17, 2010.

\bibitem{Kratsch11cono}
S.~Kratsch.
\newblock Co-nondeterminism in compositions: A kernelization lower bound for a
  ramsey-type problem.
\newblock {\em CoRR}, abs/1107.3704, 2011.

\bibitem{KratschMW10para}
S.~Kratsch, D.~Marx, and M.~Wahlstr\"{o}m.
\newblock Parameterized complexity and kernelizability of max ones and exact
  ones problems.
\newblock In {\em Proceedings of the 35th international conference on
  Mathematical foundations of computer science}, MFCS'10, pages 489--500,
  Berlin, Heidelberg, 2010. Springer-Verlag.

\bibitem{Kreutzer08algo}
S.~Kreutzer.
\newblock Algorithmic meta-theorems.
\newblock In {\em IWPEC}, pages 10--12, 2008.

\bibitem{LangerRS11line}
A.~Langer, P.~Rossmanith, and S.~Sikdar.
\newblock Linear-time algorithms for graphs of bounded rankwidth: A fresh look
  using game theory.
\newblock {\em CoRR}, abs/1102.0908, 2011.

\bibitem{KneisLR10cour}
J.~K.~A. Langer and P.~Rossmanith.
\newblock Courcelle's theorem -- a game-theoretic approach.
\newblock Submitted for publication, 2010.

\bibitem{LangstonPSSV08inno}
M.~A. Langston, A.~D. Perkins, A.~M. Saxton, J.~A. Scharff, and B.~H. Voy.
\newblock Innovative computational methods for transcriptomic data analysis: A
  case study in the use of {FPT} for practical algorithm design and
  implementation.
\newblock {\em The Computer Journal}, 51(1):26--38, 2008.

\bibitem{Lemke98inst}
P.~Lemke.
\newblock The maximum leaf spanning tree problem for cubic graphs is {NP}-
  complete.
\newblock Institute for Mathematics and its Applications, preprint, 1998.

\bibitem{LiuLCS06gree}
Y.~Liu, S.~Lu, J.~Chen, and S.-H. Sze.
\newblock Greedy localization and color-coding: Improved matching and packing
  algorithms.
\newblock In H.~Bodlaender and M.~Langston, editors, {\em Parameterized and
  Exact Computation}, volume 4169 of {\em Lecture Notes in Computer Science},
  pages 84--95. Springer Berlin / Heidelberg, 2006.

\bibitem{LokshtanovMS10know}
D.~Lokshtanov, D.~Marx, and S.~Saurabh.
\newblock Known algorithms on graphs of bounded treewidth are probably optimal.
\newblock {\em CoRR}, abs/1007.5450, 2010.

\bibitem{LokshtanovMS11know}
D.~Lokshtanov, D.~Marx, and S.~Saurabh.
\newblock Known algorithms on graphs of bounded treewidth are probably optimal.
\newblock In {\em 22st ACM--SIAM Symposium on Discrete Algorithms (SODA 2011)},
  2011.

\bibitem{LokshtanovMS11slig}
D.~Lokshtanov, D.~Marx, and S.~Saurabh.
\newblock Slightly superexponential parameterized problems.
\newblock In {\em 22st ACM--SIAM Symposium on Discrete Algorithms (SODA 2011)},
  pages 760--776, 2011.

\bibitem{Lucks82isom}
E.~M. Luks.
\newblock Isomorphism of graphs of bounded valence can be tested in polynomial
  time.
\newblock {\em Journal of Computer and System Sciences}, 25:42--65., 1982.

\bibitem{Marx08clos}
D.~Marx.
\newblock Closest substring problems with small distances.
\newblock {\em SIAM J. Comput.}, 38:1382--1410, August 2008.

\bibitem{MarxR09cons}
D.~Marx and I.~Razgon.
\newblock Constant ratio fixed-parameter approximation of the edge multicut
  problem.
\newblock {\em Inf. Process. Lett.}, 109:1161--1166, September 2009.

\bibitem{MarxR11fixed}
D.~Marx and I.~Razgon.
\newblock Fixed-parameter tractability of multicut parameterized by the size of
  the cutset.
\newblock In {\em STOC}, pages 469--478, 2011.

\bibitem{MathiesonPS04pack}
L.~Mathieson, E.~Prieto, and P.~Shaw.
\newblock Packing edge disjoint triangles: A parameterized view.
\newblock In R.~Downey, M.~Fellows, and F.~Dehne, editors, {\em Parameterized
  and Exact Computation}, volume 3162 of {\em Lecture Notes in Computer
  Science}, pages 127--137. Springer Berlin / Heidelberg, 2004.

\bibitem{McCartin06para}
C.~McCartin.
\newblock Parameterized counting problems.
\newblock {\em Ann. Pure Appl. Logic}, 138(1-3):147--182, 2006.

\bibitem{MisraRSS09thei}
N.~Misra, V.~Raman, S.~Saurabh, and S.~Sikdar.
\newblock The budgeted unique coverage problem and color-coding.
\newblock In A.~Frid, A.~Morozov, A.~Rybalchenko, and K.~Wagner, editors, {\em
  Computer Science - Theory and Applications}, volume 5675 of {\em Lecture
  Notes in Computer Science}, pages 310--321. Springer Berlin / Heidelberg,
  2009.

\bibitem{Montoya08thep}
J.~A. Montoya.
\newblock The parameterized complexity of probability amplification.
\newblock {\em Inf. Process. Lett.}, 109:46--53, December 2008.

\bibitem{Muller06rand}
M.~M{\"u}ller.
\newblock Randomized approximations of parameterized counting problems.
\newblock In H.~Bodlaender and M.~Langston, editors, {\em Parameterized and
  Exact Computation}, volume 4169 of {\em Lecture Notes in Computer Science},
  pages 50--59. Springer Berlin / Heidelberg, 2006.

\bibitem{Muller08para}
M.~M{\"u}ller.
\newblock Parameterized derandomization.
\newblock In {\em Proceedings of the 3rd international conference on
  Parameterized and exact computation}, IWPEC'08, pages 148--159, Berlin,
  Heidelberg, 2008. Springer-Verlag.

\bibitem{Muller08vali}
M.~M{\"u}ller.
\newblock Valiant-vazirani lemmata for various logics.
\newblock {\em Electronic Colloquium on Computational Complexity (ECCC)},
  15(063), 2008.

\bibitem{NaorSS95spli}
M.~Naor, L.~J. Schulman, and A.~Srinivasan.
\newblock Splitters and near-optimal derandomization.
\newblock In {\em Proceedings of the 36th Annual Symposium on Foundations of
  Computer Science}, FOCS '95, pages 182--, Washington, DC, USA, 1995. IEEE
  Computer Society.

\bibitem{NemT75vert}
G.~L. Nemhauser and L.~E. Trotter.
\newblock Vertex packings: Structural properties and algorithms.
\newblock {\em Mathematical Programming}, 8:232--248, 1975.

\bibitem{NesetrilO11onno}
Ne{\v{s}}et{\v{r}}il and P.~O. de~Mendez.
\newblock On nowhere dense graphs.
\newblock {\em European Journal of Combinatorics}, 32(4):600 -- 617, 2011.

\bibitem{NesetrilMW09char}
J.~Ne{\v s}et{\v r}il, P.~O. de~Mendez, and D.~R. Wood.
\newblock Characterisations and examples of graph classes with bounded
  expansion.
\newblock Technical Report arXiv:0902.3265, Cornell University, Feb 2009.

\bibitem{NesetrilO08grad-I}
J.~Ne{\v{s}}et{\v{r}}il and P.~Ossona~de Mendez.
\newblock Grad and classes with bounded expansion. {I}. {D}ecompositions.
\newblock {\em European J. Combin.}, 29(3):760--776, 2008.

\bibitem{NesetrilO08grad-II}
J.~Ne{\v{s}}et{\v{r}}il and P.~Ossona~de Mendez.
\newblock Grad and classes with bounded expansion. {II}. {A}lgorithmic aspects.
\newblock {\em European J. Combin.}, 29(3):777--791, 2008.

\bibitem{NesetrilO08grad-III}
J.~Ne{\v{s}}et{\v{r}}il and P.~Ossona~de Mendez.
\newblock Grad and classes with bounded expansion. {III}. {R}estricted graph
  homomorphism dualities.
\newblock {\em European J. Combin.}, 29(4):1012--1024, 2008.

\bibitem{Niedermeier06invi}
R.~Niedermeier.
\newblock {\em Invitation to fixed-parameter algorithms}, volume~31 of {\em
  Oxford Lecture Series in Mathematics and its Applications}.
\newblock Oxford University Press, Oxford, 2006.

\bibitem{NiedermeierR99uppe}
R.~Niedermeier and P.~Rossmanith.
\newblock Upper bounds for vertex cover further improved.
\newblock In {\em S{TACS} 99 ({T}rier)}, volume 1563 of {\em Lecture Notes in
  Comput. Sci.}, pages 561--570. Springer, Berlin, 1999.

\bibitem{NiedermeierR03onef}
R.~Niedermeier and P.~Rossmanith.
\newblock On efficient fixed-parameter algorithms for weighted vertex cover.
\newblock {\em J. Algorithms}, 47(2):63--77, 2003.

\bibitem{NishimuraRT05para}
N.~Nishimura, P.~Ragde, and D.~M. Thilikos.
\newblock Parameterized counting algorithms for general graph covering
  problems.
\newblock In {\em WADS}, pages 99--109, 2005.

\bibitem{OumS06appr}
S.~Oum and P.~Seymour.
\newblock Approximating clique-width and branch-width.
\newblock {\em J. Combin. Theory Ser. B}, 96(4):514--528, 2006.

\bibitem{PapadimitriouY96onli}
C.~Papadimitriou and M.~Yannakakis.
\newblock On limited nondeterminism and the complexity of the v--c dimension.
\newblock {\em J. Comput. System Sci.}, 53:161--170, 1996.

\bibitem{PhilipRS09}
G.~Philip, V.~Raman, and S.~Sikdar.
\newblock Solving dominating set in larger classes of graphs: {FPT} algorithms
  and polynomial kernels.
\newblock In {\em Proceedings of the 17th Annual European Symposium on
  Algorithms (ESA 2009)}, volume 5757 of {\em Lecture Notes in Computer
  Science}, pages 694--705. Springer, 2009.

\bibitem{PrietoS04look}
E.~Prieto and C.~Sloper.
\newblock Looking at the stars.
\newblock In R.~Downey, M.~Fellows, and F.~Dehne, editors, {\em Parameterized
  and Exact Computation}, volume 3162 of {\em Lecture Notes in Computer
  Science}, pages 138--148. Springer Berlin / Heidelberg, 2004.

\bibitem{RazgonO09almo}
I.~Razgon and B.~O'Sullivan.
\newblock Almost 2-sat is fixed-parameter tractable.
\newblock {\em J. Comput. Syst. Sci.}, 75:435--450, December 2009.

\bibitem{ReedSV04find}
B.~Reed, K.~Smith, and A.~Vetta.
\newblock Finding odd cycle transversals.
\newblock {\em Oper. Res. Lett.}, 32(4):299--301, 2004.

\bibitem{RobertsonS-XXII}
N.~Robertson and P.~Seymour.
\newblock Graph minors. {XXII}. {I}rrelevant vertices in linkage problems,
  preprint, 1992.

\bibitem{RobertsonS-XXI}
N.~Robertson and P.~Seymour.
\newblock Graph minors. {XXI}. {G}raphs with unique linkages.
\newblock {\em J. Combin. Theory Ser. B}, 99(3):583--616, 2009.

\bibitem{RobertsonST94quic}
N.~Robertson, P.~Seymour, and R.~Thomas.
\newblock Quickly excluding a planar graph.
\newblock {\em J. Combin. Theory Ser. B}, 62(2):323--348, 1994.

\bibitem{RobertsonS85}
N.~Robertson and P.~D. Seymour.
\newblock Graph minors---a survey.
\newblock In {\em Surveys in combinatorics 1985 (Glasgow, 1985)}, volume 103 of
  {\em London Math. Soc. Lecture Note Ser.}, pages 153--171. Cambridge Univ.
  Press, Cambridge, 1985.

\bibitem{RobertsonS91-X}
N.~Robertson and P.~D. Seymour.
\newblock Graph minors. {X}. {O}bstructions to tree-decomposition.
\newblock {\em J. Combin. Theory Ser. B}, 52(2):153--190, 1991.

\bibitem{RobertsonS95-XIII}
N.~Robertson and P.~D. Seymour.
\newblock Graph minors. {XIII}. {T}he disjoint paths problem.
\newblock {\em Journal of Combinatorial Theory. Series B}, 63(1):65--110, 1995.

\bibitem{RobertsonS-XIII}
N.~Robertson and P.~D. Seymour.
\newblock Graph minors. {XIII}. {T}he disjoint paths problem.
\newblock {\em Journal of Combinatorial Theory. Series B}, 63(1):65--110, 1995.

\bibitem{RobertsonS04-XX}
N.~Robertson and P.~D. Seymour.
\newblock Graph minors. {XX}. {W}agner's conjecture.
\newblock {\em J. Combin. Theory Ser. B}, 92(2):325--357, 2004.

\bibitem{Robson86algo}
J.~M. Robson.
\newblock Algorithms for maximum independent sets.
\newblock {\em J. Algorithms}, 7(3):425--440, 1986.

\bibitem{RueST10dyna}
J.~Ru{\'e}, I.~Sau, and D.~M. Thilikos.
\newblock Dynamic programming for graphs on surfaces.
\newblock In {\em ICALP (1)}, pages 372--383, 2010.

\bibitem{SauT10sube}
I.~Sau and D.~M. Thilikos.
\newblock Subexponential parameterized algorithms for degree-constrained
  subgraph problems on planar graphs.
\newblock {\em Journal of Discrete Algorithms}, 8(3):330--338, 9 2010.

\bibitem{SchmidtS90thes}
J.~P. Schmidt and A.~Siegel.
\newblock The spatial complexity of oblivious $k$-probe hash functions.
\newblock {\em SIAM J. Comput.}, 19:775--786, September 1990.

\bibitem{Schoning87grap}
U.~Sch{\"o}ning.
\newblock Graph isomorphism is in the low hierarchy.
\newblock In F.~Brandenburg, G.~Vidal-Naquet, and M.~Wirsing, editors, {\em
  STACS 87}, volume 247 of {\em Lecture Notes in Computer Science}, pages
  114--124. Springer Berlin / Heidelberg, 1987.

\bibitem{ScottIKS05effi}
J.~Scott, T.~Ideker, R.~M. Karp, and R.~Sharan.
\newblock Efficient algorithms for detecting signaling pathways in protein
  interaction networks.
\newblock {\em Journal of Computational Biology}, 13(2):133--144, 2006.

\bibitem{SernaT05para}
M.~Serna and D.~M. Thilikos.
\newblock Parameterized complexity for graph layout problems.
\newblock {\em Bull. Eur. Assoc. Theor. Comput. Sci. EATCS}, (86):41--65, 2005.

\bibitem{SeymourT94call}
P.~D. Seymour and R.~Thomas.
\newblock Call routing and the ratcatcher.
\newblock {\em Combinatorica}, 14(2):217--241, 1994.

\bibitem{ShamirT98them}
R.~Shamir and D.~Tsur.
\newblock The maximum subforest problem: approximation and exact algorithms
  (extended abstract).
\newblock In {\em Proceedings of the {N}inth {A}nnual {ACM}-{SIAM} {S}ymposium
  on {D}iscrete {A}lgorithms ({S}an {F}rancisco, {CA}, 1998)}, pages 394--399,
  New York, 1998. ACM.

\bibitem{ShlomiSRRS06qpa}
T.~Shlomi, D.~Segal, E.~Ruppin, and R.~Sharan.
\newblock Qpath: a method for querying pathways in a protein-protein
  interaction network.
\newblock {\em BMC Bioinformatics}, pages --1--1, 2006.

\bibitem{SlotB85onta}
C.~F. Slot and P.~van Emde~Boas.
\newblock On tape versus core; an application of space efficient perfect hash
  functions to the invariance of space.
\newblock {\em Elektronische Informationsverarbeitung und Kybernetik},
  21(4/5):246--253, 1985.

\bibitem{Tazari10fast}
S.~Tazari.
\newblock Faster approximation schemes and parameterized algorithms on {\it
  }-minor-free and odd-minor-free graphs.
\newblock In {\em MFCS}, pages 641--652, 2010.

\bibitem{Thilikos11fast}
D.~M. Thilikos.
\newblock Fast sub-exponential algorithms and compactness in planar graphs.
\newblock In {\em Proceedings of the 19th Annual European Symposium on
  Algorithms (ESA 2011)}, Lecture Notes in Computer Science. Springer, 2011.

\bibitem{BodlaenderT00cons}
D.~M. Thilikos and H.~L. Bodlaender.
\newblock Constructive linear time algorithms for branchwidth.
\newblock Technical Report UU-CS-2000-38, Dept. of Computer Science, Utrecht
  University, 2000.

\bibitem{ThilikosSB05cutw1}
D.~M. Thilikos, M.~Serna, and H.~L. Bodlaender.
\newblock Cutwidth. {I}. {A} linear time fixed parameter algorithm.
\newblock {\em J. Algorithms}, 56(1):1--24, 2005.

\bibitem{ThilikosSB05cutw2}
D.~M. Thilikos, M.~Serna, and H.~L. Bodlaender.
\newblock Cutwidth. {II}. {A}lgorithms for partial {$w$}-trees of bounded
  degree.
\newblock {\em J. Algorithms}, 56(1):25--49, 2005.

\bibitem{ThilikosSB00cons}
D.~M. Thilikos, M.~J. Serna, and H.~L. Bodlaender.
\newblock Constructive linear time algorithms for small cutwidth and
  carving-width.
\newblock In {\em Algorithms and computation ({T}aipei, 2000)}, volume 1969 of
  {\em Lecture Notes in Comput. Sci.}, pages 192--203. Springer, Berlin, 2000.

\bibitem{Thomasse09aqua}
S.~Thomass\'{e}.
\newblock A quadratic kernel for feedback vertex set.
\newblock In {\em Proceedings of the twentieth Annual ACM-SIAM Symposium on
  Discrete Algorithms}, SODA '09, pages 115--119, Philadelphia, PA, USA, 2009.
  Society for Industrial and Applied Mathematics.

\bibitem{Thurley07kern}
M.~Thurley.
\newblock Kernelizations for parameterized counting problems.
\newblock In {\em Proceedings of the 4th international conference on Theory and
  applications of models of computation}, TAMC'07, pages 705--714, Berlin,
  Heidelberg, 2007. Springer-Verlag.

\bibitem{Leeuwen90grap}
J.~van Leeuwen.
\newblock Graph algorithms.
\newblock In {\em Handbook of Theoretical Computer Science, Volume A:
  Algorithms and Complexity (A)}, pages 525--631. Elsevier Science, 1990.

\bibitem{vanRooijW08para}
I.~van Rooij and T.~Wareham.
\newblock Parameterized complexity in cognitive modeling: Foundations,
  applications and opportunities.
\newblock {\em The Computer Journal}, 51(3):385--404, 2008.

\bibitem{RooijBR09dyna}
J.~M.~M. van Rooij, H.~L. Bodlaender, and P.~Rossmanith.
\newblock Dynamic programming on tree decompositions using generalised fast
  subset convolution.
\newblock In {\em ESA}, pages 566--577, 2009.

\bibitem{Weihe98cove}
K.~Weihe.
\newblock Covering trains by stations or the power of data reduction.
\newblock In R.~Battiti and A.~A. Bertossi, editors, {\em Proceedings of
  Algorithms and Experiments (ALEX98)}, pages 1--8, 1998.

\end{thebibliography}
\end{document}